\documentclass[journal, draftcls, one column]{IEEEtran}

\usepackage{ifpdf}
\usepackage{color}
\usepackage{subfigure}
\usepackage{subfloat}
\usepackage{url}
\usepackage{cite}

\ifCLASSINFOpdf
  \usepackage[pdftex]{graphicx}
\else
\fi
\usepackage{epstopdf}
\usepackage{epsfig}

\usepackage[cmex10]{amsmath}
\usepackage{amsthm}
\usepackage{array}

\usepackage[font=footnotesize]{subfig}
\usepackage{url}

\newtheorem{defn}{Definition}
\newtheorem{obs}{Observation}
\newtheorem{thm}{Theorem}
\newtheorem{lem}{Lemma}
\newtheorem{rmk}{Remark}
\newtheorem{cor}{Corollary}
\newtheorem{example}{Example}
\newtheorem{assum}{Assumption}
\begin{document}

\title{Quality Sensitive Price Competition in Secondary Market Spectrum Oligopoly- Multiple Locations}

\author{Arnob~Ghosh and
        Saswati~Sarkar
\thanks{The authors are with the Department
of Electrical and Systems Engineering, University Of Pennsylvania, Philadelphia,
PA, USA. Their E-mail ids are arnob@seas.upenn.edu and swati@seas.upenn.edu.} 
\thanks{Parts of this paper have been presented in CISS\rq{}14 \cite{ciss}.}}




\maketitle
\begin{abstract}
We investigate a spectrum oligopoly market where each primary seeks to sell secondary access to its channel at multiple locations. Transmission qualities of a channel evolve randomly. Each primary needs to select a price and a set of non-interfering locations (which is an independent set in the conflict graph of the region) at which to offer its channel without knowing the transmission qualities of the channels of its competitors. At each location each secondary selects a channel depending on the price and the quality of the channels. We formulate the above problem as a non-cooperative game. We consider two scenarios-i) when the region is small, ii) when the region is large. In the first setting, we focus on a class of conflict graphs, known as mean valid graphs which commonly arise when the region is small. We explicitly compute a symmetric Nash equilibrium (NE) that selects only a small number of independent sets with positive probability. The NE is threshold type in that primaries only choose independent set whose cardinality is greater than a certain threshold. The threshold on the cardinality increases with increase in quality of the channel on sale.  We show that the symmetric NE strategy profile is unique in a special class of conflict graphs (linear graph) which commonly arises in practice. In the second setting, we consider node symmetric conflict graphs which arises  when the number of locations is large (potentially, infinite). We explicitly compute a symmetric NE that randomizes equally among the maximum independent sets at a given channel state vector. In the NE a primary only selects the maximum independent set at a given channel state vector. We show that the two symmetric NEs computed in two settings exhibit important structural difference. We numerically evaluate the ratio of the expected payoff attained by primaries in the  game and the payoff attained by primaries when all the primaries collude.
\end{abstract}
\vspace{-0.3cm}
\section{Introduction}
\subsection{Motivation}
Secondary access of the spectrum where license holders (primaries) allow unlicensed users (secondaries) to use their channels can enhance the efficiency of the spectrum usage. However, secondary access will only proliferate when it is rendered profitable to the primaries. We investigate a spectrum oligopoly where  primaries lease their spectrum to secondaries in lieu of financial remuneration.  Each primary owns a channel throughout a large region consisting of several {\em locations}. The channel of a primary provides a transmission rate to a secondary depending on the state which evolves randomly and reflects the usage of the primary as well as the transmission rate due to fading. We consider the state of a channel is $0, 1\ldots,$ or $ n$ where higher state corresponds to higher transmission rate. A secondary receives a payoff from a channel depending on the transmission rate offered by the channel and the price quoted by the primary.  Secondaries buy those channels which give them the highest payoff, which leads to a {\em competition} among primaries.

 Price competition in  economics and wireless setting  ignore two important properties which distinguish spectrum
oligopoly from standard oligopolies:   First, a primary selects a price knowing only the state of its own channel; it is unaware of states of its competitors\rq{} channels. Thus, if a primary quotes a high price, it will earn a large profit if it sells its channel, but it may not be able to sell at all; on the other hand a low price will enhance the probability of a sale but may also fetch lower profits in the event of a sale. Second, the same spectrum band can be utilized simultaneously at geographically dispersed locations without interference; but the same band can not be utilized simultaneously at interfering locations. This special feature known as {\em spatial reuse} adds another dimension in the strategic interaction as now a primary has to cull a set of non-interfering locations, which is denoted as an {\em independent set} in the conflict graph representation of the region\cite{graph}; at which to offer its channel apart from selecting a price at every node of that set. Intuitively, a primary would like to make its channel available at an independent set of the maximum size (cardinality). However, if the competition at the largest independent set is intense, a primary may achieve higher payoff by setting high price at small independent sets (where the competition is not so intense). 

\subsection{Our Contributions}
We devise the problem as a game in which each primary\rq{}s strategy space consists of independent set selection strategy and the pricing strategy at each node of the independent set when the channel is available for sale. When the channel is in state $0$, the transmission rate is very low and thus, we consider the channel is not available for sale. 
We first show that there may exist multiple asymmetric NEs.  Asymmetric NEs are difficult to implement in the symmetric game that we consider  (Section~\ref{sec:solution_concept}). We, therefore, focus only on finding symmetric NEs subsequently.  We prove a {\em separation theorem} (Section~\ref{sec:separation}) which entails that the NE pricing strategy at each location can be uniquely computed if the independent set selection strategy is known.  By virtue of  our previous work \cite{isit,arnob_ton} which characterizes pricing strategies of primaries for different transmission rates when the region has only one location (i.e. no spatial reuse). We then focus only on the independent set selection strategy.  


  {\em Scenario 1}: We consider two possible scenarios (Section~\ref{sec:two settings}). First, we consider the setting when the region is small consisting of few locations (Section~\ref{sec:same_channel_state}).  Therefore,  the usage statistics and the propagation condition of a channel do not vary substantially over the region. Thus, we assume that the channel state  is identical at each location in this setting. In the initial stages of deployment of the secondary market, it is expected that the secondary market will be introduced in small regions consisting of a few locations. Hence, the price competition in this setting reduces to a price selection problem where the transmission quality of each primary remains the same throughout the region. 
  
  In this setting, we focus on a particular class of graphs, introduced as  {\em mean valid graph}\cite{gauravjsac} since most of the small graphs observed in practice are mean valid graphs (Section~\ref{sec:meanvalidgraph}). In a mean valid graph, nodes can be partitioned in $d$ disjoint maximal independent sets  namely $I_1,\ldots, I_d$\cite{gauravjsac}.  But the total number of independent sets in such a graph may be substantially large; generally, the number of independent sets grows exponentially with the number of nodes. We show that  there exists a symmetric NE strategy which selects independent sets only amongst $I_1,\ldots,I_d$ which characterize the mean valid graph (Section~\ref{sec:existencemultiplenodes}); we explicitly compute the strategy (Section~\ref{sec:structuremultiplenodes}).  Such a strategy profile can be stored using a $d$ dimensional vector. Thus, the space required to store strategy profile scales with $d$ rather than increasing exponentially with nodes.   Primaries also need to know only $I_1,\ldots, I_d$ rather than the entire graph in order to compute a symmetric NE. 
 
The characterization of the symmetric NE strategy profile reveals that a primary only selects an independent set whose cardinality is greater than or equal to a certain threshold (Section~\ref{sec:structuremultiplenodes}). This threshold turns out to be a non-decreasing function of channel quality (Section~\ref{sec:properties_threshold}). Thus, when the channel quality is high, a primary restricts itself only to independent sets of large cardinalities; when the channel quality is poor, the primary diversifies among independent sets of different sizes.   We show using an example that arises in practice that primaries only offer their poor quality channels at independent sets of lower cardinalities (Section~\ref{sec:properties_threshold}). Thus, a social planner may have to provide some incentives to primaries so as to ensure that users of those locations can get access to higher quality channels.

Next, we examine the uniqueness among symmetric NE strategy profiles in mean valid graphs (Section~\ref{sec:symmetricNEunique}). Nodes in such a graph can be partitioned into different collections of  maximal independent sets (Fig.~\ref{fig:different partitions}). A primary in general would not know the partition other primaries are selecting. Our result reveals that each such partition leads to a unique symmetric NE; yet primaries need not co-ordinate with each other regarding the partition one is selecting (Theorem~\ref{thm:policyuniqueness}). Hence the symmetric NE strategy profile is easy to implement.  Theorem~\ref{thm:policyuniqueness} also reveals that all these symmetric NEs lead to the same node selection probabilities. The NE pricing strategy at a node depends only on the probability with which it is selected. Thus, all these symmetric NEs are functionally unique.   Finally, we focus on a special class of mean valid graphs known as {\em linear graphs} (Figure~\ref{fig:linear})  which frequently arises in practice such as in the modeling of communication nodes over a highway or a row of shops. We prove that the symmetric NE strategy is unique (is not merely functionally unique) in linear graphs (Theorem~\ref{uniquelinear}). 

{\em Scenario 2}: We subsequently consider the  scenario when the secondary spectrum market is operated on a large region consisting of several locations.  In this setting the transmission quality of a channel may be different at different locations in the region. Thus, a primary needs to specify a strategy for each possible channel state across the network (Section~\ref{sec:diff_channel_states}). The number of channel states and thus, the strategy space increases exponentially with number of nodes. The conflict graph representation of the region depends on the channel state across each location since a primary must select an independent set of nodes only among those nodes where the channel is available for sale.  A primary is not aware of the conflict graph from which other primaries are selecting their independent sets let alone their channel states.  The characterization of a symmetric NE strategy profile in the above setting is thus, more challenging. We simplify the model by assuming that the channel is either available or not (i.e. $n=1$), but the availability can differ across the nodes. 

We focus on node symmetric or node transitive graphs (Section~\ref{sec:nodesymmetricgraphs}) \cite{nodetransitive} such as finite cyclic graph, infinite lattice graphs (e.g. infinite linear graph (infinite in both directions), infinite square graph, infinite grid graph, infinite triangular graphs)\cite{lattice} which arise in practice when the region becomes large. We allow some statistical correlations which arise naturally among the channel states at different locations (Section~\ref{sec:channel_statistic}). We show that there exists a symmetric NE strategy profile ($\mathrm{SP_{sym}}$ ) for those graphs (Theorem~\ref{thm:nodesymmetricne}). In the symmetric NE strategy profile, a primary randomizes uniformly among the maximum independent sets (the independent set of the highest cardinality). A primary thus only need to enumerate the maximum independent sets in order to determine $\mathrm{SP_{sym}}$. In contrast to the setting where the channel state remains the same through the network, in $\mathrm{SP_{sym}}$ the channel is offered at every node with equal probability. We also show that  $\mathrm{SP_{sym}}$ may not be an NE in a finite linear graph which {\em is not a node symmetric graph}. We show that the symmetric NE may not be unique for a linear graph unlike the setting where the channel state remains the same throughout the network (Lemma~\ref{thm:notsame}). 

In $\mathrm{SP_{sym}}$ each primary needs to enumerate the maximum independent sets. The number of independent sets grow exponentially with the nodes. However,  at a given channel state vector over the region, the conflict graph may consist of several components. A primary can find maximum independent sets and $SP_{sym}$  in each component in parallel. However, the number of maximum independent sets in a component grows exponentially with the number of nodes in the component.   We, thus, investigate the size of the expected component size both analytically and empirically (Section~\ref{sec:computation}).  Empirical result shows that the average size of components is often moderate and the upper bound computed analytically is often loose.  However, the component size can be substantially large when the channel availability probability is large. In order to control the component size we, thus, consider the setting where each primary decides to estimate the channel state at a node with a certain probability ($p$). A primary then sells its channel at nodes only amongst the nodes where it estimates the channel.  We show that $\mathrm{SP_{sym}}$ is a NE strategy in this setting as well.  However, if $p$ is small, then a primary can only sell its channel at few locations which will potentially reduce the payoff. A primary thus needs to select $p$ judiciously in order to attain a required trade-off between the computation cost and the expected payoff.

Finally, we numerically compare the expected profit obtained by the primaries using our NE strategy profile in both of the settings to  the maximum possible profit allowing for collusion among primaries (Section~\ref{sec:numerical}). The proofs do not follow from the standard game theory results. The proofs rely on the specific properties of the conflict graphs, and the game under consideration. {\em Thus, both the results and the proofs  are the central contributions of this paper.}

\subsection{Related Literature}
Price selection in oligopolies has been extensively investigated in economics as a non co-operative Bertrand Game \cite{mwg} and its modifications \cite{Osborne, Kreps}. Price competition among wireless service providers have also been explored to a great extent \cite{Ileri, Mailespectrumsharing, Mailepricecompslotted, Xing, Niyatospeccrn, Niyatomultipleseller,Zhou,kavurmacioglu,yitan,duan,zhang,jia,yang,sengupta,kim}. But all these papers did not consider the uncertainty of competition and the spatial reuse property of the spectrum oligopoly.

We now distinguish our contributions compared to \cite{gauravjsac} which is the closest to our work. First, \cite{gauravjsac} considered that the channel state remains the same throughout the region and the state of the channel can be either $0$ (not available for sale) or $1$ (available); this assumption does not capture the different transmission qualities offered by the available channels. When we consider that the channel state remains the same throughout the network we consider that the available channel can be in one of the $n$ states depending on the transmission qualities. Thus, in our setting a primary now needs to employ different pricing strategies and different independent set selection strategies for different channel states while in the former case a single pricing and independent set selection strategy would suffice as the price need not be quoted for an unavailable commodity. Second, we also consider the setting where the channel state need not be the same unlike in \cite{gauravjsac}.  In our setting a primary does not know the  conflict graph of other primaries from which they will select their independent sets.   Thus, the collection of independent sets from which a primary selects an independent set may be different for different primaries at a given time slot since the channel state vector may be different for different primaries. Whereas in \cite{gauravjsac} the channel is either available at all locations or unavailable at any location. Thus in \cite{gauravjsac}, a primary knows the conflict graph from which other primaries will select their independent sets when their channels are available.  Thus, the characterization of an NE becomes significantly challenging in our setting compare to \cite{gauravjsac}. The result we obtain also significantly differs from \cite{gauravjsac}. For example, in \cite{gauravjsac} a primary can select an independent set of lower cardinalities, however, in our setting, a primary only selects the maximum independent set. Additionally, the symmetric NE is unique in a finite linear graph in \cite{gauravjsac}, whereas there are infinitely many symmetric NEs in our setting. 





\vspace{-0.3cm}

\section{System Model} \label{sec:model}
Each primary owns a channel over a region. 
Unless otherwise stated, we consider that there are $l$ number of primaries and $m$ number of secondaries at each location throughout this paper. We, however, generalize our result for random {\em apriori} unknown $m$  in Section~\ref{sec:random_demand}. Different channels constitute disjoint frequency bands. Each primary only allows at most secondary to transmit at a given location.
\subsection{Transmission Rate and Channel State}\label{sec:transmissionrate}
  The channel of a primary provides a certain transmission rate at a location to a secondary who is granted access. Transmission rate  (i.e. Shannon Capacity) at a location depends on--  1) the number of subscribers of a primary that are using the channel at that location\footnote{Shannon Capacity \cite{cover} for user $i$ at a channel is equal to $\log\left(1+\dfrac{p_{i}h_i}{\sum_{j\neq i}p_jh_j+\sigma^2}\right)$ where $p_k$ is the transmitted power of user $k$, $\sigma^2$ is the power of white noise, $h_k$ is the channel gain between transmitter and receiver which depends on the propagation condition. If a secondary is using the channel then $p_i, h_i$ of the numerator are the attributes associated with the secondary while $p_j, h_j j\neq i$ are those of the subscribers of the primaries. In general, the power $p_j$ for subscriber of primaries is constant for subscriber $j$ of primary, but the number of subscribers vary randomly over time. The power $p_i$ with which a secondary will transmit may be a constant or may decrease with the number of subscribers of primaries in order to limit the interference caused to each subscriber. The above factors contributes to the random fluctuation in the capacity of a channel offered to a secondary.} and 2) the propagation condition of the radio signal\cite{arnob_ton}. The transmission rate at a location evolves randomly over time owing to the randomness of the usage of subscribers of primaries and the propagation condition\footnote{Referring to footnote 1, $h_k$ and $\sigma^2$ evolve randomly owing to the random scattering of the particles in the atmosphere; this phenomenon is also known as {\em fading}\cite{proakis}.}.  We discretize the transmission rate into a number of states $0, 1, \ldots, n $. State $i$ provides a lower transmission rate to a secondary than state $j$ if $i < j$ and state\footnote{Generally a minimum transmission rate is required to send data. State $0$ indicates that the transmission rate is below that threshold due to either the excessive usage of subscribers of primaries or the transmission condition.} $0$ arises when the secondary can not use the channel making the channel unavailable for sale.

 Let $J$ denote the channel state vector which indicates the channel state at each node. For example, when the number of nodes are 3, then $J=(1,1,0)$ is a channel state vector which indicates that the channel is in state $1$, $1$, and $0$ at nodes $1, 2,$ and $3$ respectively. We assume that the channels are statistically identical, specifically the probability that the channel state vector of a primary is $J$ is $q_J$.  We also assume that the probability of the event where the channel state is $0$ at every location is non-zero i.e.
 \begin{align}\label{prob1}
 q_J>0 \quad \text{when } J=\{0,0,\ldots,0\}
 \end{align}
\subsection{Penalty functions}\label{sec:penaltyfunction}
Secondaries are passive entities. At a given location they select channels considering the price and the transmission rate offered by the channel.  We assume that the preference of secondaries can be modeled by a penalty function. If a primary selects a price $p$ at channel state $i$ at a given location, then the channel incurs a penalty $g_i(p)$ for all secondaries at that location. As the name suggests, a secondary prefers a channel with a lower penalty. Since lower prices should induce lower penalty, thus, we assume that each $g_i(\cdot)$ is strictly increasing; therefore, $g_i(\cdot)$ is invertible. For a given price, a channel of higher transmission rate must induce lower penalty, thus, $g_i(p)<g_j(p)$ if $i>j$.  No secondary will buy any channel whose penalty exceeds $v$.  Secondaries have the same penalty function and the same upper bound for penalty value ($v$), thus, secondaries are statistically identical\cite{isit, arnob_ton}.

We denote $f_i(\cdot)$ as the inverse of $g_i(\cdot)$. Thus, $f_i(x)$ denotes the price when the penalty is $x$ at channel state $i$. We assume that $g_i(\cdot)$ is continuous, thus $f_i(\cdot)$ is continuous and strictly increasing.  Also, $f_i(x)<f_j(x)$ for each $x$ and $i<j$. 

We focus on  penalty functions  of the form $g_i(p)=h_1(p)-h_2(i)$, where $h_1(\cdot)$ and $h_2(\cdot)$ are strictly increasing in their arguments. Note that $-g_i(p)$ may be considered as the utility that a secondary gets at channel state $i$ and price $p$. Since utility functions are generally assumed to concave, thus, we consider $h_1(\cdot)$ is convex.  
We show in \cite{arnob_ton} that when $h_1(\cdot)$ is convex, then penalty functions $g_i(p)=h_1(p)-h_2(i)$ satisfy the following property:\\
\begin{assum}\label{assum1}
\begin{align}\label{con1}
    \dfrac{f_i(y)-c}{f_j(y)-c}<\dfrac{f_i(x)-c}{f_j(x)-c}  \ \text{ for all } x>y > g_i(c), i<j.
    \end{align} 
    \end{assum}
    Moreover, we also show in\cite{isit, arnob_ton} that when $g_i(p)=h_1(p)/h_2(i)$, then, the inequality in (\ref{con1}) is satisfied for some certain convex functions $h_1(\cdot)$ like $h_1(p)=p^r (r\geq 1),\exp(p)$.  In addition, there is also a large set of functions that satisfy (\ref{con1}), such as: $g_i(p) = \zeta\left(p -h_2(i)\right), g_i(p) = \zeta\left(p/h_2(i)\right)$ where
 $\zeta(\cdot)$ is continuous and strictly increasing. Moreover, Assumption~\ref{assum1} is satisfied by penalty functions
 $g_i(\cdot)$ whose  inverses are of the form $f_i(x)=h(x)+h_2(i), f_i(x) = h(x)*h_2(i)$, where $h(\cdot)$ is {\em any} strictly increasing function. In this setting,  we consider penalty functions which satisfy Assumption 1.

In the special class, when $n=1$ i.e. the channel is either available or not, then the available channels offer the same transmission rates. Hence, we do not need the penalty functions to capture the preference order of secondaries for available channels having different transmission rates. Thus, the penalty functions are redundant when $n=1$. But to be consistent with the notations, we still use the penalty function $g_1(\cdot)$ and the inverse penalty function $f_1(\cdot)$ when $n=1$. We do not need Assumption 1 when $n=1$ and  we only assume that  penalty function $g_1(\cdot)$ is strictly increasing.
 
  \subsection{Conflict Graph}\label{sec:conflict_graph}
Each primary owns a channel over a broad region consisting of several {\em locations}. Typically, secondary users can not transmit simultaneously using the same channel at adjacent locations  due to interference. In order to sell its channel a primary needs to find a set of locations which do not interfere with each other. Wireless networks have been traditionally modeled as {\em conflict graphs} (Figures ~\ref{fig:linear},~\ref{fig:grid_region},~\ref{fig:grid}) in most of the existing literature including in several seminal papers \cite{gupta,shroff, srikant}. Let $G=(V,E)$ be the overall conflict graph of the region where $V$ is the set of nodes and $E$ is the set of edges; an edge exists between two nodes iff transmission at the corresponding locations interfere. In a conflict graph, the set of nodes in which no edge exists between any pair of nodes is called an {\em independent set} (Fig.~\ref{fig:linear},~\ref{fig:grid}). Thus, secondaries at all nodes in an independent set, can transmit simultaneously using the same channel without any interference.

Note that when the channel of a primary is at state $0$ at a node, then the primary can not sell its channel at that node. Thus, a primary ought to offer its channel at a set of non interfering locations among the locations where the channel is available for sale (i.e. the state of the channel is not $0$).  Let $G_{J}=(V_{J},E_{J})$ be the conflict graph representation of the channel state vector $J$: $V_{J}$ is the set of nodes (locations) where the channel is available for sale at channel state vector $J$ of a primary and $E_{J}$ is the set of edges in $G$ between the nodes of $V_J$.   $G_{J}$ is obtained by removing nodes and the edges corresponding to those nodes from $G$ where the channel is not available i.e. the channel is at state $0$.  Thus, $G_{J}$ is a subgraph of $G$. Figure~\ref{fig:reducedconflict} represents a conflict graph $G$ of a region and the conflict graph $G_{J}$ when the channel state vector is $J$. A primary needs to select an independent set from $G_{J}$ when the channel state vector is $J$.
\begin{figure}
\begin{subfigure}[]
{
\includegraphics[width=70mm,height=40mm]{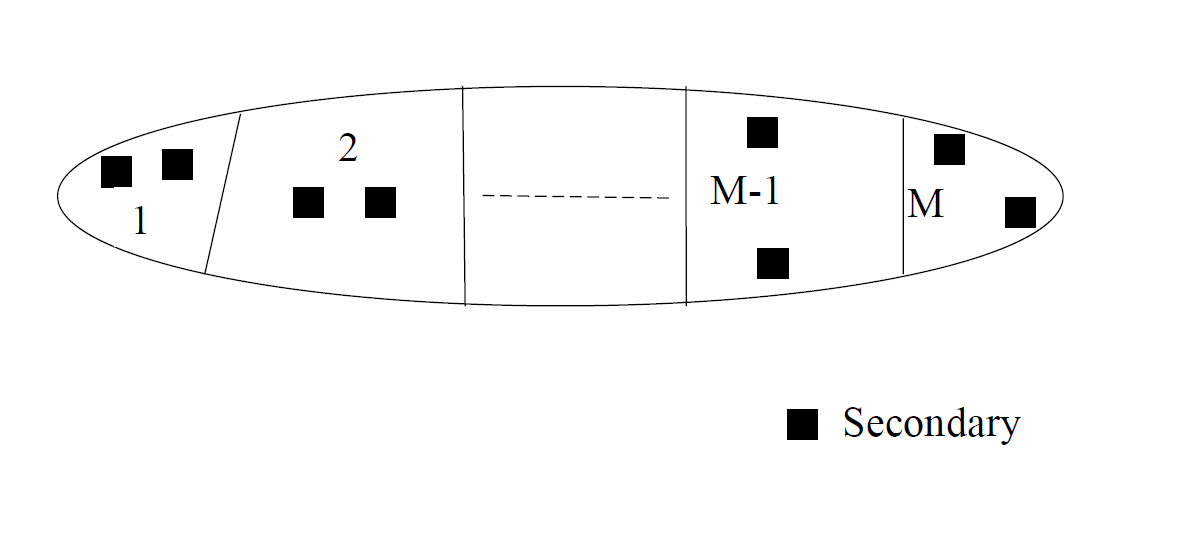}}
\end{subfigure}
\label{fig:linear_region}
\begin{subfigure}[]
{\includegraphics[width=90mm,height=30mm]{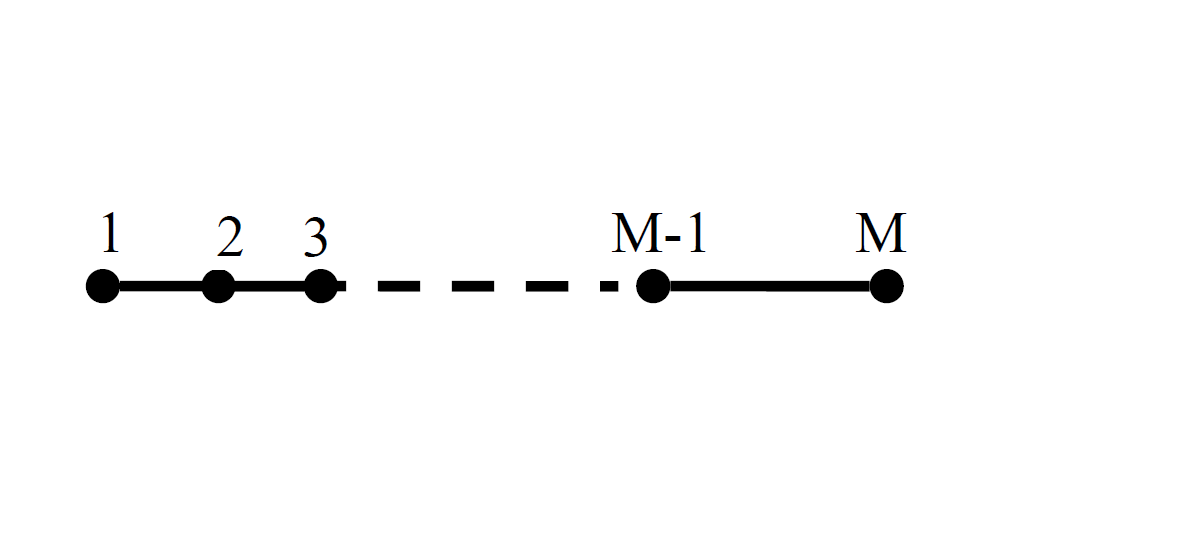}
\vspace{-0.7cm}}
\end{subfigure}

\caption{\small Figure in (a) shows a wireless network with $M$ number of locations. There are $m=2$ secondaries at each location. Signals at locations $1$ and $2$ and $2$ and $3$ interfere with each other, but signals at locations $1$ and $3$ do not interfere. Linear Graph in figure (b) models the conflict graph of the network in (a). Note that there is an edge between nodes $1$ and $2$, but not between nodes $1$ and $3$. $I_1=\{1,3,5,\ldots,M_o\}$ and $I_2=\{2,4,\ldots,M_e\}$ constitute independent sets, where $M_o$ ($M_e$, respectively) is the greatest odd (even, respectively) less than or equal to $M$. There are other independent sets too e.g. \{1,4,6\}. Also \{1,2,4\} is not an independent set since there is an edge between nodes 1 and 2.}
\label{fig:linear}
\vspace*{-0.5cm}
\end{figure}
\begin{figure}
 \includegraphics[width=90mm,height=30mm]{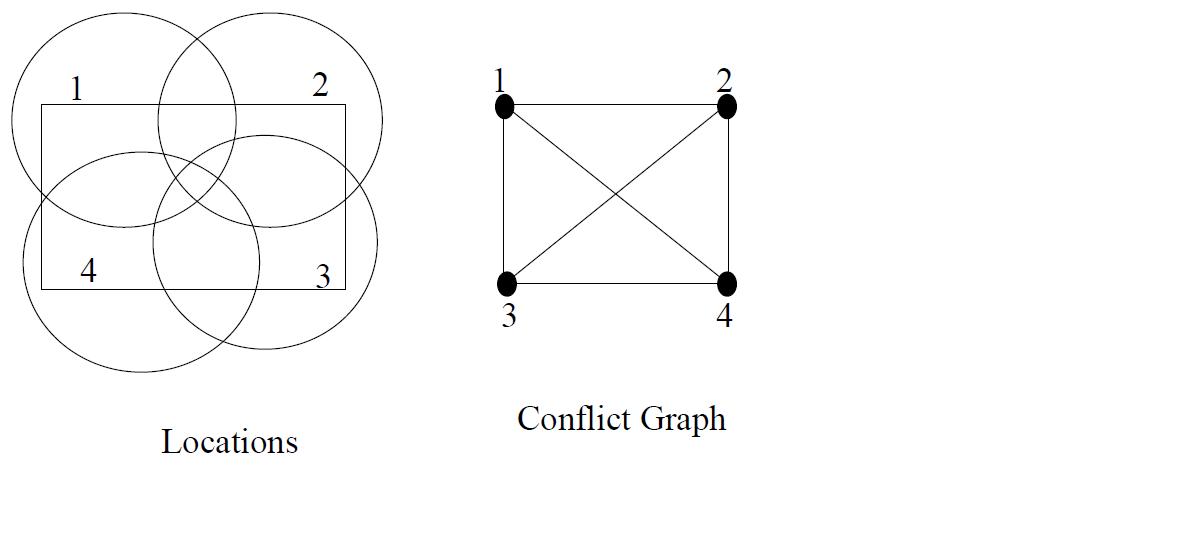}
 \vspace{-0.3cm}
 \caption{\small The rectangle represents a shop in a shopping complex or a department in a university campus. Circles $1,2,3,4$ are the ranges of Wireless access points. Each circle corresponds to a node in the conflict graph.  Since ranges of Wireless access points intersect with each other, thus there exists an edge between every pair of nodes.}
 \label{fig:grid_region}
 \vspace{-0.5cm}
 \end{figure}
 
 \begin{figure}
 \includegraphics[width=150mm,height=50mm]{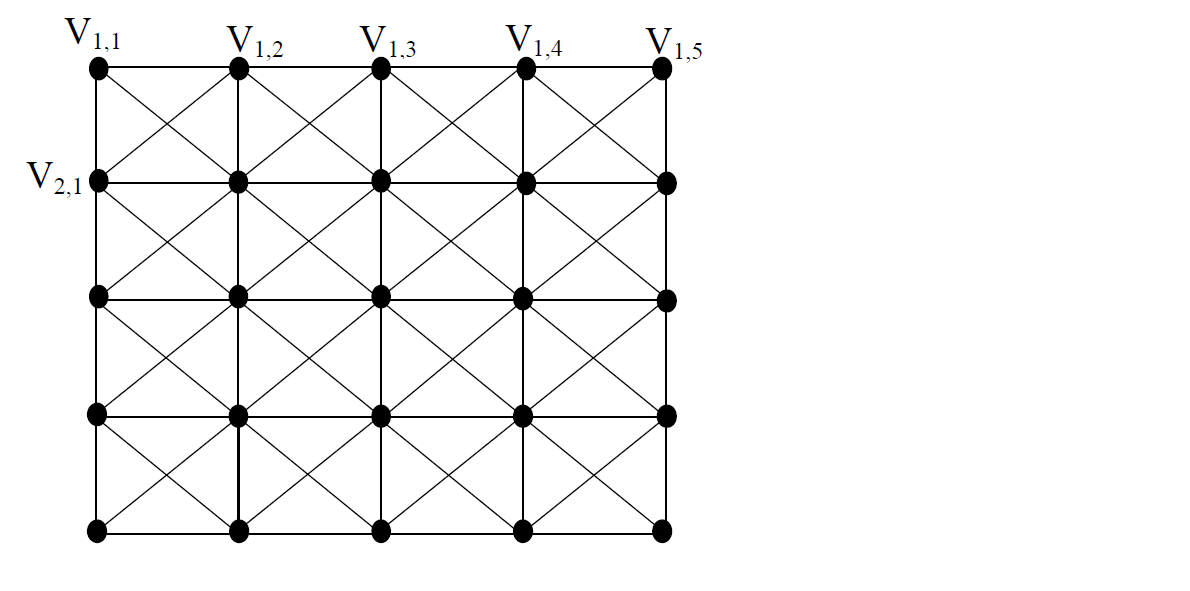}
 \caption{\small The above graph is the conflict graph representation of a larger region consisting of several networks depicted in Fig.~\ref{fig:grid_region}. It is a grid conflict graph with $k$ rows and columns (here $k=5$). Nodes correspond to the Wireless access points. $\{V_{1,1},V_{1,3}\}$ is an independent set and users at these two nodes can transmit simultaneously. But $\{V_{1,1},V_{1,2}\}$ or $\{V_{1,1},V_{2,1}\}$ are not independent sets.}
 \label{fig:grid}
 \vspace*{-0.5cm}
 \end{figure}   
 \begin{figure}
\includegraphics[width=120mm,height=50mm]{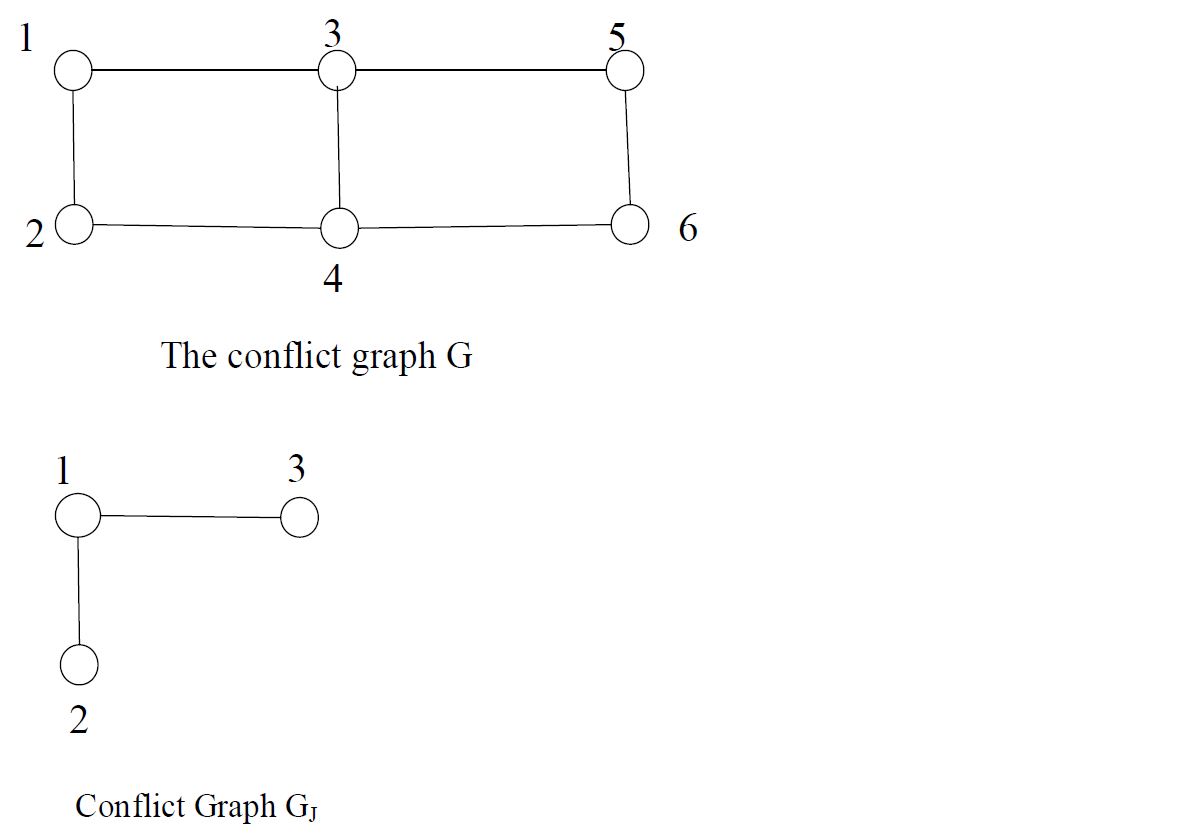}
\caption{\small The conflict graph for the overall region is $G$ which corresponds to the situation where the channel is available at all nodes in the region, $G_J$ is the conflict graph when the channel state vector is $J=(j_1,j_2,j_3,0,0,0)$ where $j_i\geq 1, i=1,2,3$. Since the channel states are $0$ at nodes $4, 5, $ and $6$, thus, $G_{J}$ is obtained by removing those nodes and the edges corresponding to those nodes.}
\label{fig:reducedconflict}\vspace{-0.5cm}
\end{figure}

\subsection{Strategy and Payoff of Each Primary}\label{sec:strategy}
Let $\mathcal{P}$ denote the set of all possible channel state vectors except when the channel state is $0$ across all the locations. Note that $|\mathcal{P}|=(n+1)^{|V|}-1$.

For each channel state vector $J\in \mathcal{P}$  a primary selects\footnote{A primary does not need to select a strategy when the channel state is $0$ at all locations.}:
a) an independent set of the conflict graph $G_J$  where it will sell its channel;
b) a price at every node of that independent set. 
 A primary arrives at its decision with the knowledge of its own channel state vector $q_J$ but without knowing the {\em channel state vector of other primaries}.    A primary however knows $l, m, n, G, g_1,\ldots, g_n, f_1,\ldots, f_n, $ and $ q_J, J\in \mathcal{P}$. Secondaries strictly prefer a channel which induces lower penalty compared to the higher penalty one  as discussed in Section~\ref{sec:penaltyfunction}.  Since there is a one-to-one correspondence between the price and the penalty at a given channel state, thus, for the ease of analysis we consider that primaries select penalties instead of prices. The ties among channels with identical penalties are broken randomly and symmetrically
   among the primaries. We formulate the decision problem of primaries as a non-cooperative game with primaries as players.
\begin{defn}\label{defn:strategy}
A strategy of a primary $i$ $\psi_{i,J}$ provides the probability mass function (p.m.f) for selection among the independent sets (I.S.s) and the penalty distribution it uses at each node, when its channel state vector is $J$ . $S_i=(\psi_{i,1},....,\psi_{i,|\mathcal{P}|})$ denotes the strategy of primary $i$, and $(S_1,...,S_l)$ denotes the strategy profile of all primaries (players).
$S_{-i}$ denotes the strategy profile of primaries other than $i.$\\
\end{defn}

Each primary incurs a transition cost $c$ at each location where it is able to sell its channel. If primary $i$ selects a penalty $x$ at node $s$ when the channel state is $j$, then its  payoff at node $s$ is\footnote{Note that if $Y_s$ is the number of channels offered for sale at a node $s$, for which the penalties are upper bounded by $v$, then
those with $\min(Y_s, m)$ lowest penalties are sold  since secondaries select channels in the increasing order of penalties.}  \\
\begin{equation*}
\begin{cases} f_j(x)-c  & \text{if the primary sells its channel}\\
0 & \text{otherwise.}
\end{cases}
\end{equation*}
  The payoff of a primary over an independent set is the sum of payoff that it gets at each node of that independent set. Thus, if a primary is unable to sell at any node of an independent set, then its payoff is $0$ over that independent set.
\begin{defn}\label{defn:expectedpayoff}
$u_{i,J}(\psi_{i,J},S_{-i})$ is the expected payoff when primary $i$\rq{}s channel state vector is $J$ and selects strategy $\psi_{i,J}(\cdot)$ and other primaries use strategy $S_{-i}$.
\end{defn}

\subsection{Solution Concept}\label{sec:solution_concept}
We seek to obtain a Nash Equilibrium (NE) strategy profile which we define below using $u_{i,J}$ (Definition~\ref{defn:expectedpayoff}), $\psi_{i,J}$ and $S_{-i}$ (Definition~\ref{defn:strategy}):  
\begin{defn}\label{defn:ne}\cite{mwg}
A \emph{Nash  equilibrium}  $(S_1, \ldots, S_l)$ is a strategy profile such that no primary can improve its expected profit by unilaterally deviating from its strategy.  So, with $S_i=(\psi_{i,1},....,\psi_{i,|\mathcal{P}|})$, $(S_1, \ldots, S_l)$, is  a  Nash equilibrium (NE) if for each primary $i$ and channel state vector $J$
\begin{align}
u_{i,J}(\psi_{i,J},S_{-i})\geq u_{i,J}(\tilde{\psi}_{i,J},S_{-i}) \ \forall \ \tilde{\psi}_{i, J}.
\end{align}
An NE $(S_1, \ldots, S_l)$  is a \emph{symmetric NE} if $S_i = S_j$ for all $i, j.$
\end{defn}
If $S_i\neq S_k$ for some $i,k\in \{1,\ldots,l\}$ in an NE strategy profile, then the strategy profile is an asymmetric NE.

In a symmetric game, as the one we consider, it is difficult to implement an asymmetric NE. For example, if there are two players and $(S_1,S_2)$ is an asymmetric NE i.e.  $S_1\neq S_2$, then $(S_2,S_1)$ is also an NE due to the symmetry of the game. The realization of such an NE is only possible when one player knows whether the other  is using $S_1$ or $S_2$. But,  apriori coordination among players is infeasible as the game is non co-operative.  


Note that if $m\geq l$, then primaries select the highest penalty $v$ at each node and will select one of the maximum independent sets of $G_{J}$ at channel state vector $J$ with probability $1$. This is because, when $m\geq l$, then, the channel of a primary will always be sold at a location. Hence, a primary will be always be able to sell its channel at the highest possible penalty. Henceforth, we will consider that $m<l$.

\subsection{Two Different Settings}\label{sec:two settings}
We consider two different settings: i) First, we consider that the region is small and consists of a few (but, multiple) locations (Section~\ref{sec:same_channel_state}).  Initially,  it is expected that the secondary market will be introduced in a small region consisting of few locations. In a small region, the usage statistic and the propagation condition of a channel will be similar at each location, thus, in an analytical abstraction we consider that the transmission rate offered by a channel is the same at each location.  In this setting, the interference relations amongst the locations may not be symmetric in general which we accommodate in our model. Since we only consider that the channel state is the same across the nodes, thus, $q_J=0$ for all those channel state vectors where the channel state is not identical at each location.

ii) Second, we consider the region consists of large number of locations (Section~\ref{sec:diff_channel_states}).  This is likely to happen in later stages of deployment of the secondary market. Since the geographical region is large, the transmission rate offered by a channel at different locations may be different. Thus, we consider that the channel state of a primary can be different at different locations. Given the large region, there will be an inherent symmetry in the interference relations among the locations which we characterize and exploit. 



\vspace{-0.2cm}


\section{Initial Results, Multiple NEs, and A Separation Result}
\subsection{Results Of One-shot Single Location Game}\label{sloneshot}
Now, we briefly summarize the main results of the game when it is limited to only one location, which we have studied in \cite{isit, arnob_ton}. 
 Since there is only node, thus, the channel state vector reduces to a scalar and  we denote $q_j$ when the channel is in state $j\in \{0,\ldots,n\}$ at that node. Note that there is no spatial reuse constraint in this setting, thus a primary\rq{}s decision is only to select a penalty. 

We start with following definitions. Let $w(x)$ be the probability of at least m successes out of $l-1$ independent Bernoulli trials, each of which occurs with probability $x.$ Thus,
 \begin{eqnarray}
w(x)&=&\sum_{i=m}^{l-1}\dbinom{l-1}{i}x^i(1-x)^{l-i-1}.\label{d4}
\end{eqnarray}
Note that $w(\cdot)$ is continuous and strictly increasing in $[0,1]$, so its inverse exists. 

Now, let
 for $1 \leq j \leq n$,
\begin{eqnarray}
L_0& =& U_1=v,\nonumber\\
 p_j-c& =& (f_j(U_j)-c)(1-w(\sum_{k=j}^{n}q_k))
\label{n51}\\
\mbox{and } L_j&=&g_j(\dfrac{p_j-c}{1-w(\sum_{k=j+1}^{n}q_k)}+c), U_j=L_{j-1}\label{n52}
\end{eqnarray}
Since $U_1=v$, thus we obtain $p_{j}, L_{j}$ (which in turn gives $U_{j+1}$) recursively starting from $j=1$ using \eqref{n51} and \eqref{n52}. Note that $v>L_1>\ldots>L_n$ and $f_j(L_j)>c$ \cite{isit}. We have shown
\begin{lem} \cite{isit,arnob_ton}\label{lm:computation}
A NE strategy profile  $\left(\phi_1(\cdot), \ldots, \phi_n(\cdot)\right)$ must
comprise of: \begin{align}\label{c5}
\phi_j(x)= & 
 0 ,  \text{if} \ x<L_j\nonumber\\
& \dfrac{1}{q_j}(w^{-1}(\dfrac{f_j(x)-p_j}{f_j(x)-c})-\sum_{k=j+1}^{n}q_k), \text{if} \ L_{j-1}\geq x\geq L_j\nonumber\\
& 1,  \text{if} \ x>L_{j-1}.
\end{align}
\end{lem}
At channel state $j$ a primary selects a penalty using $\phi_j(\cdot)$. Note that $\phi_j(\cdot)$ not only depends on $q_j, j>1$ but also depends on $q_i, i\leq j$.  The support of $\phi_j(\cdot)$ is the closed interval $[L_{j},L_{j-1}]$ $j\in \{1,\ldots,n\}$. $L_{j-1}$ or $U_j$ is the upper endpoint of the support of $\phi_j(\cdot)$. $\phi_j(\cdot)$ is strictly increasing from $L_j$ to $U_j$ and there is no \lq\lq{}gap\rq\rq{} between the support sets of $\phi_j(\cdot), j=1,\ldots,n$ \cite{isit}. Fig.~\ref{fig:dist} illustrates $L_j$s and $U_j$s in an example scenario. Since a secondary always prefers a channel of the lower quality thus, Lemma~\ref{lm:computation} entails that primaries select prices to render the channel of the highest quality as more preferable to the secondaries at a location. 
\begin{figure*}
\begin{minipage}{.32\linewidth}
\begin{center}
\includegraphics[width=65mm, height=40mm]{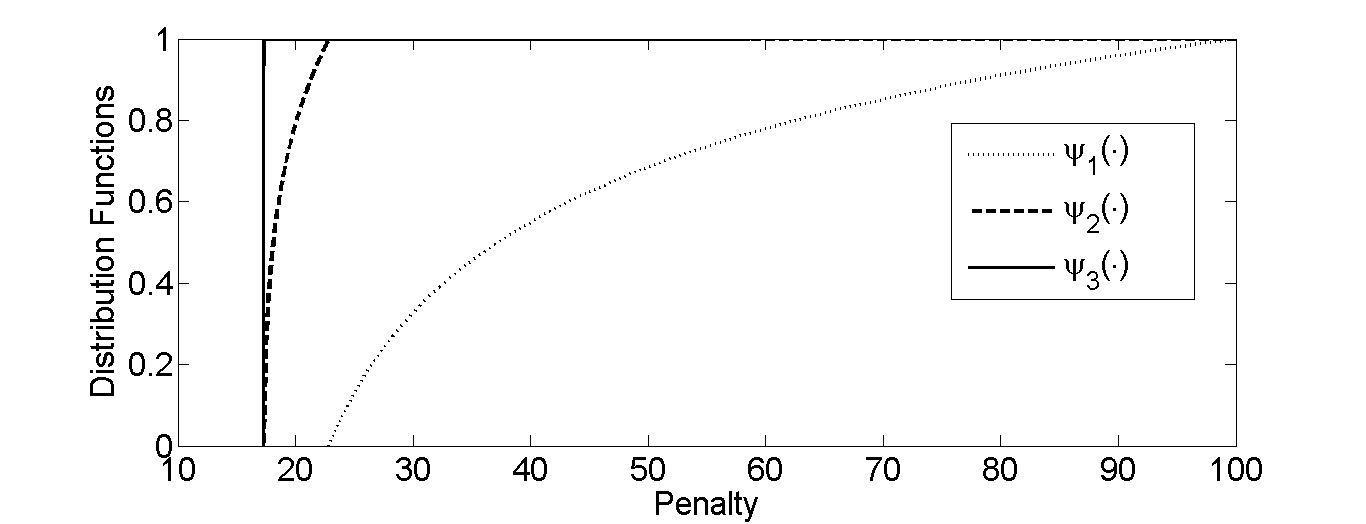}
\vspace*{-.1cm}

\end{center}
\end{minipage}\hfill
\begin{minipage}{.32\linewidth}
\begin{center}
\includegraphics[width=60mm, height=40mm]{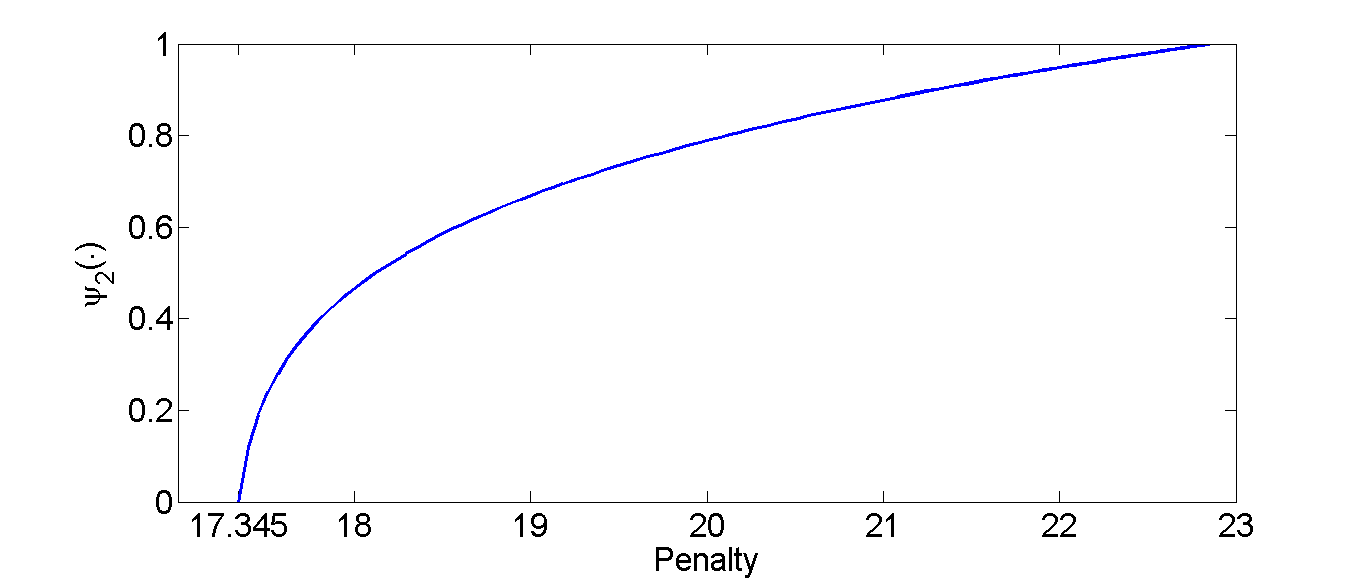}
\vspace*{-.1cm}
\end{center}
\end{minipage}\hfill
\begin{minipage}{.32\linewidth}
\begin{center}
\includegraphics[width=60mm, height=40mm]{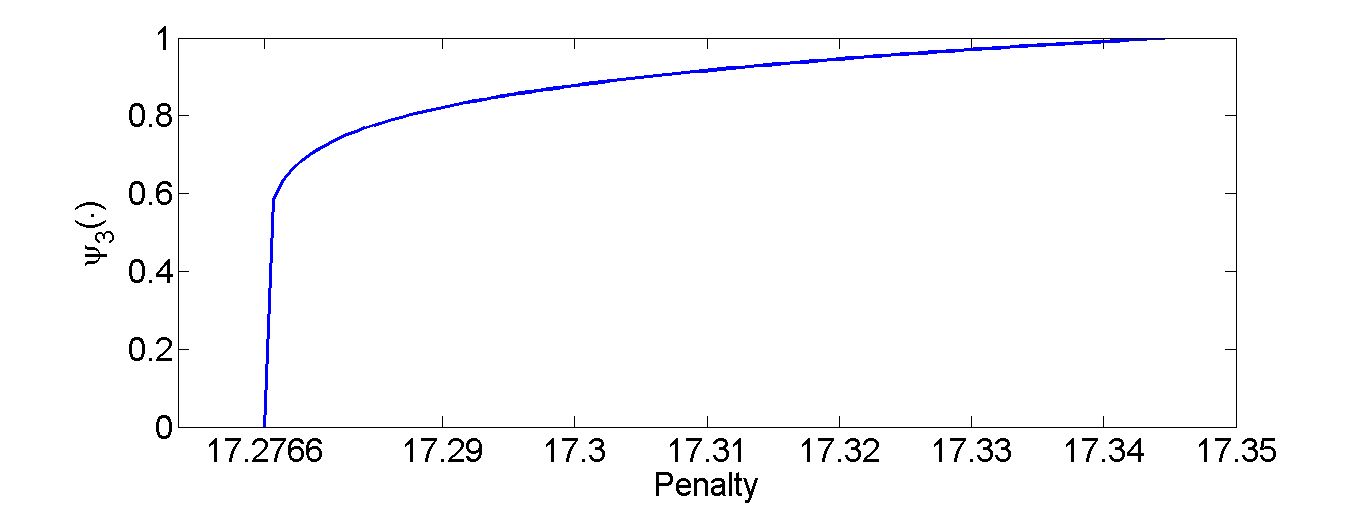}
\end{center}
\end{minipage}
\caption{\small Figure in the left hand side shows the d.f. $\psi_i(\cdot), i=1,\ldots,3$ as a function of penalty for an example setting: $v=100, c=1, l=21, m=10, n=3, q_1=q_2=q_3=0.2$ and $g_i(x)=x-i^3$. Note that support sets of $\psi_i(\cdot)$s are disjoint with $L_3=17.2766$, $U_3=17.345=L_2$, $U_2=22.864=L_1$, and $U_1=100=v$. Figures in the center and the right hand side show d.f. $\psi_2(\cdot)$ and $\psi_3(\cdot)$ respectively, using different scales compared to the left hand figure.}
\label{fig:dist}
\vspace*{-.6cm}
\end{figure*}
\begin{thm}\cite{isit,arnob_ton}\label{singlelocation}
The strategy profile, in which each  primary randomizes over the penalties in the range $[L_j,L_{j-1}]$ using the continuous distribution function $\phi_j(\cdot)$ (Lemma~\ref{lm:computation}) when the channel state is $j$, is the unique NE strategy profile. The expected payoff that a primary attains at every penalty within the interval $[L_j, L_{j-1}]$ is  $p_j-c$ at channel state $j$.
\end{thm}
\subsection{Multiple Asymmetric NEs}\label{sec:multipleNEs}
We first show that there can be multiple NEs in this game unlike in the single location game. Consider the linear conflict graph (Fig.~\ref{fig:linear}) with 2 nodes, 2 primaries and 1 secondary. 

We show multiple asymmetric NEs for two different settings which we have discussed in Section~\ref{sec:two settings}. First, we consider the setting where the channel state is the same across the network. Thus, a primary needs to select a strategy when the channel state is not $0$ across the network. Note that if primaries selects different nodes, then each primary can attain a maximum profit of $(f_i(v)-c)$ at the channel state $i$ which corresponds to selecting penalty $v$. Thus, both the following strategy profiles are asymmetric NE: 1) primary 1 (2, respectively) selects $V_1$ ($V_2$, respectively) w.p. $1$ and selects penalty $v$ irrespective of the channel state; 2) primary 1 (2, respectively) selects $V_2$ ($V_1$, respectively) w.p. $1$ and selects penalty $v$ w.p. $1$ irrespective of the channel state across the network. The realization of one of the above NEs is possible only when a primary knows other\rq{}s strategy apriori; this is ruled out  due to non-cooperation.  Thus, asymmetric NE can not be realized in this game.

Now, we will provide multiple asymmetric NE strategies for the above linear conflict graph when the channel state can be different at different locations using the NE penalty strategy for single location as presented in Section~\ref{sloneshot}. We need to specify strategy at each possible channel state vector.   We consider $n=1$ i.e. at any given node the channel is either available (state $1$) or not (state $0$). We also consider that the channel state of a primary is $1$ at a given location w.p. $q_1$ independent of the channel state at other location.  The following strategy profiles are NE strategy profiles:
i) When the channel state vector is $(0,1)$ ($(1,0)$ respv.) then a primary selects node $2$ ($1$ respv.) w.p. $1$ and selects the single location penalty strategy stated in Theorem~\ref{singlelocation} with $q_1q_0$ in place of $q_1$\footnote{$q_1q_0$ is the probability that the channel state vector is either $(0,1)$ or $(1,0)$.}. When the channel state vector is $(1,1)$ then primary $1$ (primary $2$ respv.) selects node $1$ (node $2$ respv.) w.p. $1$ and selects penalty $v$ w.p. $1$. 

ii) When the channel state vector is either $(0,1)$ or $(1,0)$ then the strategy profile is the same as before. When channel state vector is $(1,1)$ then primary $1$ (primary $2$ respv.) selects node $2$ (node $1$ respv.) w.p. $1$ and selects penalty $v$ w.p. $1$. 

Note that NE strategy profiles cited above are asymmetric. The game is a symmetric one since  primaries have the same action sets, payoff functions and their channels are statistically identical.  In a symmetric game, we have already discussed in Section~\ref{sec:strategy} that implementing an asymmetric NE is difficult.
  We therefore focus on finding a symmetric NE and investigate whether it is unique. Clearly, for any symmetric NE, we can represent the strategy of any primary as $S = (\psi_{1}(.), \psi_2(.),.....,\psi_{|\mathcal{P}|}(.)) $ where we drop the index corresponding to the primary.

\subsection{A Separation Result}\label{sec:separation} 



We now observe that the NE penalty selection at a node in an independent set can be uniquely computed using the single location NE penalty selection strategy stated in Section~\ref{sloneshot} once the independent set selection strategy is known. 
\begin{lem}\label{separation-s1}
Suppose, under a symmetric NE, each primary offers its channel which is at state $j$ at node $a$ for sale at node $a$ w.p. $\alpha_{a,j}$. Then, the unique NE penalty distribution of each primary is the d.f. $\phi_j(\cdot)$ as described in Lemma~\ref{lm:computation} with $\alpha_{a,j}$ in place of $q_j$ at node $a$.
\end{lem} 

We next obtain the expression for $\alpha_{a,j}$. We first introduce some notations:
Let $\mathcal{I}_J$ be the set of independent sets of the graph $G_J$. Let $\mathcal{P}_{a,j}$ be the set of channel state vectors where the channel state is $j$ at node $a$. 
\begin{defn}\label{defn:beta}
Let $\beta_J(I)$ be the probability with which the independent set $I\in \mathcal{I}_J$ is selected by a primary, under a symmetric NE strategy when the channel state vector is $J$. 
\end{defn} 
Note that though $\beta_J(I)$ depends on the symmetric NE strategy, we do not make it explicit in the notation in order to keep the notational simplicity. Thus,
\begin{align}\label{defn:nodeprob_general}
\alpha_{a,j}= \sum_{I\in \mathcal{I}_J: a\in I}\sum_{J: J\in \mathcal{P}_{a,j}} q_J\beta_J(I)
\end{align}

 Since the penalty selection strategy of a primary is unique given the independent set selection strategy $\{\beta_J(I)\}$  (by Lemma~\ref{separation-s1}), henceforth, we only focus on independent set  selection probability which provides the node selection probability as defined in (\ref{defn:nodeprob_general}). 


\vspace{-0.3cm}
\section{Same Channel State Across the region}\label{sec:same_channel_state}
We first consider the setting where the channel state is the same across the network. Recall from Section~\ref{sec:two settings} that this setting occurs when the region is of moderate size. We first introduce some notations specific to this setting (Section~\ref{sec:modification}).  We focus on symmetric NEs on a specific class of conflict graphs known as mean valid graph since conflict graphs of most of the commonly observed  wireless networks of moderate sizes belong to this category (Section~\ref{sec:meanvalidgraph}). We subsequently focus on a policy which provides a storage and computation efficient NE strategy (if it exists) (Section~\ref{sec:policy}). We identify certain key properties that any NE strategy profile of the above policy  (should it exist) must satisfy (Section~\ref{sec:structuremultiplenodes}). Then, we show that the identified structure is a unique and there exists a strategy profile  which satisfies the identified structure (Theorem~\ref{dist:existence}). We show that the strategy profile which satisfies the identified structure is an NE (Theorem~\ref{nemulti}). Finally, we investigate the uniqueness and implementation issues of the symmetric NE profile (Section~\ref{sec:symmetricNEunique}). 

\subsection{Modifications of Notations}\label{sec:modification}
Since the channel state is the same across the region, we denote the channel state vector $J$ as the scalar $j$ in this setting when the channel state is $j$ at each location. For example, if the channel state is $3$ everywhere, we denote the channel state at the network as $3$.  $q_J=0$ for all $J$ where the channel state is not identical at each location and we denote the probability that the channel state is $j$ over the region as $q_j$ with slight abuse of notation.  Note that in this setting, when the channel state is $j\geq 1$, then the channel is available at each node, hence, a primary always selects an independent set from the conflict graph $G$ when the channel is available. 

We replace $\beta_{J}(I)$ in Definition~\ref{defn:beta} with $\beta_j(I)$ which denotes the probability with which a primary selects independent set $I$ under a symmetric NE strategy. Note that $\mathcal{P}_{a,j}$ is now simply $j$. $\alpha_{a,j}$ is thus,  
\begin{align}\label{charac}
\alpha_{a,j}= \sum_{I: a\in I} q_j\beta_j(I)
\end{align}
Also note from (\ref{prob1}) that the channel state is $0$ over the network with some non zero probability i.e.
\begin{align}\label{prob}
\sum_{j=1}^{n}q_j<1
\end{align}
The cardinality of the strategy space $\mathcal{P}$ in this setting is $n$. The  NE strategy profile is thus represented as $(\psi_{1}(\cdot),\ldots,\psi_n(\cdot))$ in this setting.   Note that though the state of a channel is the same across the nodes, the propagation condition and the usage level of different channels can be different, thus, a primary is still not aware of the states of the channel of other primaries. 
 \subsection{Mean Valid Graphs}\label{sec:meanvalidgraph}
 In practice most of the finite size wireless networks are of the following types:
\begin{itemize}
\item Wireless network of roadside shops.
\item Wireless network of buildings.
\item Cellular networks with hexagonal or square cells.
\end{itemize}
Conflict graphs of all the above wireless networks belong to a category, introduced as {\em mean valid graphs}\cite{gauravjsac}. 

 
 \begin{defn}\label{dmvg}\cite{gauravjsac}
A graph $G=(V,E)$ is said to be a mean valid graph if and only if 
\begin{enumerate}
\item Its vertex set can be partitioned into $d$ disjoint maximal\footnote{ An independent set $I$ is said to be maximal if for each $a\notin I, a\in V$, $I\cup \{a\}$ is not an independent set \cite{graph}.} I.S. for some integer $d\geq 2: V=I_1\cup I_2\cup\ldots\cup I_d$\footnote{For example,  linear conflict graph (Fig.~\ref{fig:linear}) is mean valid graph with $d=2$, with $I_1$ being the set of odd numbered nodes and $I_2$ being the set of even numbered nodes. In Fig.~\ref{fig:grid} $d=4$, with $I_1=\{V_{1,1},V_{1,3},\ldots,V_{1,k_o},V_{3,1},V_{3,3},\ldots,V_{3,k_o},\ldots\}, I_2=\{V_{1,2},V_{1,4},\ldots, V_{1,k_e},V_{3,2},V_{3,4},\ldots,V_{3,k_e},\ldots\}, I_3=\{V_{2,1},V_{2,3}, \ldots,V_{2,k_o}, V_{4,1}, V_{4,3},\ldots, V_{4,k_o},\ldots \}, I_4=\{V_{2,2}, V_{2,4}, \ldots, V_{2,k_e}, V_{4,2},V_{4,4},\ldots, V_{4,k_e},\ldots\}$, where $k_o$ (respectively, $k_e$) denote the greatest odd (respectively, even) integer less than or equal to $k$.} where $I_s, s\in\{1,\ldots,d\}$, is a maximal independent set and $I_s\cap I_r=\emptyset, s\neq r$.  Let, $|I_s|=M_s$,
\begin{equation}\label{eq:orderedcardinality}
M_1\geq M_2\geq \ldots\geq M_d.
\end{equation}
and $I_s=\{a_{s,k}:k=1,\ldots, M_s\}.$ 
\item Suppose $I\in \mathcal{I}$contains $m_s(I)$ nodes from $I_s,s=1,\ldots,d$, then ,
\begin{equation}\label{mvg}
\sum_{s=1}^{d}\dfrac{m_s(I)}{M_s}\leq 1\quad \forall I\in \mathcal{I}.
\end{equation}
\end{enumerate}
\end{defn}
$I_1,\ldots, I_d$ are said to characterize the mean valid graph. The following graphs are mean valid graphs\cite{gauravjsac}.
\begin{itemize}
\item Linear Graph constitutes a conflict graph for locations along a highway or a row of shops (Fig.~\ref{fig:linear}). It is a mean valid graph with $d=2$. 
\item Grid Graph constitutes a conflict graph for a building (Fig.~\ref{fig:grid})  or cellular network with square cells. It is a mean valid graph with $d=4$. Three dimensional grid graph is also a mean valid graph with $d=8$.
\item Conflict graph of a cellular network with hexagonal cells is also a mean valid graph with $d=3$, if it has an even number of rows and all rows have the same number of nodes which should be a multiple of 3.
\end{itemize} 
Henceforth, we focus on mean valid graphs in this setting. 

\subsection{ A storage \& Computation efficient policy}\label{sec:policy}
 As in any graph, in mean valid graphs, the number of independent sets grows exponentially with the number of nodes. We have to compute probability distribution over all independent sets in order to find an independent set selection strategy. Thus, computation and storage requirements grow exponentially as the number of nodes increases. However, mean valid graphs are characterized by maximal independent sets $I_1,\ldots,I_d$ which partition the set of nodes. 
So, if there exists an NE strategy profile which only selects independent sets amongst $I_1,\ldots, I_d$, then we only need to store $d$ independent sets and the corresponding probability distribution. Thus, the storage and computation requirement only scales with $d$ and does not increase exponentially with the number of nodes. We therefore examine if there exists an NE strategy profile under which 
\begin{itemize}
\item Each primary  selects \textbf{only} independent sets  in $\{I_1,\ldots, I_d\}$. Specifically, at channel state $j$, independent set $I_k, k\in \{1,\ldots,d\}$ is selected with probability $t_{k,j}$.
\end{itemize}
Under the policy, thus, 
\begin{align}\label{eq:beta}
\beta_j(I_k)=t_{k,j}\quad \forall k\in \{1,\ldots d\} \quad \text {such that }
\sum_{k=1}^{d}\beta_j(I_k)=1.
\end{align}
Thus,  from (\ref{charac}) and (\ref{eq:beta})
for any two nodes $s,r\in I_k ,k\in \{1,\ldots,d\}, j\in\{1,\ldots,n\}$:
\begin{align}\label{dist1}
\alpha_{s,j}=\alpha_{r,j}=q_jt_{k,j}\quad \sum_{k=1}^{d}t_{k,j}=1.
\end{align}
In the next section, we show that there exists a unique symmetric NE strategy which satisfies (\ref{dist1}).
\vspace{-0.2cm}

\subsection{Characterization of Symmetric NE}\label{sec:structuremultiplenodes}
We first, characterize the properties that any symmetric NE strategy profile of the form (\ref{dist1}) must satisfy. 

By virtue of Theorem~\ref{singlelocation} and Lemma~\ref{separation-s1}, we know the penalty selection strategy for each state at a given node for a given NE independent set selection strategy. The support sets of penalty distributions are contiguous (Section~\ref{sloneshot}). However, the end-points of the support sets are not necessarily the same across the location. Surprisingly, we show that the upper endpoints of the penalty selection strategy at a particular channel state $i$, $i=1,\ldots,n$ are identical across different locations regardless of the choice of independent sets (Lemma~\ref{upend}). We show that there exists a threshold such that only those independent sets, whose cardinalities are equal to or greater than that threshold, are selected with positive probabilities (Lemma~\ref{bestr}).   Drawing from the above lemmas we characterize the structure that any NE strategy profile of the form (\ref{dist1}) (if it exists) has to satisfy (Theorem~\ref{thm:structure}). The proofs of the results have been provided at the end of this subsection. 

We start with some notations which we use throughout. 
\begin{defn}
Let,
\begin{align}\label{defnW}
W(x)=1-w(x).
\end{align}
\end{defn}
Since $w(\cdot)$ is continuous and strictly increasing (by (\ref{d4})),thus, $W(\cdot)$ is a continuous and  a strictly decreasing function with $W(0)=1$.
\begin{defn}\label{gammadefn}
Let $\gamma_{s,j}$ denote the probability that a channel of state $j$ or higher is offered at a node of $I_s$. Thus,
\begin{align}\label{recurg}
\gamma_{s,j}=\sum_{k=j}^{n}t_{s,k}q_k=\sum_{k=j}^{n}\alpha_{a,k}.
\end{align}
\end{defn}
From (\ref{recurg}), we obtain a recursive method to calculate $\gamma_{s,j}$.
\begin{align}\label{eq:recursivegamma}
\gamma_{s,j-1}=\sum_{k=j-1}^{n}t_{s,k}q_k=t_{s,j-1}q_{j-1}+\gamma_{s,j}.
\end{align}

In the class of policies of the form (\ref{dist1}), $\alpha_{a,j}$ is equal to $q_jt_{s,j}$ for every node $a$ in independent set $I_s, s\in \{1,\ldots,d\}$.  Thus, by Lemma~\ref{separation-s1} the penalty selection strategy at any node of $I_s$ is given by Lemma~\ref{lm:computation} with $q_jt_{s,j}$ in place of $q_j$.  Thus, by (\ref{n51}), (\ref{n52}), and Theorem~\ref{singlelocation}, expected payoff obtained by a primary at every node of $I_s$ at channel state $j$ is--
\vspace{-0.4cm}
\begin{align}\label{n152a}
p_{s,j}-c& =(f_j(U_{s,j})-c)(1-w(\sum_{i=j}^{n}t_{s,i}q_i))\nonumber\\
& =(f_j(U_{s,j})-c)W(\gamma_{s,j})
\end{align}
where 
\begin{eqnarray}
U_{s,j}& =&g_j(\dfrac{p_{s,j}-c}{W(\gamma_{s,j})}+c) \quad U_{s,1}=v, U_{s,j}=L_{s,j-1}\label{n152up}\\
L_{s,j}& =& g_j(\dfrac{p_{s,j}-c}{W(\gamma_{s,j+1})}+c)\quad L_{s,0}=v.\label{n152inv}
\end{eqnarray}
\begin{rmk}
Starting from $U_{s,1}=v$, we can find $p_{s,1}$ using (\ref{n152a}) which we use to find  $L_{s,1}$ (from (\ref{n152inv})). Since $L_{s,1}=U_{s,2}$, thus utilizing $U_{s,2}$   we obtain $p_{s,2}$ (from (\ref{n152a})) which in turn gives $L_{s,2}$ (from (\ref{n152inv})). Thus, recursively we obtain $U_{s,j}, p_{s,j}, L_{s,j}$ for all $s\in \{1,\ldots,d\}$ and $j\in \{1,\ldots,n\}$. Hence, we can easily compute a penalty selection strategy at each node of $I_s$ for a given $t_{s,j}$.
\end{rmk}
\begin{rmk}
Note from Lemmas~\ref{lm:computation} and ~\ref{separation-s1} that  each primary selects penalty only from the interval $[L_{s,j},U_{s,j}]$ at channel state $j$ at every node of $I_s$ when $t_{s,j}>0$.
\end{rmk} 

 Since $p_{s,j}-c$ is the expected payoff that a primary gets at any node in $I_s$ at channel state $j$ when primaries select $I_s$ with probability $t_{s,j}>0$, thus, the expected payoff to a primary at channel state $j$ over independent set $I_s$ when $t_{s,j}>0$ is 
\begin{equation}\label{payoffis}
M_s(p_{s,j}-c)=M_s(f_j(U_{s,j})-c)W(\gamma_{s,j})\quad(\text{from} (\ref{n152a})).
\end{equation}
Now, we introduce some notations that we use throughout.
\begin{defn}\label{maxexpayoff}
Let $P_j(I_k)$ denote the maximum expected payoff that a primary can get at independent set $I_k$ at channel state $j$ when other primaries select a symmetric  NE strategy profile which is of the form (\ref{dist1}) . Let $P_j^{*}$ be the maximum among $P_j(I_r)$ $r\in \{1,\ldots,d\}$ i.e. 
\begin{align}
P_j^{*}=\max_{r\in \{1,\ldots,d\}}P_j(I_r).\nonumber
\end{align}
Let $B_j$ denote the set of indices out of $I_1,\ldots, I_d$ which are selected with positive probability under a symmetric NE strategy profile at channel state $j$.
\end{defn}
At channel state $j$ an independent set  is selected with positive probability in an NE strategy profile only if the expected payoff at that independent set is\footnote{Consider that in an NE strategy profile $I_s$ is selected w.p. $t_{s,j}>0$, but expected payoff is strictly less than $P^{*}_j$ which it obtains at $I_r$ (say). Let in the NE strategy profile $I_r$ is selected w.p. $t_{r,j}$. Note that the expected payoff of a primary at an independent set only depends on the strategy of other primaries. Thus, a primary can unilaterally deviate by selecting $I_r$ w.p. $t_{s,j}+t_{r,j}$ and $I_s$ w.p. $0$; but under the new strategy profile its expected payoff is strictly higher. Hence, the original strategy profile can not be an NE.} $P^{*}_j$ ; hence when the channel state is $j$, then
\begin{eqnarray}\label{breq1}
M_s(f_j(U_{s,j})-c)W(\gamma_{s,j})=P^{*}_j\quad \text{if } s\in B_j (\text{from} (\ref{payoffis})).
\end{eqnarray}
Now, we are ready to state the results.
\begin{lem}\label{upend}
 If $t_{s,j}>0, t_{r,j}>0$, then $U_{s,j}=U_{r,j}$.
\end{lem}
The above lemma shows that upper end points of penalty selection strategy is the same across the nodes of  the independent sets which are chosen with positive probability.\footnote{Note that we have not shown any relation between $L_{s,j}$ and $L_{r,j}$. Thus, even though $U_{s,j}=U_{r,j}$, it is possible that $L_{s,j}\neq L_{r,j}$. But if $t_{s,j+1}>0, t_{r,j+1}>0$, then from Lemma~\ref{upend} we obtain $U_{s,j+1}=U_{r,j+1}$; since $L_{s,j}=U_{s,j+1}, L_{r,j}=U_{r,j+1}$, thus we have $L_{s,j}=L_{r,j}$. Hence, lower endpoint of penalty selection strategy at every node of independent sets $I_s, I_r$ is also the same if both $I_s, I_r$ are selected with positive probabilities for both the states $j$ and $j+1$.}
\begin{rmk}
From lemma ~\ref{upend} we can write $U_{s,j}$ as $U_j$ $\forall s\in B_j$. So, for any $s,r\in B_j$, we must have from (\ref{breq1}) 
\begin{align}\label{n2}
M_s(f_j(U_{j})-c)W(\gamma_{s,j})& =M_r(f_j(U_{j})-c)W(\gamma_{r,j})=P^{*}_j.\nonumber\\
M_sW(\gamma_{s,j})& =M_rW(\gamma_{r,j}).
\end{align} 
\end{rmk}
Next lemma characterizes the best response set $B_j$.
\begin{lem}\label{bestr}
There exists an integer $d_j\in\{1,\ldots,d\}$, such that $I_1,\ldots, I_{d_j}$ are selected with positive probability and $I_{d_j+1},\ldots,I_d$ are selected with zero probability at channel state $j$.
\end{lem}
Thus, from (\ref{eq:orderedcardinality}), only those independent sets whose cardinalities are greater than or equal to $M_{d_j}$ are selected with positive probabilities at channels state $j$ . 
We show in Lemma~\ref{prop} that this  above threshold $M_{d_j}$ is a non-decreasing function in channel state $j$. 

 In an NE strategy only those independent sets are selected with positive probabilities which give an expected payoff of $P^{*}_j$, thus we can evaluate the expected payoff under the NE strategy using Lemma~\ref{bestr}. Since we know from Lemma~\ref{bestr} that NE strategy profile only selects those independent sets whose indices are less than or equal to $d_j$, thus, under NE strategy expected payoff of a primary at channel state $j$  is given by
\begin{align}\label{c61}
P^{*}_{j}=M_s(f_j(U_j)-c)W(\gamma_{s,j}) \quad s\leq d_j.
\end{align}
 
We will also show that $P^{*}_j\geq M_r(f_j(U_j)-c)W(\gamma_{r,j})$ for $r>d_j$ to prove Lemma~\ref{bestr}  (Lemma~\ref{lm:strictpayoff} in Section~\ref{sec:proof_lemma_bestr}). Drawing from the above it readily follows that
\begin{thm}\label{thm:structure}
The structure of a symmetric NE strategy profile which satisfies (\ref{dist1}) (if it exists), is of the following form $\forall a\in I_s$ for $j\in \{1,\ldots,n\}$
\begin{align}\label{dist}
\alpha_{a,j}=q_jt_{s,j},\sum_{s=1}^{d}t_{s,j}=1, t_{s,j}>0, s\leq d_j, t_{s,j}=0, s>d_j
\end{align}
such that 
\begin{align}\label{ne2}
& M_1W(\gamma_{1,j})=\ldots=M_{d_j}W(\gamma_{d_j,j})\geq M_{d_{j}+1}W(\gamma_{d_j+1,j})
 \geq M_{d_{j}+2}W(\gamma_{d_j+2,j})\geq \ldots\geq M_dW(\gamma_{d,j}).
\end{align}
\end{thm}
Note that the number of equations  increases linearly with the number of states $n$.

Theorem~\ref{thm:structure} provides an  iterative way to compute $t_{s,j}$ for all $s,j$. Noting that $\gamma_{s,n}=t_{s,n}q_n$, (\ref{ne2})  has only one variable $t_{s,n}$ at $j=n$ for $s\in \{1,\ldots,d\}$. Thus, we first compute $t_{s,n}$ for all $s$ using (\ref{dist}) and (\ref{ne2}) for $j=n$. From (\ref{eq:recursivegamma}), $\gamma_{s,n-1}$ depends on $\gamma_{s,n}$ and $t_{s,j-1}$. Since we have already computed $t_{s,n}$ or $\gamma_{s,n}$, thus we solve for $t_{s,n-1}$ from (\ref{dist}) and (\ref{ne2}).  Thus, recursively we obtain $t_{s,j}$ for all $s$ and $j$. {\em A primary only needs to know $I_1,\ldots, I_d$ to compute the independent set selection strategy and does not need to know the information regarding the network (e.g. edges)}. 
\begin{example}\label{example:dist_calculate}
We consider a grid graph with $d=4$ and $k=5$ (Fig.~\ref{fig:grid}) . Here,  $M_1=9, M_2=M_3=6, M_4=4$. We consider $l=20, m=6, n=3, q_1=q_2=q_3=0.2$. We first calculate $t_{s,3}$ for all $s$. We obtain $M_1W(\gamma_{1,3})=7.5324$, $M_2W(\gamma_{2,3})=6, M_3W(\gamma_{3,3})=6$, $M_4W(\gamma_{4,3})=4$. Thus, $d_3=1$ and the solution of (\ref{dist}) and (\ref{ne2}) is: $t_3=(1,0,0,0)$, where $t_j=(t_{1,j},\ldots,t_{d,j})$ for $j=1,\ldots,3$. Next, we compute $t_{s,2}$ following the recursive algorithm we stated. We obtain $d_2=3$ and $ t_2=(0.2532,0.3734,0.3734,0)$.  Finally, we calculate $t_{s,1}$. We obtain $d_1=3$ and $t_1=(0.071,0.4645,0.4645,0)$ Fig.~\ref{fig:example_ind_prob} shows plots of $t_{s,j}$ for all $s$ and $j$.
\end{example}
\begin{figure}
\includegraphics[width=90mm,height=40mm]{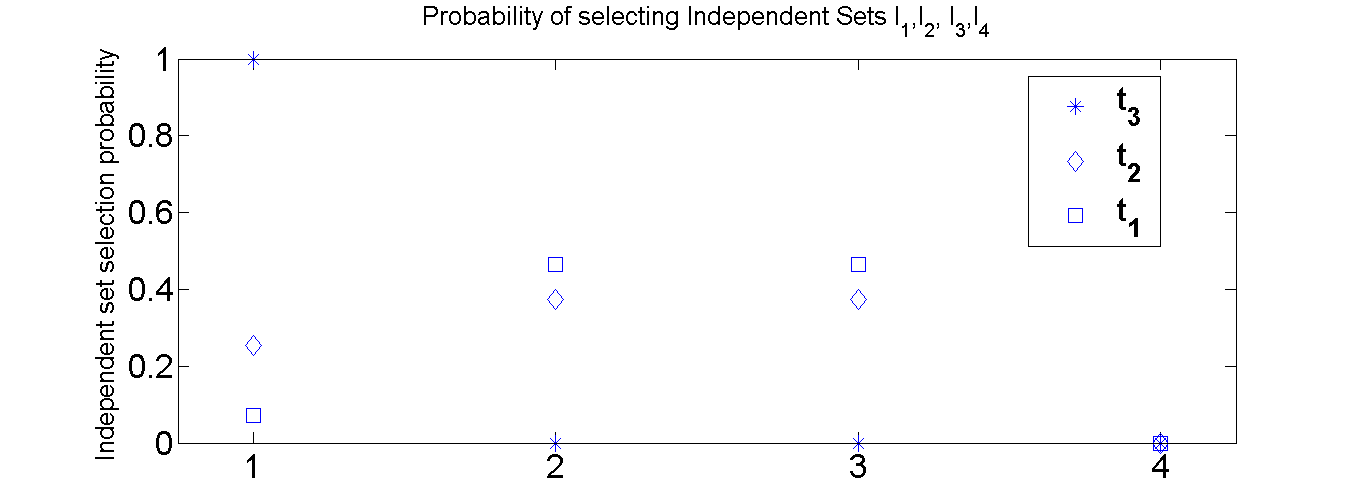}
\caption{\small This figure shows  $t_j=(t_{1,j},\ldots,t_{d,j})$ at channel state $j=1,2,3$ for Example~\ref{example:dist_calculate}.}
\label{fig:example_ind_prob}
\vspace{-0.5cm}
\end{figure}

\subsubsection{Proof of Lemma~\ref{upend}}We first deduce some results  which we use throughout. 

Since $\gamma_{s,j}\leq \sum_{i=1}^{n}q_i<1$, thus $p_{s,j}-c>0$. Hence,
\begin{align}\label{pay0}
f_j(U_{s,j})>c, f_j(L_{s,j})>c.
\end{align} 
Now, we provide the expression for expected payoff that a primary attains at $L_{s,i}$ $i=1,\ldots,n$ at any node in $I_s$ at channel state $j$. Note that players with channel state higher than $i$ select a penalty lower than or equal to $L_{s,i}$ with probability $1$ and players with channel state lower than or equal to $i$ select a penalty lower than or equal to $L_{s,i}$ with probability $0$ at every node of $I_s$. Thus, the expected payoff to a primary when it selects penalty $L_{s,i}$ at channel state $j\in \{1,\ldots,n\}$ at any node of $I_s$ is
 \begin{align}\label{expayofflowerendpointstate}
 (f_j(L_{s,i})-c)W(\sum_{k=i+1}^{n}q_{k}t_{s,k})=(f_j(L_{s,i})-c)W(\gamma_{s,i+1}).
\end{align}

Now, we state and prove Observations~\ref{identity1} and \ref{identity2} which we use throughout.
\begin{obs}\label{identity1}
$\gamma_{s,k}=\gamma_{s,k_1}+\sum_{i=k}^{k_1-1}t_{s,i}q_i$ for $s\in \{1,\ldots,d\}$, $n\geq k_1>k$.
\end{obs}
The observation readily follows from (\ref{recurg}). Since from (\ref{recurg})
\begin{align}
\gamma_{s,k}=\sum_{i=k}^{k_1-1}t_{s,i}q_i+\sum_{i=k_1}^{n}t_{s,i}q_i
=\sum_{i=k}^{k_1-1}t_{s,i}q_i+\gamma_{s,k_1}.\nonumber
\end{align}
\begin{obs}\label{identity2}
$U_{s,j}=L_{s,j}$ for $j\in\{1,\ldots,n\}$ if and only if (iff) $t_{s,j}=0$. $U_{s,j}=L_{s,k}$ iff $t_{s,i}=0$ $\forall k<i<j$.  Hence, $U_{s,j}=v$ iff $t_{s,k}=0$ $\forall k<j$.
\end{obs}
\begin{proof}
$L_{s,j}=U_{s,j}$ implies from (\ref{n152up}) and (\ref{n152inv}) that $\gamma_{s,j+1}=\gamma_{s,j}$; thus by Observation~\ref{identity1} we have $t_{s,j}=0$. On the other hand if $t_{s,j}=0$ then by Observation~\ref{identity1} $\gamma_{s,j}=\gamma_{s,j+1}$. Thus,   it follows that  $U_{s,j}=L_{s,j}$ iff $t_s{,j}=0$. 

Since $L_{s,k}=U_{s,k+1}$, hence $U_{s,j}=L_{s,k}$ iff $t_{s,i}=0$ $\forall k<i<j$. Hence, $U_{s,j}=L_{s,1}=U_{s,1}$ iff $t_{s,i}=0$ $\forall i<j$. On the other hand by (\ref{n152up}) $U_{s,1}=v$.  Thus, the result follows.
\end{proof}

Now we are ready to show Lemma~\ref{upend}.\\
\textit{Proof of Lemma~\ref{upend}}: 


Since both $s,r\in B_j$, hence from (\ref{breq1})
\begin{align}\label{n30}
M_s(f_j(U_{s,j})-c)W(\gamma_{s,j})=M_r(f_j(U_{r,j})-c)W(\gamma_{r,j})=P^{*}_j.\nonumber\\
\end{align}
Suppose, the statement is false, i.e. $U_{s,j}\neq U_{r,j}$ when $s,r\in B_j$. Without loss of generality, we can assume that $U_{s,j}>U_{r,j}$. So, $U_{r,j}<v$. Thus, by Observation~\ref{identity2}, there exists $k$ such that $t_{r,k}>0$, $k<j$ and $L_{r,k}=U_{r,j}$. Thus, from (\ref{breq1})
\begin{align}\label{n27a}
& P^{*}_k=M_r(f_k(U_{r,k})-c)W(\gamma_{r,k})\nonumber\\ & =M_r(f_k(L_{r,k})-c)W(\gamma_{r,k+1}) (\text{from }(\ref{n152up})\& (\ref{n152inv})).
\end{align}
If a primary selects penalty $U_{s,j} (=L_{s,j-1})$ at a node of $I_s$ when its channel state is $k$, then from (\ref{expayofflowerendpointstate}) its expected payoff would be
\begin{align}
(f_k(U_{s,j})-c)W(\gamma_{s,j})=\dfrac{(f_k(U_{s,j})-c)}{(f_j(U_{s,j})-c)}\dfrac{P^{*}_j}{M_s} \quad (\text{from } (\ref{n30})).\nonumber
\end{align}
Thus a primary obtains an expected payoff of at least 
\begin{align}
M_s(f_k(U_{s,j})-c)W(\gamma_{s,j})=P^{*}_j\dfrac{(f_k(U_{s,j})-c)}{(f_j(U_{s,j})-c)}.\nonumber
\end{align}
at independent set $I_s$ at channel state $k$. By definition of $P^{*}_k$,
\begin{align}\label{n31a}
P^{*}_j\dfrac{(f_k(U_{s,j})-c)}{(f_j(U_{s,j})-c)}\leq P^{*}_k.
\end{align}
Since $U_{r,j}=L_{r,k}$ and thus $f_k(U_{r,j})>c$ (by (\ref{pay0})). Thus expected payoff at $I_r$ at channel state $j$ is at least $M_r(f_j(U_{r,j})-c)W(\gamma_{r,k+1})$ which is
\begin{align}
&=\dfrac{(f_j(U_{r,j})-c)}{f_k(U_{r,j})-c}P^{*}_k\quad (\text{from }(\ref{n27a}))\nonumber\\
& \geq P^{*}_j\dfrac{(f_j(U_{r,j})-c)(f_k(U_{s,j})-c)}{(f_k(U_{r,j})-c)(f_j(U_{s,j})-c)}\quad (\text{from }(\ref{n31a}))\nonumber\\
& >P^{*}_j\quad (\text{from }(\ref{con1}), j>k, U_{s,j}>U_{r,j})\nonumber
\end{align}
which is not possible by Definition~\ref{maxexpayoff}.\qed

\subsubsection{Proof of Lemma~\ref{bestr}}\label{sec:proof_lemma_bestr}
We state and prove Lemmas~\ref{lowerendpoint}, \ref{lmmax}, and \ref{lm:strictpayoff} which we use to prove Lemma~\ref{bestr}.
\begin{lem}\label{lowerendpoint}
$L_{s,k}\geq U_j$ if $s\in B_k,s\notin B_j, k<j, j\geq 2$.
\end{lem}
\begin{rmk}
Note that if $s\in B_k\cap B_j$ and $k<j$, then from (\ref{n152up}), $L_{s,k}\geq U_j$. But, it is not apriori clear the relationship between $L_{s,k}$ and $U_j$ when $s\in B_k$ but $s\notin B_j$ for $k<j$. The above lemma provides the answer.
\end{rmk}
Since $s\in B_k$ , thus expected payoff obtained at $I_s$ at channel state $k$ is $P^{*}_k$ (by (\ref{breq1})). If $U_j>L_{s,k}$ for some $k<j$ and $s\notin B_j$, then it can be shown that by selecting independent set $I_r$ (where $r\in B_j$) a primary can attain a strictly higher payoff compared to $P^{*}_k$ at channel state $k$ which is not possible by Definition~\ref{maxexpayoff}. The argument will be similar to the proof of Lemma~\ref{upend}. Thus, we omit it.

It is not  clear that  $P^{*}_j\geq M_r(f_j(U_j)-c)W(\gamma_{r,j})$ if $r\notin B_j$. Since $r\notin B_j$, thus, a primary  will not employ any penalty selection strategy at any node of $I_r$ when the channel state is $j$. Thus, at any given node in $I_r$, the expected payoff at $U_j$ is still unknown.  The following lemma provides the answer.
\begin{lem}\label{lmmax}
If $r\notin B_j$, then $P^{*}_j\geq M_r(f_j(U_j)-c)W(\gamma_{r,j})$.
\end{lem}
\begin{proof}
Since $r\notin B_j$, hence we must have $t_{r,j}=0$.  Suppose the statement is false, then for some $r\notin B_j$, we must have 
\begin{align}\label{n24}
P^{*}_j& < M_r(f_j(U_j)-c)W(\gamma_{r,j}).
\end{align}
Now we show that a primary will attain an expected payoff which is strictly higher than $P^{*}_j$ at channel state $j$ at independent set $I_r$.\\
Let, $k=\max\{i\in \{1,\ldots,j-1\}:r\in B_i\}$, if $r\notin B_i,\forall i<j$, then set $k=0$. By definition of $k$, $t_{r,i}=0$ $\forall k<i<j$. Thus by Observation~\ref{identity1}, $\gamma_{r,k+1}=\gamma_{r,j}$. Thus, from (\ref{expayofflowerendpointstate}) the expected payoff at $L_{r,k}$ (if $k=0$, then $L_{r,0}=v$) at channel state $j$ is 
\begin{align}\label{upperpay}
(f_j(L_{r,k})-c)W(\gamma_{r,k+1})=(f_j(L_{r,k})-c)W(\gamma_{r,j}).
\end{align}
 Now from Lemma~\ref{lowerendpoint} $L_{r,k}\geq U_j$ when $k>0$. If $k=0$, then $L_{r,k}=v$ by Observation~\ref{identity2}. Thus, $L_{r,k}\geq U_j$ $\forall k$. Hence, from (\ref{upperpay}) total expected payoff at $I_r$ is at least
\begin{align}
M_r(f_j(L_{r,k})-c)W(\gamma_{r,j})& \geq M_r(f_j(U_j)-c)W(\gamma_{r,j})\nonumber\\
& >P^{*}_j \quad (\text{from }(\ref{n24})
\end{align}
which contradicts $P_j^{*}$ from Definition~\ref{maxexpayoff}.
\end{proof}
Thus, if $s, s_1\in B_j$ and $s_2\notin B_j$, then we have 
\begin{align}\label{n34}
 P^{*}_j& =M_{s}(f_j(U_j)-c)W(\gamma_{s,j})=M_{s_1}(f_j(U_j)-c)W(\gamma_{s_1,j})
 \geq M_{s_2}(f_j(U_j)-c)W(\gamma_{s_2,j})\quad (\text{from Lemma~\ref{lmmax}})\nonumber\\
M_sW(\gamma_{s,j})& =M_{s_1}W(\gamma_{s_1,j})\geq M_{s_2}W(\gamma_{s_2,j}).
\end{align}
We now state and prove Lemma~\ref{lm:strictpayoff}.
\begin{lem}\label{lm:strictpayoff}
$M_rW(\gamma_{r,j})\geq M_sW(\gamma_{s,j})$ if $r<s$ for all $j\in \{1,\ldots,n\}$.
\end{lem}
\begin{proof}
Suppose the statement is false i.e. $M_rW(\gamma_{r,j})<M_sW(\gamma_{s,j})$ for some $r<s$ and $j\in \{1,\ldots,n\}$. Since $M_r\geq M_s$ (by (\ref{eq:orderedcardinality})), thus there must exist a $k\in \{j,\ldots,n\}$ such that $M_rW(\gamma_{r,k+1})\geq M_sW(\gamma_{s,k+1})$ but $M_rW(\gamma_{r,k})<M_sW(\gamma_{s,k})$ with $\gamma_{r,n+1}=\gamma_{s,n+1}=0$. 

Since $\gamma_{s_1,k}\geq \gamma_{s_1,k+1}$ (by Observation~\ref{identity1}) $\forall s_1\in \{1,\ldots,d\}$ and $W(\cdot)$ is strictly decreasing function, thus, we have
\begin{align}\label{crossover}
M_rW(\gamma_{r,k})& <M_sW(\gamma_{s,k})\nonumber\\
& \leq M_sW(\gamma_{s,k+1})
\leq M_rW(\gamma_{r,k+1}).
\end{align}
Since $W(\cdot)$ is strictly decreasing function and $\gamma_{r,k}=t_{r,k}q_k+\gamma_{r,k+1}$ (from Observation~\ref{identity1}), thus $t_{r,k}>0$ from (\ref{crossover}); which implies that $r\in B_k$. But this contradicts (\ref{n34}). Hence, the result follows.
\end{proof}
Now, we are ready to show Lemma~\ref{bestr}.\\
\textit{proof of Lemma~\ref{bestr}}:
Suppose that $r<s$, but $r\notin B_k, s\in B_k$ for some $k\in \{1,\ldots,n\}$.  Note from Observation~\ref{identity1} that $\gamma_{s,k}>\gamma_{s,k+1}$ since $t_{s,k}>0$. Since $W(\cdot)$ is strictly decreasing thus $W(\gamma_{s,k})<W(\gamma_{s,k+1})$. On the other hand, since $r\notin B_k$, thus $t_{r,k}=0$. Thus, from Observation~\ref{identity1} $\gamma_{r,k+1}=\gamma_{r,k}$ and therefore, we obtain $W(\gamma_{r,k+1})=W(\gamma_{r,k})$. Thus we obtain from Lemma~\ref{lm:strictpayoff}--
\begin{align}
M_rW(\gamma_{r,k})=M_rW(\gamma_{r,k+1})\geq M_sW(\gamma_{s,k+1})>M_sW(\gamma_{s,k}).\nonumber
\end{align}
But $s\in B_k, r\notin B_k$, thus the above inequality contradicts (\ref{n34}).\qed

\vspace{-0.3cm}
\subsection{Existence}\label{sec:existencemultiplenodes}
Theorem~\ref{thm:structure} characterizes the structure of independent set selection strategy which is of the form (\ref{dist1}). We have not yet shown whether there exists such a distribution and whether such a distribution is unique. We resolve both the issues in the following theorem which we have proved in Appendix~\ref{sec:proof_dist_existence}:
\begin{thm}\label{dist:existence}
There exists a unique probability distribution $t_{j}=(t_{1,j},\ldots,t_{d,j}), j=1,\ldots,n$ which satisfies (\ref{dist}) \& (\ref{ne2}).
\end{thm}
We now show that independent set selection strategy profile described in (\ref{dist}) and (\ref{ne2}) is an NE.
\begin{thm}\label{nemulti}
At channel state $j\in \{1,\ldots,n\}$, consider the following strategy profile. The unique independent set selection strategy profile is given by  (\ref{dist}) and (\ref{ne2}) and at every node of $I_s, s\in \{1,\ldots,d\}$, penalty selection strategy is $\psi_j(\cdot)$ with $q_jt_{s,j}$ in place of $q_j$ as described in Lemma ~\ref{lm:computation}. Such a strategy profile constitutes an NE in the class of mean valid graphs.
\end{thm}
Thus, there exists a symmetric NE which selects an independent set among $I_1,\ldots, I_d$. Such a selection strategy is storage and computationally efficient as explained in the first paragraph of Section~\ref{sec:policy}. By virtue of Theorem~\ref{thm:structure} we also know how to compute the probabilities of these independent sets by solving $n$ equations.
\subsubsection{Outline of Proof of Theorem~\ref{nemulti}}
We first show that  a primary at channel state $j$ attains an expected payoff of $P_j^{*}$ at each independent set $I_s$, $s\leq d_j$. Subsequently, we show that at any  independent set $I_s,  s>d_j$, the maximum attainable payoff of a primary at channel state $j$ is less than $P_j^{*}$ when other primaries select strategies according to (\ref{dist}) and (\ref{ne2}). Finally, we show that if a primary selects an independent set which does not belong to the partition, then, its maximum expected payoff is also less than $P_j^{*}.$ Thus, it shows that a primary attains maximum expected payoff only at independent sets $I_s, s\leq d_j$, hence, a primary  does not have any incentive to deviate unilaterally from the strategy profile which proves the theorem. The detail of the proof is given in Appendix~\ref{sec:proof_nemulti}.
\subsection{Properties of Threshold}\label{sec:properties_threshold}
Recall from Lemma~\ref{bestr} and Theorem~\ref{thm:structure} that a primary only selects those independent sets  which have cardinalities greater than or equal to $M_{d_j}$ with positive probabilities at channel state $j$.  In this section, we discuss some important properties of $d_j, j=1,\ldots,n$. 
\begin{lem}\label{prop}
Threshold is a non-decreasing function of transmission rate i.e. $d_j\geq d_{j+1}$
\end{lem}
From (\ref{eq:orderedcardinality}) and Lemma ~\ref{prop} we have $M_{d_j}\leq M_{d_{j+1}}$.   From Example~\ref{example:dist_calculate}, we obtain $d_3<d_2=d_1$ which validates the above lemma.    In Example~\ref{example:dist_calculate} only $I_1$ is selected when the channel state is the highest i.e. $3$. Thus, a primary never selects $I_2, I_3$ and $I_4$ when its channel has the highest transmission rate.

{\em This tells that in practice, secondary users in some locations can never get access to a channel of higher quality}. In Example~\ref{example:dist_calculate}, users in the locations belonging to independent sets $I_2, I_3$ and $I_4$ will never get access to the highest quality channel. To avoid such socially unacceptable situation a social planner may have to provide some incentives to primaries so that they offer their high quality channels in independent sets of lower cardinalities. Designing such an incentive constitutes an important problem for future research. 



Since $t_{s,j}>0$ for $s\leq d_j$ the following result is immediate from Lemma~\ref{prop}.
\begin{cor}\label{obs1}
$t_{s,k}>0$ implies that $t_{s,j}>0$ where $j<k$; $t_{1,j}>0$ $\forall j\in\{1,\ldots,n\}$.
\end{cor}
Thus, independent set $I_1$ is always selected with positive probability at every channel state (Fig.~\ref{fig:example_ind_prob}). Corollary~\ref{obs1} implies that if a given primary offers its channel at an independent set $I_s, s\in\{1,\ldots,d\}$ with positive probability when the channel provides higher transmission rate, then the primary also offers its channel at  $I_s$ with positive probability when its channel  provides lower transmission rate. But note that the converse is not always true.
\subsubsection{Proof of Lemma~\ref{prop}}\label{sec:proof_bestr}

Suppose, the statement is false, i.e. $d_j<d_{j+1}$ for some $j$. 

From (\ref{ne2}) we obtain for state $j+1$
\begin{align}\label{n74}
M_1W(\gamma_{1,j+1})=M_{d_j}W(\gamma_{d_j,j+1})=M_{d_{j+1}}W(\gamma_{d_{j+1},j+1}).
\end{align}
Since $t_{d_j,j}>0$ thus $\gamma_{d_j,j}>\gamma_{d_{j},j+1}$ by Observation~\ref{identity1}. Since $W(\cdot)$ is strictly decreasing, thus we have
\begin{align}\label{n74in}
W(\gamma_{d_j,j})<W(\gamma_{d_{j},j+1}).
 \end{align} 
Since $d_j<d_{j+1}$, thus $t_{d_{j+1},j}=0$. Thus, from Observation~\ref{identity1}, $\gamma_{d_{j+1},j+1}=\gamma_{d_{j+1},j}$. 
 Thus from (\ref{n74}) and (\ref{n74in}), we obtain
\begin{align}\label{n72}
M_{d_j}W(\gamma_{d_j,j})& <M_{d_{j+1}}W(\gamma_{d_{j+1},j+1})\nonumber\\& =M_{d_{j+1}}W(\gamma_{d_{j+1},j}).
\end{align} 
Since $d_j<d_{j+1}$ thus (\ref{n72}) contradicts (\ref{ne2}).\qed
\vspace{-0.3cm}
\subsection{Uniqueness of Symmetric NE \& Implementation Issues}\label{sec:symmetricNEunique}
Till now we have shown that when primaries select among maximal independent sets characterizing the mean valid graphs
, then there exists a unique symmetric NE (Theorems~\ref{dist:existence} and~\ref{nemulti}). Figure~\ref{fig:different partitions} reveals that partition of nodes amongst maximal independent sets need not be unique. We have shown that each such partition leads to a unique symmetric NE (Theorems~\ref{dist:existence} and \ref{nemulti}).  Thus, symmetric NE is not unique. 

A primary would not know the partitions other primaries are selecting since the co-ordination among the primaries is infeasible in a non co-operative game. Theorem~\ref{thm:policyuniqueness}  entails that co-ordination among the players is not required when the independent set selection strategy is of the form (\ref{dist}) and (\ref{ne2}).  We obtain an even stronger result in a special case: we show that there is a unique symmetric NE in a linear conflict graph (Theorem~\ref{uniquelinear}). 
\begin{figure}
\includegraphics[width=90mm,height=40mm]{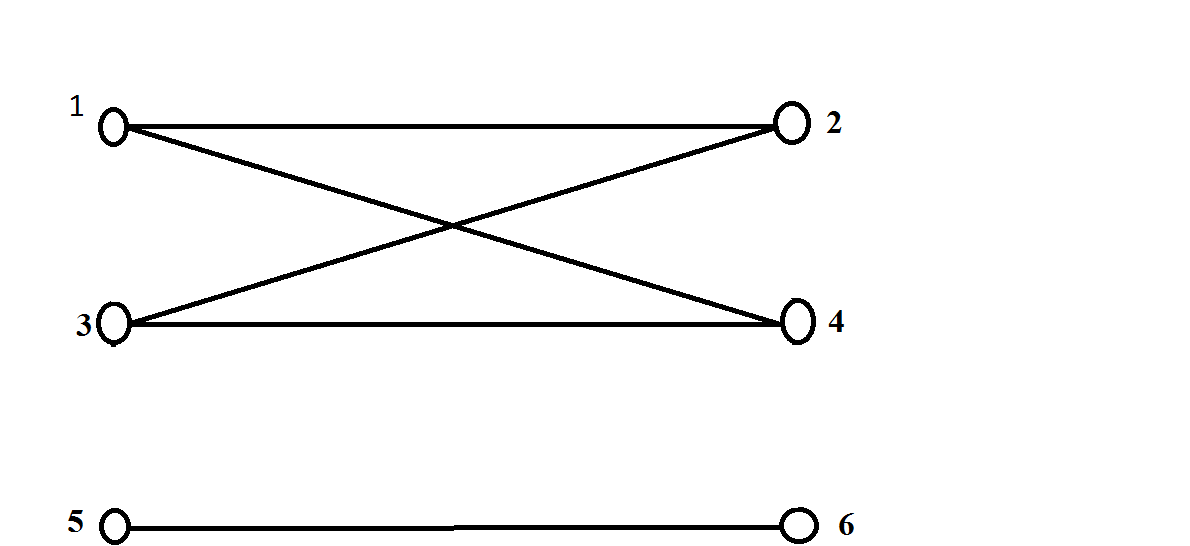}
\caption{\small The above mean valid graph has two different sets of partitions: 1) $I_1=\{1,3,5\}, I_2=\{2,4,6\}$ and 2) $\bar{I}_1=\{1,3,6\},\bar{I}_2=\{2,4,5\}$. If $\alpha_{a,j}$ ($\bar{\alpha}_{a,j}$, respectively) is the node selection probability at node $a$ under NE strategy profile where primaries select $I_1,I_2$ ($\bar{I}_1,\bar{I}_2$, respectively), then according to Theorem~\ref{thm:policyuniqueness}, we obtain $\alpha_{a,j}=\bar{\alpha}_{a,j}$ for all channel states $j$. There exists independent sets which are different from $I_1, I_2$ and $\bar{I}_1,\bar{I}_2$ e.g. $\{1,3\}$, $\{2,4\}$.}
\label{fig:different partitions}
\vspace{-0.5cm}
\end{figure}
 
\begin{thm}\label{thm:policyuniqueness}
Consider that nodes in a mean valid graph can be partitioned into two different sets of maximal independent sets: i) $I_1,\ldots, I_d$, and ii) $\bar{I}_1,\ldots, \bar{I}_{\bar{d}}$. Suppose at channel state $j=1,\ldots,n$,  $0\leq n_j\leq l$ number of primaries select independent sets among $I_1,\ldots, I_d$ and $l-n_j$ number of primaries  select independent sets among $\bar{I}_1,\ldots, \bar{I}_{\bar{d}}$ according to (\ref{dist}) and (\ref{ne2}).  Then the strategy profile constitutes an NE.  

Additionally, let $\alpha_{a,j}$ ($\bar{\alpha}_{a,j}$ respv.) be the probability with which primary $i$ offers its channel at node $a$ at channel state $j$ when it selects independent sets among $I_1,\ldots, I_d$ ($\bar{I}_1,\ldots, \bar{I}_{\bar{d}}$ reps.) such that (\ref{dist}) and (\ref{ne2}) are satisfied, then $\alpha_{a,j}=\bar{\alpha}_{a,j}$.  
\end{thm}
  The first part of the above theorem implies that regardless of the partition other primaries select, a primary can attain its NE strategy profile by selecting independent sets using one of the partition. Hence, a primary needs not co-ordinate with other primaries in order to decide which partition it will choose. Thus, the strategy profile of the form (\ref{dist}) and (\ref{ne2}) is easy to implement.

The second part of the theorem  implies that regardless of the partition primary $i$ selects, the node selection probability will be identical. Thus,  the independent set selection strategies are {\em functionally unique}. Note that when different primaries select independent set selection strategies using different partitions, then the strategy profile is not symmetric, however, the node selection probabilities will be identical.

In Theorem~\ref{thm:policyuniqueness} we show that when primaries select independent sets which belong to a partition, then the symmetric NE will lead to the same node selection probability.  
 But there are independent sets which do not belong to a partition characterizing the mean valid graph (Fig.~\ref{fig:different partitions}). We have not ruled out a symmetric NE which selects an independent set which is outside of a partition characterizing the mean valid graph. We rule this out in the special class of linear conflict graphs. Linear conflict graphs frequently arise in practice: e.g. in the modeling of wireless  access point across a highway or along a row of shops.  

We show\footnote{In a linear conflict graph, the number of independent sets grows exponentially with $M$. Since $I_1$, $I_2$ are not the only independent sets (Fig.~\ref{fig:linear}), thus, it is not apriori clear whether every NE strategy profile only selects independent sets among $I_1,I_2$ with positive probability.} in Appendix~\ref{sec:appendix_uniquelinear}-- 

\begin{thm}\label{uniquelinear}
There exists a unique (not merely functionally unique) symmetric NE strategy profile in a linear conflict graph. In the symmetric NE each primary selects only independent sets $I_1$ and $I_2$, where $I_1$ ($I_2$, respectively) consists of odd (even, respectively) numbered nodes (Fig.~\ref{fig:linear}). 
\end{thm}
\subsubsection{Proof of Theorem~\ref{thm:policyuniqueness}}
First, we provide an outline of the proof.

  Suppose both the partitions 1) $\bar{I}_1,\ldots,\bar{I}_{\bar{d}}$ and 2) $I_1,\ldots, I_d$ characterize a mean valid graph $G$ (i.e. they satisfy conditions 1 and 2 of Definition~\ref{dmvg}). Let $|\bar{I}_s|=\bar{M}_s$ for $s\in \{1,\ldots,\bar{d}\}$ with 
\begin{align}
\bar{M}_1\geq \bar{M}_2\geq \ldots\geq \bar{M}_{\bar{d}}.\nonumber
\end{align}
We show in Appendix~\ref{sec:proof_policyuniqueness}
\begin{lem}\label{thm:equalcardinality}
$M_s=\bar{M}_s$, thus $d=\bar{d}$.
\end{lem}
Thus, $|I_s|=|\bar{I}_s|$, $s\in \{1,\ldots,d\}$ . Since the solution of (\ref{dist}) and (\ref{ne2}) only depend on the cardinalities of $I_s$, hence if a primary selects the partition $\bar{I}_1,\ldots, \bar{I}_d$ then a primary selects independent sets by solving  (\ref{dist}) and (\ref{ne2}). Since the solution of (\ref{dist}) and (\ref{ne2}) is unique by Theorem~\ref{dist:existence}, hence, if $|I_s|=|\bar{I}_k|$, they are selected with identical probability. Thus, if $a\in I_s$, and $a\in \bar{I}_k$ such that $|I_s|=|\bar{I}_k|$ then, the node selection probability at node $a$ at any channel state will be identical. However, if $a\in \bar{I}_k$ and $|\bar{I}_k|\neq |I_s|$, then the node selection probability may be different. We eliminate the above possibility in the following which we also show in Appendix~\ref{sec:proof_policyuniqueness}. 
\begin{lem}\label{thm:disjointunequal}
If $|I_j|\neq |\bar{I}_k|$, then $I_j\cap \bar{I}_k=\Phi$. 
\end{lem}
We have explained the relationship between $I_1,\ldots, I_d$ and $\bar{I}_1,\ldots, \bar{I}_d$ in Fig. \ref{fig:equalpartitions}. The proof of Theorem~\ref{thm:policyuniqueness} readily follows from the fact  that the node selection probability is identical irrespective of the partitions selected by primaries. The detailed proof is given below:

\textit{Proof of Theorem~\ref{thm:policyuniqueness}}: 
First, we show the following: if $\alpha_{a,j}$ ($\bar{\alpha}_{a,j}$ resp.) is the node selection probability when a primary selects among independent sets among  $I_1,\ldots, I_d$ ($\bar{I}_1,\ldots, \bar{I}_{\bar{d}}$ resp.) such that (\ref{dist}) and (\ref{ne2}) are satisfied, then $\alpha_{a,j}=\bar{\alpha}_{a,j}$.  It will essentially prove the second part of the theorem.

Fix a node $a$. Let $a\in I_s$ and $a\in \bar{I}_k$.  By theorem~\ref{dist:existence} there exists a unique solution $t_{j}=(t_{1,j},\ldots,t_{d,j})$ of (\ref{dist}) and (\ref{ne2}). Since $a\in I_s$, thus, 
\begin{align}\label{eq:node}
\alpha_{a,j}=q_jt_{s,j}.
\end{align}

\begin{figure}
\includegraphics[width=150mm,height=60mm]{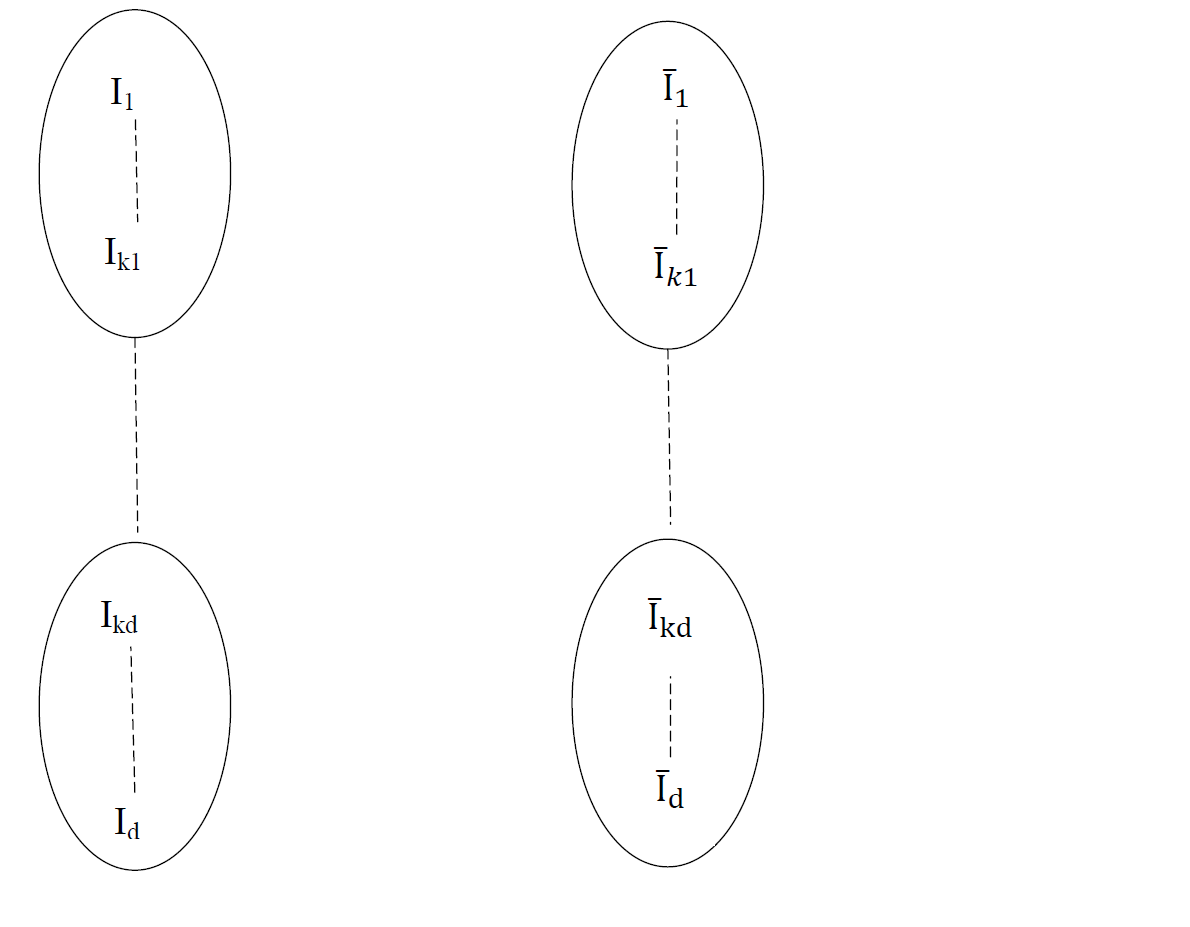}
\caption{\small Independent sets of same cardinality are grouped together. Thus, $|I_1|=\ldots=|I_{k1}|$ . $I_1\cup\ldots \cup I_{k1}=\bar{I}_1\cup \ldots\cup \bar{I}_{k1}$. If node $a$ belongs to $I_1$, then it must belong to $\bar{I}_1\cup \ldots\cup \bar{I}_{k1}$, but it can not belong to in $\bar{I}_s$, $s>k1$. 
}
\label{fig:equalpartitions}
\vspace{-0.3cm}
\end{figure}
Since $\bar{d}=d$ and $|\bar{I}_s|=|I_s|$ for all $s\in \{1,\ldots,d\}$ by Lemma~\ref{thm:equalcardinality}, thus,  (\ref{dist}) and (\ref{ne2}) are identical irrespective of whether a primary selects independent sets among $I_1,\ldots, I_d$ or $\bar{I}_1,\ldots, \bar{I}_{d}$.   Since there exists unique solution of (\ref{ne2}) and (\ref{dist}) (by Theorem~\ref{dist:existence}) , thus $t_{j}$ is the only solution of (\ref{dist}) and (\ref{ne2}). Hence,  probability with which the independent set $\bar{I}_{i}$ is selected at channel state $j$ is $t_{i,j}$. Since node $a\in \bar{I}_k$, thus, from (\ref{charac})
\begin{align}\label{eq:nodebar}
\bar{\alpha}_{a,j}=q_jt_{k,j}.
\end{align} 

So, it is clear that if $s=k$, then $\alpha_{a,j}$ and $\bar{\alpha}_{a,j}$ are identical (by (\ref{eq:node}) and (\ref{eq:nodebar})). Thus, we are only left to show when $s\neq k$ then  (\ref{eq:node}) and (\ref{eq:nodebar}) are equal which we show in the following.
 
By Lemma~\ref{thm:disjointunequal} and \ref{thm:equalcardinality}, we must have $|I_k|=|\bar{I}_k|=|I_s|$ since $a\in I_s$ and $a\in \bar{I}_k$. Since the solution of (\ref{dist}) and (\ref{ne2}) is the unique (by Theorem~\ref{dist:existence}), thus,
\begin{align}
 t_{k,j}=t_{s,j}\nonumber.
 \end{align} Thus, $\alpha_{a,j}$ and $\bar{\alpha}_{a,j}$  are also identical (by (\ref{eq:node}) and (\ref{eq:nodebar})) when $s\neq k$. Hence, we show that $\alpha_{a,j}=\bar{\alpha}_{a,j}$.

Now, we show that if a primary selects independent sets among $I_1,\ldots, I_d$ irrespective of partition the other primaries select such that (\ref{dist}) and (\ref{ne2}) are satisfied, then the strategy profile is an NE. This will conclude the proof since by symmetry, it will follow that if a primary selects independent sets among $\bar{I}_1,\ldots, \bar{I}_d$ irrespective of the partition other primaries select then the strategy profile is an NE.

We have so far showed that every node in $I_s$ is selected with identical probability by each primary irrespective of the partition  selected by them when the independent set selection strategy is of the form (\ref{dist}) and (\ref{ne2}). Thus, at every node $a\in I_s$, each primary offers its channel at node $a$ when the channel state is $j$ or higher w.p. $\sum_{k=j}^{n}\alpha_{a,k}=\sum_{k=j}^{n}q_jt_{s,j}$ which is equal to $\gamma_{s,j}$  (recall from (\ref{recurg})) irrespective of the partition selected by the primaries.    In proving that a primary does not have any incentive to deviate unilaterally from the strategy profile which is of the form (\ref{dist}) and (\ref{ne2}) (Theorem~\ref{nemulti}), we only use the properties of $\gamma_{s,j}$. Hence, if a primary selects independent sets among $I_1,\ldots, I_d$ according to (\ref{dist}) and (\ref{ne2}) irrespective of the partitions selected by other primaries, then it is an NE. Hence, the result follows.    \qed




\section{Different channel states at different locations}\label{sec:diff_channel_states}
  At later stages of deployment, the secondary market will operate at a region consisting of a large number of locations. The channel states will be different at different locations in this large region which we consider in this section.   We first present specific assumptions that we have made in this scenario (Section~\ref{sec:specific_assumptions}). For example, nodes of commonly observed large conflict graphs exhibit an inherent symmetry in the interference relations, we therefore consider a class of conflict graphs, known as {\em node symmetric graphs} in the literature \cite{nodetransitive}.   We subsequently obtain a symmetric NE strategy profile $SP_{sym}$  in a node symmetric graph (Section~\ref{sec:spsym}). We show some important structural properties of $SP_{sym}$ which are significantly different from the symmetric NE strategy profile obtained in the scenario where the channel state is identical across the network (Theorem~\ref{thm:nodesymmetricne}, Lemmas~\ref{lm:payoff_same}, and \ref{thm:notsame}). We show that $SP_{sym}$ may not be a NE when the conflict graph is not a node symmetric (Lemma~\ref{thm:notane}). Finally, we analytically and empirically evaluate the computational issues of computing the strategy $SP_{sym}$ and how a primary can attain a desired trade-off between the expected payoff and the computational cost by selective estimation of channel states at randomly selected subset of nodes  (Section~\ref{sec:computation}). 
  


\subsection{Specific Assumptions}\label{sec:specific_assumptions}
We revert to the notations  introduced in Sections II and III. Specifically, we do not need simplifications of the notations used for the first setting which have been introduced in Section~\ref{sec:modification}.

  \subsubsection {$n=1$}  In the previous setting (Section~\ref{sec:same_channel_state}) we consider that the channel state is the same across the locations, thus,  a primary always selects an independent set from the conflict graph $G$ whenever the channel is available (i.e. the channel is not in state $0$). Thus, a primary knows that its competitors always select independent sets from $G$  (a primary does not select any independent set when the channel state is $0$).  In the current setting, the conflict graph representation of the region depends on the channel state vectors. Since the conflict graph representation can be different for different channel state vectors ($G_J$ may not be equal to $G_{K}$ when $J\neq K$), thus, a primary does not know  the conflict graphs from which its competitors are selecting their independent sets. Thus, the collection of independent sets from which a primary selects its independent set may be different for different primaries.  Additionally, the strategy space $\mathcal{P}$ ($|\mathcal{P}|=(n+1)^{|V|}-1$) increase exponentially with the number of nodes.   Thus, obtaining an NE in this setting is challenging. In order to simplify the setting, we consider  
\begin{assum}
$n=1$ i.e. the channel is either available (i.e. at state $1$) or not available (i.e. at state $0$) at each node, but still the channel state can be different at different nodes. 
\end{assum}
Note that even though $n=1$, the cardinality of strategy space $\mathcal{P}$  is $2^{V}-1$ which is still exponential in the number of nodes and the conflict graph representation will be different for different channel state vectors. 

\begin{defn}\label{defn:alpha_a}
Since $n=1$, we drop the index $j$ from $\alpha_{a,j}$ and $\mathcal{P}_{a,j}$ in (\ref{defn:nodeprob_general}) corresponding to the channel state at a given location. We denote $\alpha_{a}$ as the probability with which an available channel at node $a$ is offered  under a symmetric strategy profile and $\mathcal{P}_a$ as the set of channel state vectors where the channel state is $1$ at node $a$.
\end{defn} 
Note that from Theorem~\ref{singlelocation} and (\ref{n51}), the upper endpoint of the penalty selection strategy is $v$ at all nodes. The maximum expected payoff of a primary at  node $a$  under a symmetric NE strategy is
\begin{align}\label{eq:ex_paya1s}
p_{a}-c=(f_1(v)-c)(1-w(\alpha_a))
\end{align}

\subsubsection{Node Symmetric Graphs}\label{sec:nodesymmetricgraphs}
 We consider large size wireless networks. As an analytical abstraction, we mainly consider infinite size conflict graphs.  In large conflict graphs,   there is an inherent symmetry in the interference relations among the nodes in the network. We, therefore consider node symmetric graphs, which in the literature is also known as {\em node transitive graphs} \cite{nodetransitive}. 

First, we provide a formal definition of node symmetric graph. Towards that end, we first define an {\em automorphism} in a conflict graph $G$. We denote $V(G)$ as the set of nodes of $G$.
\begin{defn}\label{defn:automorphism}
An automorphism is a bijective mapping $F: V(G)\rightarrow V(G)$ such that nodes $F(a)$ and $F(b)$ are adjacent\footnote{In an undirected graph, two nodes are adjacent iff there is an edge between them.} if and only if nodes $a$, $b$ are adjacent in $G$.
\end{defn}
In an automorphism, the nodes are renumbered such that it maintains the adjacency between the nodes. For example, consider a linear graph consisting of $4$ nodes. Fig.~\ref{fig:automorphism} shows an automorphism on this graph. Now we are ready to define the node symmetric graph.
\begin{figure}
\includegraphics[width=120mm,height=20mm]{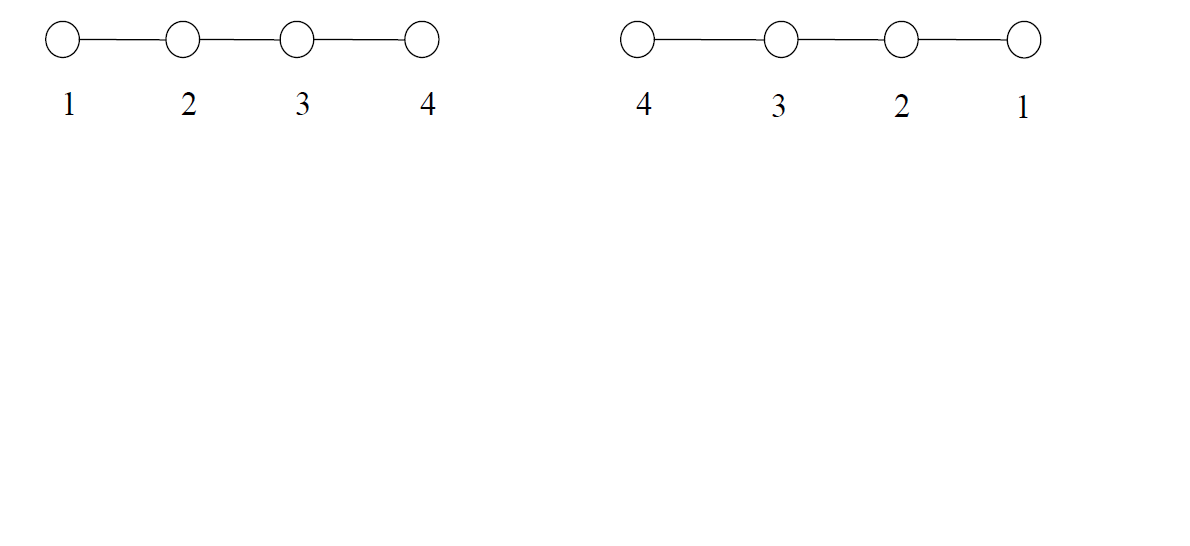}
\vspace{-0.8cm}
\caption{\small Left hand figure shows a linear graph with $4$ nodes. Right hand figure an automorphism where $F(1)=4, F(2)=3, F(3)=2, F(4)=1$. However, there is no automorphism between nodes $2$ and $1$. If there is an automorphism such that $F(2)=1$, then by the property of automorphism nodes $F(2)$ and $F(3)$ should be adjacent and nodes $F(2)$ and $F(1)$ also should be adjacent, but since $F(2)=1$, thus, either $F(3)$ or $F(1)$ will not be adjacent to node $F(2)$ since node $1$ only has one degree in $G$. }
\label{fig:automorphism}
\vspace{-0.3cm}
\end{figure}
\begin{defn}\label{defn:node_symmetric}\cite{nodetransitive}
In a node symmetric graph, for every pair of vertices $a$ and $b$ of $G$, there is some automorphism $F:V(G)\rightarrow V(G)$  such that $F(a)=b$.
\end{defn}
For a graph to be node symmetric every node should be mapped to every other node through an automorphism. Informally,   in a node symmetric graph the graphs looks the same from each node. 

  For example cyclic graph is a node symmetric graph.  But linear graph with $4$ nodes is not a node symmetric graph since there is no automorphism between nodes $1$ and $2$ (Fig.~\ref{fig:automorphism}).

Now, we provide some examples of infinite node symmetric graphs which resemble the conflict graphs of large wireless networks.
 \begin{itemize}
 \item Infinite linear graph with no end points (Fig.~\ref{fig:linear_infinity}):   This is an abstraction of the conflict graph of the network of a large number of wireless access points arranged in a linear fashion. 

\item Infinite square graphs (Fig.~\ref{fig:random_square}):   This is an abstraction of the conflict graph of wireless networks in a large region with square cells.

\item Infinite grid graphs (Fig.~\ref{fig:random_grid}): This is an abstraction of the conflict graph of a large shopping mall. 

\item Infinite triangular graphs (Fig.~\ref{fig:random_hexacell}): This  is an abstraction of  the conflict graph representation of large number of  hexagonal cells \cite{matula}.
 \end{itemize}
 There are also several commonly observed node symmetric conflict graphs which are finite. For example, cyclic graph of any size is a node symmetric graph \footnote{Note that cyclic graph is not a mean valid graph if $|V|>3$ and $|V|$ is odd, thus, node symmetric graphs may not be mean valid graphs.}. Cyclic conflict graph represents a collection of wireless access points arranged in a circular fashion, possibly circumambulating a city or ring size road.   Figure~\ref{fig:kregular} also shows a finite node symmetric graph and the corresponding wireless network. The complete graphs \footnote{In a complete graph a node has edge with every other node.} are also node symmetric graphs. We find a symmetric NE in a node symmetric graph irrespective of whether it is finite or infinite (Theorem~\ref{thm:nodesymmetricne}). 
 
   Note from Section~\ref{sec:meanvalidgraph} that the commonly observed conflict graphs of small networks are mean valid graphs which we analyze in the previous setting.  {\em These graphs may not be node symmetric graphs}.

\begin{figure}
\begin{subfigure}[]
{
\includegraphics[width=90mm,height=30mm]{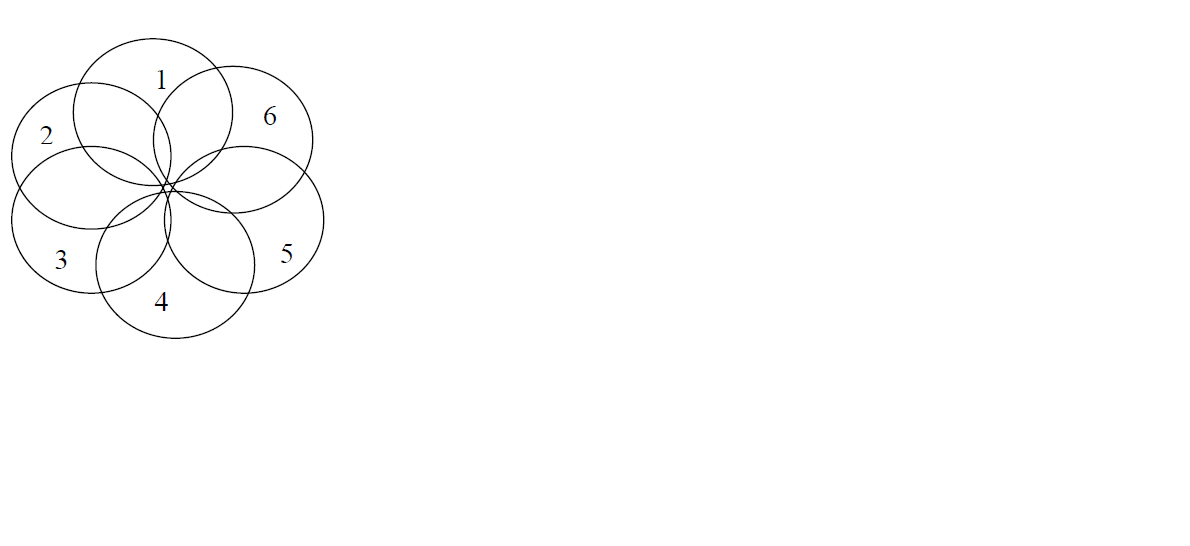}}
\end{subfigure}
\label{fig:linear_region}
\begin{subfigure}[]
{\includegraphics[width=90mm,height=30mm]{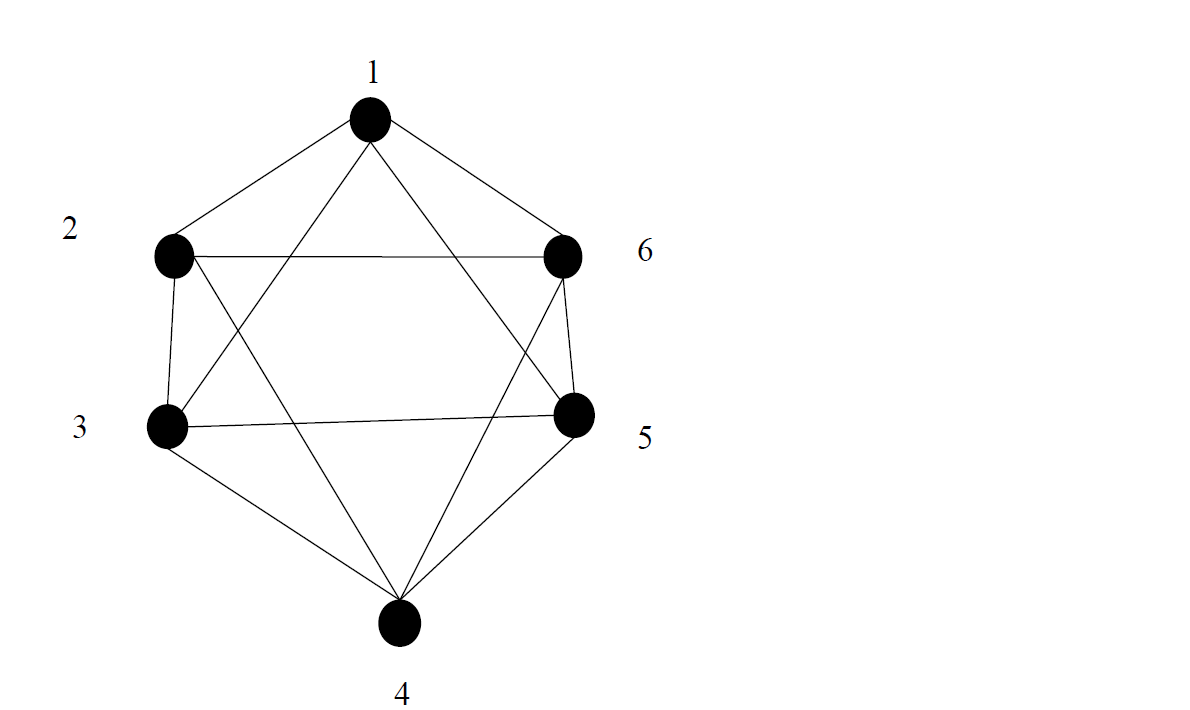}
\vspace{-0.7cm}}
\end{subfigure}
\caption{\small Circles in Figure (a) shows the ranges of the wireless access points located at the center of the circle. Figure (b) shows the corresponding conflict graph with each circle is represented as a node. Each circle intersects with $1$ hop and $2$ neighbors, thus, in the conflict graph each node has edges with $1$ hop and $2$ hop neighbors. The conflict graph is a node symmetric graph. }
\label{fig:kregular}
\vspace*{-0.2cm}
\end{figure}
\begin{figure}
\includegraphics[width=120mm,height=20mm]{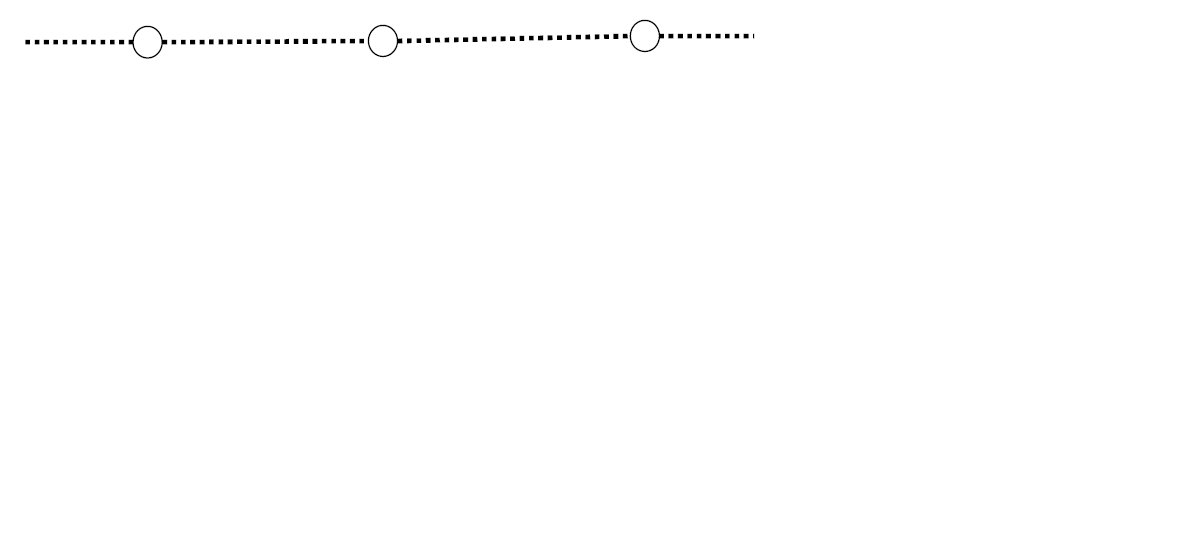}
\vspace*{-1cm}
\caption{\small Infinite linear graph with no end-points: each node has degree $2$.}
\label{fig:linear_infinity}
\vspace{-0.3cm}
\end{figure}
\begin{figure*}
\begin{minipage}{.32\linewidth}
\begin{center}
\includegraphics[width=50mm,height=30mm]{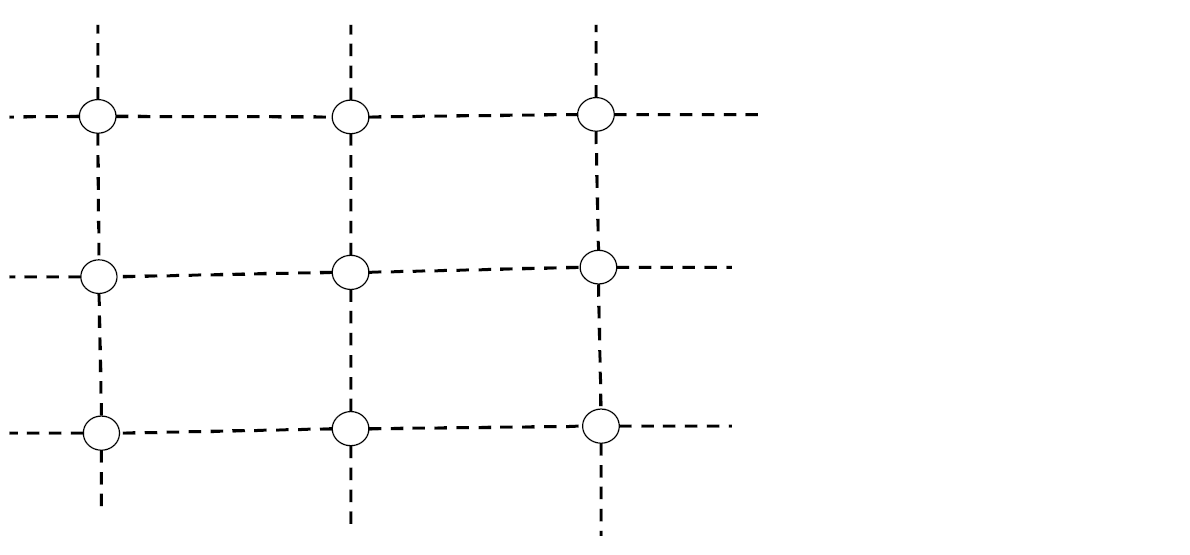}
\vspace*{-.1cm}
\caption{\small Infinite square graph: each node has degree $4$.}
\label{fig:random_square}
\end{center}
\end{minipage}\hfill
\begin{minipage}{.32\linewidth}
\begin{center}
\includegraphics[width=50mm,height=30mm]{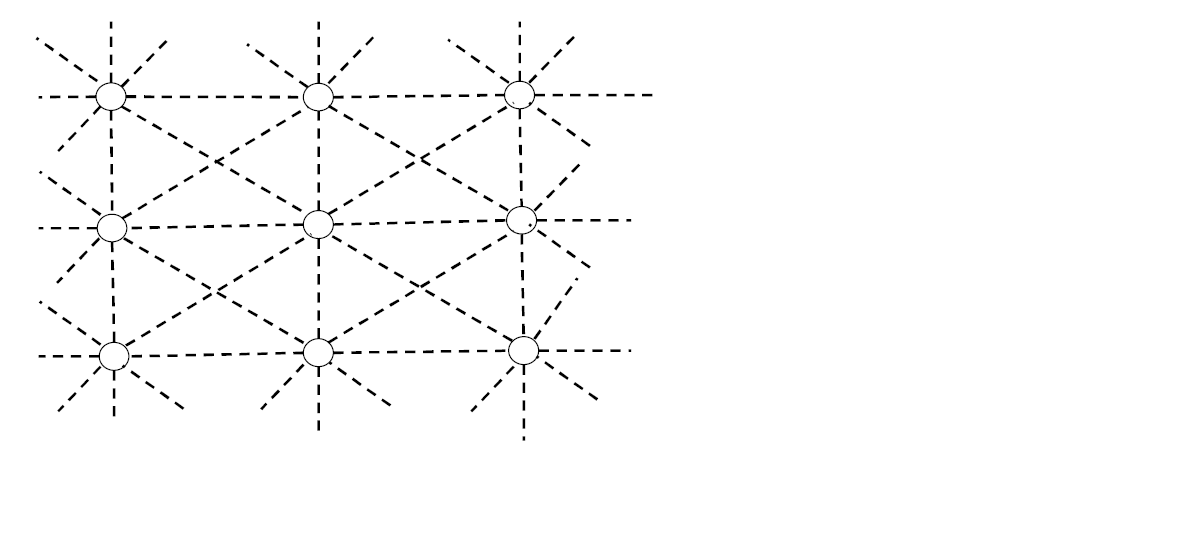}
\vspace*{-.1cm}
\caption{\small Infinite Grid graph: each node has degree $8$. }
\label{fig:random_grid}
\end{center}
\end{minipage}\hfill
\begin{minipage}{.32\linewidth}
\begin{center}
\includegraphics[width=50mm,height=30mm]{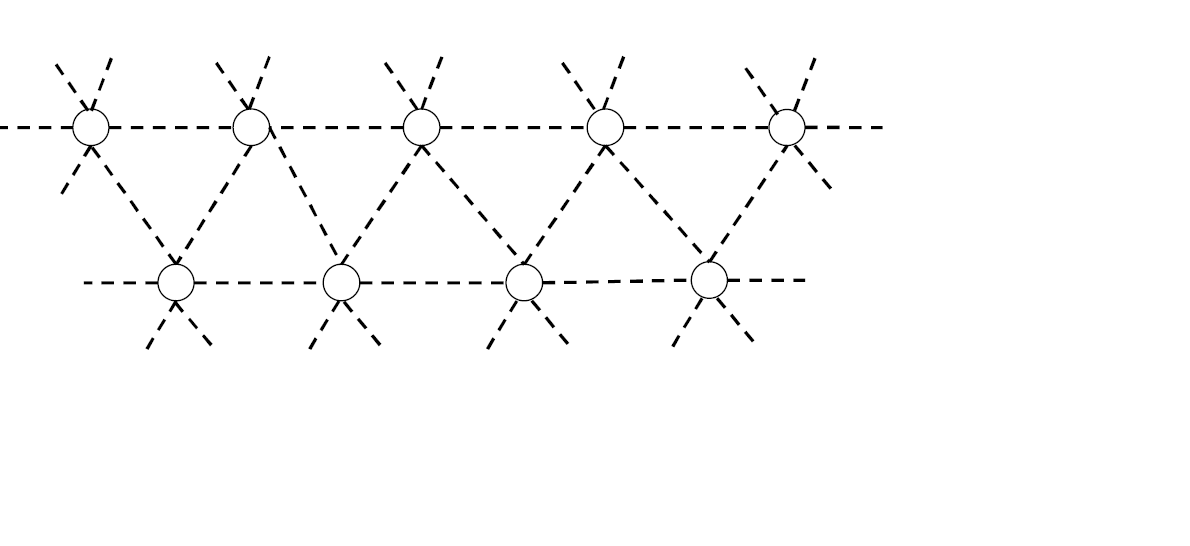}
\vspace*{-.1cm}
\caption{\small Infinite Triangular Graph: each node has degree $6$.}
\label{fig:random_hexacell}
\end{center}
\end{minipage}
\vspace*{-.6cm}
\end{figure*}
\subsubsection{Probability Distribution of Channel State Vectors}\label{sec:channel_statistic}
In the previous setting, we consider an extreme case where the channel state is identical across each location. In a large network, the channel states will be different. However, the  channel states are often spatially proximal.  Since the graph is large,  like the interference relationship we expect that the statistical correlation pattern would also exhibit some symmetry.  We consider one such symmetric relationship  among the channel states across the network which arise naturally.

First, we define an isomorphism between two graphs: 
\begin{defn}\label{defn:isomorphism}
Two graphs $G$ and $H$ are isomorphic to each other if there is a bijective mapping $F:V(G)\rightarrow V(H)$ such that any two vertices $F(a), F(b)$ are adjacent in $H$ if and only if $a, b$ are adjacent in $G$.
\end{defn}
Informally, if two graphs look alike subject to renumbering of nodes, then they are isomorphic to each other. Note that automorphism is a special case of isomorphism which occurs when $H=G$ (Definition~\ref{defn:automorphism}). 

We assume that 
\begin{assum}\label{assum:prob}
$q_J$ and $q_K$ are identical whenever the $G_{J}$ and $G_{K}$ are isomorphic to each other.
\end{assum}
Intuitively, since $G_J$ and $G_K$ are alike subject to the renumbering of nodes, we therefore expect $q_J=q_K$.  We show in Section~\ref{sec:same_prob} that the above assumption implies that 
\begin{lem}\label{lm:identicalprob}
The probability that a channel of a primary is in state $1$ at a given location is the same across the network.
\end{lem} However, the converse of the above result not true in general.

We now provide some examples of joint probability distributions which arise in practice and satisfy Assumption~\ref{assum:prob}. 

{\em Independent and identically distributed channel states}:
 The state of the channel is $i=1$ w.p. $q$ at a given location independent of the channel states at other locations. At a given channel state vector $J$, if the channel is available at $n_j$ number of nodes, then $q_J=q^{n_j}(1-q)^{|V|-n_j}$. When $G_{J}$ and $G_{K}$ are isomorphic, then both contain the same number of nodes, thus, the number of locations where the channel state is $1$ ($0$, respv.)  are the same  in channel state vectors $J$ and $K$. Hence, the probability distributions $q_J$ and $q_K$ are identical whenever $G_{J}$ and $G_{K}$ are isomorphic. 

{\em Correlated Channel states}: We now show that Assumption~\ref{assum:prob} can accommodate statistical correlations across the channel states at different nodes. We provide an example in a small node symmetric graph. Consider a linear graph with $2$ nodes such that  $q_{(1,0)}=q_{(0,1)}$. Since $G_{(0,1)}$ and $G_{(1,0)}$ are the only possible isomorphic graphs in this case, thus, the above joint probability distribution satisfies Assumption~\ref{assum:prob}.  Now, if 
$q_{(1,1)}> q_{(1,0)}=q_{(0,1)}$ and $q_{(0,0)}>q_{(1,0)}=q_{(0,1)}$, then, the channel states are not independent\footnote{Suppose that the channel is in state $1$ at node $i$ w.p. $q_i$ independent of the channel state at other location, then, $q_{1,0}=q_{0,1}$ implies that $q_1=q_2$. Now, $q_1(1-q_1)$ can not be less than both $q_1^2 (=q_{1,1})$ \& $(1-q_1)^2(=q_{0,0})$, hence,  independent channel states can not satisfy the above joint distribution.}. Thus, Assumption~\ref{assum:prob} allows correlation among the channel states across the locations. Also note that the above probability distributions commonly arise in practice.   This is because when the channel is in state $1$ ($0$ respv.) at one location, then there is a higher  probability that the channel is in state $1$ ($0$ respv.)  compared to state $0$ ($1$ respv.) at other location. 

The joint probability distributions of random variables associated with spatial locations and exhibiting correlations are often represented as Markov Random Field.  We provide a formal definition of  Markov random field  in Appendix~\ref{sec:mrf} and show  that the Markov random field modeling of channel states where the channel states in neighboring locations are correlated,  satisfy Assumption~\ref{assum:prob} under some additional assumptions which naturally arise (Lemma~\ref{lm:mrf_sameprob} in Appendix~\ref{sec:mrf}). 

\subsubsection{Proof of Lemma~\ref{lm:identicalprob}}\label{sec:same_prob}
We first show Observation~\ref{obs:sameisomorphicpairs}. Subsequently, we show Lemma~\ref{lm:identicalprob}.

\begin{obs}\label{obs:sameisomorphicpairs}
Consider any pair of nodes $a$ and $b$. For distinct channel state vectors $J,J_1\in \mathcal{P}_a$ (Definition~\ref{defn:alpha_a}), there exists distinct channel state vectors $K, K_1\in \mathcal{P}_b$ such that $G_{K}$ and $G_{K_1}$ are isomorphic to $G_{J}$ and $G_{J_1}$ respectively, such that the in the isomorphic function $F(a)=b$ (Definition~\ref{defn:isomorphism}). 
\end{obs}
\begin{proof}
We first show that for a channel state vector $J\in \mathcal{P}_a$ there exists a channel state vector $K\in \mathcal{P}_b$ such that $G_{K}$ is isomorphic to $G_{J}$ and in the isomorphic function $F(a)=b$. Since the graph is node symmetric, thus, there exists an automorphism $F(\cdot)$ (Definition~\ref{defn:automorphism}) such that $F(a)=b$.   Now consider the channel state vector $K$ where the channel is available only at nodes $F(a_1)$ if $a_1\in V(G_J)$. In the conflict graph representation of $G_{K}$, the set of edges are the edges incident on $F(a_1)$ where $a_1\in G_J$. Since $F(\cdot)$ is itself is an automorphism on $G$, thus  $F(a_1)$ and $F(a_2)$ are adjacent in $G_K$ if and only if $a_1$ and $a_2$ are adjacent in $G_J$. Hence, $F(\cdot)$ is an isomorphic mapping from $V(G_J)$ to $V(G_K)$ such that $F(a)=b$.   

Note that since $J$ is arbitrary, thus, if $J_1\in \mathcal{P}_a$, then, following the above procedure we obtain an isomorphic graph $G_{K_1}$ such that $F(a)=b$ in $G_{K_1}$. Since $F(\cdot)$ is automorphism and thus, $F(\cdot)$ is bijective. Thus, if  $J_1\neq J$, then, using the function $F(\cdot)$ we obtain a  channel state vector $K_1$ which is different from $K$. Also note that  $b\in V(G_{K_1})$ since $F(a)=b$ and $a\in V(G_{J_1})$. Hence, the result follows.
\end{proof}
Now, we show Lemma~\ref{lm:identicalprob}.
\begin{proof}
Consider any two nodes $a$ and $b$. Recall the definition of $\mathcal{P}_a$ (Definition \ref{defn:alpha_a}). First, note that $|\mathcal{P}_a|=|\mathcal{P}_b|=2^{|V|-1}$ since the channel state must be $1$ at node $a$ (node $b$, respv.) for every channel state vector in $\mathcal{P}_a$ ($\mathcal{P}_b$ respv.). Now, the probability that the channel state is $1$ at node $a$ is
\begin{align}
\beta_a=\sum_{J:J\in \mathcal{P}_a}q_J\nonumber
\end{align}
and the probability that the channel state is $1$ at node $b$ is
\begin{align}
\beta_b=\sum_{K:K\in \mathcal{P}_b}q_K\nonumber
\end{align}
Note that by Observation~\ref{obs:sameisomorphicpairs}, for distinct channel state vectors $J, J_1\in \mathcal{P}_a$ there exist distinct channel state vectors $K,K_1\in \mathcal{P}_b$ such that $G_{K}$, $G_{K_1}$ are isomorphic to $G_{J}$ and $G_{J_1}$ respectively. Also note that cardinalities of $\mathcal{P}_a$ and $\mathcal{P}_b$ are the same. Since $q_J=q_K$ whenever $G_J$ and $G_K$ are isomorphic to each other, thus, we obtain
\begin{align}
\beta_a=\beta_b
\end{align}
Hence, the result follows.
\end{proof}

\subsection{Symmetric NE strategy Profile}\label{sec:spsym}
We, first, obtain a symmetric NE strategy profile  (Theorem~\ref{thm:nodesymmetricne}). We then show that the NE strategy has an important structural difference compared to the NE strategy in the previous setting (Section~\ref{sec:same_channel_state}) where the channel state is the same across the network (Lemmas~\ref{lm:payoff_same}, \ref{thm:notsame}).

We first start with introducing a notation. Let $I_{max,J}$ be the set of maximum independent sets (i.e. independent sets of highest cardinalities) of  the graph $G_{J}$ .

{\em Strategy Profile ($\mathrm{SP}_{sym}$)}:  A primary selects each of the independent set within the set $I_{max,J}$ with probability $\dfrac{1}{|I_{max,J}|}$ and select other independent sets with probability $0$ at channel state vector $J$.

\begin{thm}\label{thm:nodesymmetricne}
The Strategy profile $\mathrm{SP}_{sym}$ is an NE strategy profile.
\end{thm}
A primary only needs to find the maximum independent sets in order to find the NE strategy profile $SP_{sym}$.  In contrast to $SP_{sym}$, a primary may select an independent set which is not a maximum independent set  in the scenario where the channel state is identical across the locations (Theorems~\ref{thm:structure} and ~\ref{nemulti}).  Note that in $SP_{sym}$ a primary puts equal weight on each of the maximum independent sets in $G_{J}$. Hence, a primary {\em needs not communicate} with other primaries to obtain its strategy. Hence, $SP_{sym}$ is easy to implement.

\begin{lem}\label{lm:payoff_same}
Expected payoff at every node is the same under $SP_{sym}$.
\end{lem}
 Intuitively, since the graph is node symmetric, each node belongs to the same number of maximum independent sets, thus a channel is offered with the same probability at every node under $SP_{sym}$; thus, the expected payoff is the same at every node. {\em Since each node is selected with the same probability, hence there is an equity in secondary access of the channel amongst different nodes as opposed to that scenario where the channel state is identical across the network (Example 1)}. 

We show that  unlike in the scenario where the channel state is the same across the network (Theorem~\ref{uniquelinear}), the symmetric NE may not be unique in linear conflict graph in this setting.  
 \begin{lem}\label{thm:notsame}
There may exist infinitely many symmetric NEs in the linear conflict graph.
\end{lem}
The proof of the above lemma is algebraic and we relegate it to Appendix~\ref{sec:appendix_lineargraph}.


We also show in Appendix~\ref{sec:appendix_lineargraph} that symmetry in interference relations among the nodes is required for $SP_{sym}$ to be an NE.
\begin{lem}\label{thm:notane}
$SP_{sym}$ may not be an NE for a finite linear graph which is not a node symmetric graph. 
\end{lem}



\subsubsection{Proof of Theorem~\ref{thm:nodesymmetricne}} \label{sec:appendix_nodesymmetricne}

We use Observation~\ref{obs:sameisomorphicpairs} stated in previous subsection (Section~\ref{sec:same_prob}).
Since the strategy profile is symmetric, it is enough to  show that primary $1$ does not have any incentive to deviate from $SP_{sym}$ when other primaries also select $SP_{sym}$. 

We first give an outline of the proof. First, we show that the maximum expected payoff attainable by primary $1$ is identical across the nodes using the node symmetric property and Assumption~\ref{assum:prob}.  Thus, it directly implies that primary $1$ will attain the maximum expected payoff by selecting a maximum independent set. Since $SP_{sym}$ only randomizes among the maximum independent sets, thus, primary $1$ will not have any incentive to deviate from $SP_{sym}$ which in turn proves Theorem~\ref{thm:policyuniqueness}. The details of the proof is given below.

In order to show Theorem~\ref{thm:nodesymmetricne} we show the following:

i) First, we show that the node selection probability $\alpha_{a}$ for a primary is identical when Assumption~\ref{assum:prob} is satisfied for each node under $SP_{sym}$ using Node symmetric graph and Observation~\ref{obs:sameisomorphicpairs}.

ii) Subsequently, we show that when all primaries other than primary $1$ select $SP_{sym}$, then the maximum expected payoff obtained by primary $1$ is identical across the nodes. 

iii) Finally, we show that primary $1$ does not have any incentive to deviate unilaterally from  $SP_{sym}$ which shows that $SP_{sym}$ is indeed an NE. 

 Part i): First, we introduce some notations. Let $I^{a}_{max,J}$ be the set of maximum independent sets of $G_{J}$ which contains node $a$. Note that the node $a$ can only be selected at a channel state vector $J$ if $J\in \mathcal{P}_a$ (Definition~\ref{defn:alpha_a}). 

Thus, under the strategy profile $SP_{sym}$ the node selection probability at node $a$ i.e. $\alpha_{a}$ is 
\begin{align}\label{eq:samenode_nodesym}
\alpha_{a}=\sum_{J\in \mathcal{P}_a}\dfrac{|I^{a}_{max, J}|}{|I_{max,J}|}q_J
 \end{align} 
Now, we show that $\alpha_a=\alpha_b$ where $b\neq a$. By Observation~\ref{obs:sameisomorphicpairs} for every $G_{J}$, there exists a distinct $G_{K}$ which is isomorphic to $G_{J}$ such that  in the isomorphic mapping $F(a)=b$.  Thus, $|I^{b}_{max,K}|=|I^{a}_{max,J}|$. Since $G_{J}$ and $G_{K}$ are isomorphic to each other thus $|I_{max,J}|=|I_{max,K}|$. Also note that, $q_J=q_K$ since $G_{J}$ is isomorphic to $G_{K}$ by Assumption~\ref{assum:prob}. Finally, note that the cardinalities of $\mathcal{P}_a$ and $\mathcal{P}_b$ are the same. Hence, $\alpha_{a}=\alpha_{b}$ for any two nodes $a,b\in V$ by (\ref{eq:samenode_nodesym}). Hence, the node selection probability is the same at every node. 

Part ii): When all the other primaries apart from primary $1$ selects $SP_{sym}$, then at node $a$ , the channel is offered for sale at node $a$ w.p. $\alpha_a$ by other primaries apart from primary $1$. Thus, by Theorem~\ref{singlelocation}, when all the other primaries select $SP_{sym}$, then the maximum  expected payoff obtained by primary $1$ at node $a$ is $(f_1(v)-c)(1-w(\alpha_a))$ (from (\ref{eq:ex_paya1s})). Moreover, by Theorem~\ref{singlelocation} the payoff is obtained by selecting any penalty within $[L_1,v]$.  Since $\alpha_a$\rq{}s are identical, hence, the maximum attainable expected payoff by primary $1$ is identical at each node.

Part iii): Consider a channel state vector $J$. Since the maximum attainable expected payoff at every node is identical, hence, primary $1$ can attain the total maximum expected payoff only by selecting a maximum independent set of $G_{J}$ when other primaries select the strategy $SP_{sym}$. Under $SP_{sym}$, primary $1$ randomizes among the maximum sets of $G_{J}$. Hence, the total expected payoff of primary $1$ is equal to the maximum expected payoff. Hence, primary $1$ does not have any incentive to deviate from $SP_{sym}$ when other primaries  select $SP_{sym}$. Thus, $SP_{sym}$ is an NE. \qed

\subsubsection{Proof of Lemma~\ref{lm:payoff_same}} Note that the proof of this result directly follows from part (ii) of the Theorem~\ref{thm:nodesymmetricne} where we have shown that primary $1$ will attain the same expected payoff at every node of the conflict graph if primary $1$ selects  $SP_{sym}$ when the other primaries also select strategy $SP_{sym}$.\qed

\subsection{Computational Complexity}\label{sec:computation}
In $SP_{sym}$ a primary needs to enumerate the maximum independent sets at a given channel state vector. In general, the number of maximum independent sets scales exponentially with the number of nodes.  But if a graph consists of disjoint components, then a primary can compute the maximum independent sets of each component and compute the strategy profile in each component  in parallel.   Hence, the size of the component will govern the computation time. 

 The conflict graph of a primary depends on the channel state vector which evolves randomly. Hence, the conflict graphs are random graphs.   Thus, it is important to find the average size of a component in a conflict graph which will govern the average computation time of maximum independent sets. In the following, we provide a bound on the expected  size of a component for some node symmetric graphs that arise in practice.   We also discuss how primaries can govern the component size using random sampling technique (selecting each node w.p. $p$) . Throughout this section, we consider that the channel states are I.I.D.  where the channel state is $1$ at a given location w.p. $q$. 

Let $\Delta$ be the degree of a node. We consider those node symmetric graphs where  $\Delta$ is finite. Nodes in most of the conflict graphs that we have discussed in Section~\ref{sec:nodesymmetricgraphs} have finite degrees. For example, in cyclic graph (of any size) $\Delta=2$, in infinite linear graph $\Delta=2$, in infinite square graph (Fig.~\ref{fig:linear_infinity}), $\Delta=4$ (Fig.~\ref{fig:random_square}), in infinite grid graph $\Delta=8$ (Fig.~\ref{fig:random_grid}), in infinite triangular graph $\Delta=6$ (Fig.~\ref{fig:random_hexacell}). 

We  find out the expected size of a component $C$ originating from node $a$ in a conflict graph $G_{J}$. Since the graph is a node symmetric graph, hence the expected size of a component originating from any other node will be the same.  Each node has an expected degree of $q\Delta$.   The component $C$ grows when $G_{J}$ contains neighbors of node $a$, the neighbors of the neighbors of node $a$ and so on. Thus, the growth of $C$  can be compared to the Galton-Watson branching process \cite{galton-watson} where each individual gives birth to $q\Delta$ number of children on average.  The only difference in our approach to the Galton-Watson process is that the number of nodes added each step may be smaller as some of the neighbors of a node may already be in $C$, thus, reducing the number of neighbors that can be added in $C$. Thus, the expected size of $C$ can be upper bounded by the expected umber of total descendants in Galton-Watson process \cite{galton-watson}. Hence,  the upper bound of expected size of $C$ is obtained from \cite{galton-watson}
\begin{lem}\label{lm:exp_comp}
$E(C)\leq \dfrac{1}{1-q\Delta}$ if $q\Delta<1$. 
\end{lem}

A primary can not control $q$, hence, the component size (and thus, the computational complexity) can be large for higher $q$. Thus, a primary may estimate its channel quality only at a subset of the locations of the region instead of the whole region  and sell its channel only among the locations where it knows the transmission quality in order to minimize the computation cost. Equivalently, a primary will consider that the channel state is $0$ at locations where it does not estimate its channel quality. In one simplistic setting which we consider,  each primary computes the transmission quality at a node w.p. $p$ independent of the other nodes.  A primary does not know the nodes where its competitors are estimating their channel states. But , it knows $p$. Thus, a primary is aware of the fact that the channel state is $1$ at any given node of its competitor w.p. $pq$ independent of the channel states at other locations. Thus, the probability distribution satisfies Assumption~\ref{assum:prob}. {\em Hence, the strategy profile $SP_{sym}$ will be a symmetric NE strategy profile in this setting where the channel states are I.I.D. and  the channel is in state $1$  at a given location w.p. $pq$  instead of $q$.}   Thus, from Lemma~\ref{lm:exp_comp} the expected size of component is now upper bounded by
\begin{align}\label{eq:exp_component}
E(C)\leq \dfrac{1}{1-pq\Delta} \quad \text{if } pq\Delta<1
\end{align}
Note that the above procedure also decreases the measurement and estimation cost, since a primary only estimates the channel states at a randomly selected subset of locations.  


Note that the right hand side in (\ref{eq:exp_component}) increases as $pq\Delta$ increases. If a primary selects lower (higher, respv.) $p$ the expected component size will decrease (increase, respv.), and thus, the computation complexity will decrease (increase, respv.); The bound in (\ref{eq:exp_component}) also decreases (increases, respv.). However, the expected payoff of a primary  will also decrease (increase, respv.)  with decrease in $p$ (increase, respv.) since the number of nodes where a primary can potentially sell its channel also decreases (increases, respv).  Hence, a primary has to judiciously select $p$ in order to achieve a desired   trade-off between the computation complexity, and the expected payoff. 
\begin{figure*}
\begin{minipage}{.49\linewidth}
\begin{center}
\includegraphics[width=90mm,height=40mm]{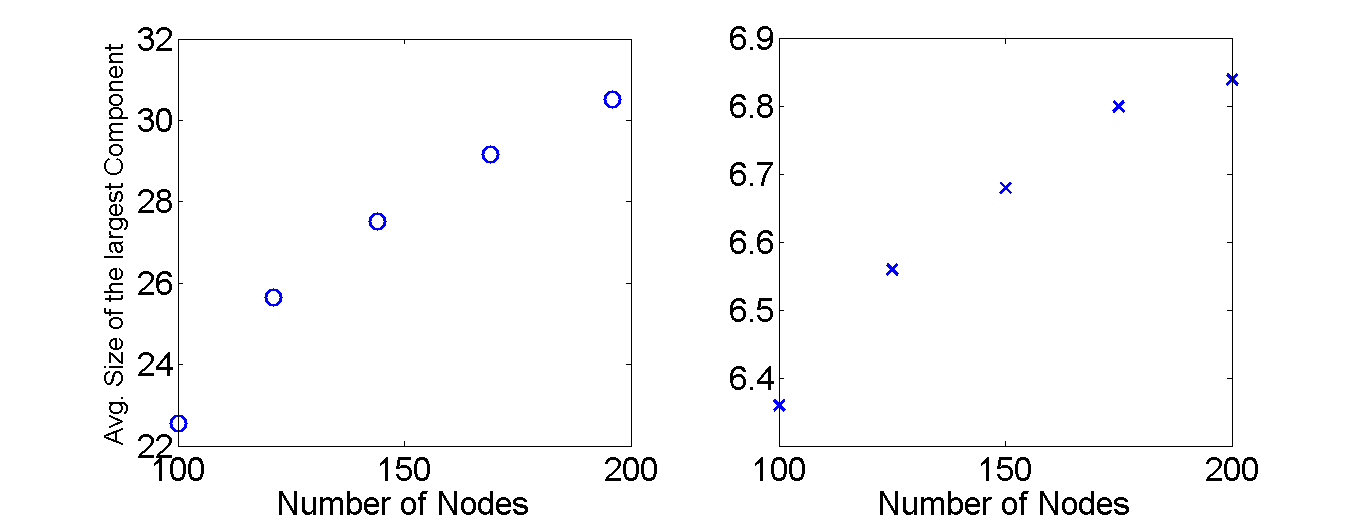}
\caption{\small Mean size of the largest component in a Square Graph and Linear graph  for $pq=0.5$. The square graph consists of a $j$ rows and columns. We vary $j$ while we increase the number of nodes.}
\label{fig:component_nodes}
\vspace{-0.3cm}
\end{center}
\end{minipage}\hfill
\begin{minipage}{0.49\linewidth}
\begin{center}
\includegraphics[width=90mm, height=40mm]{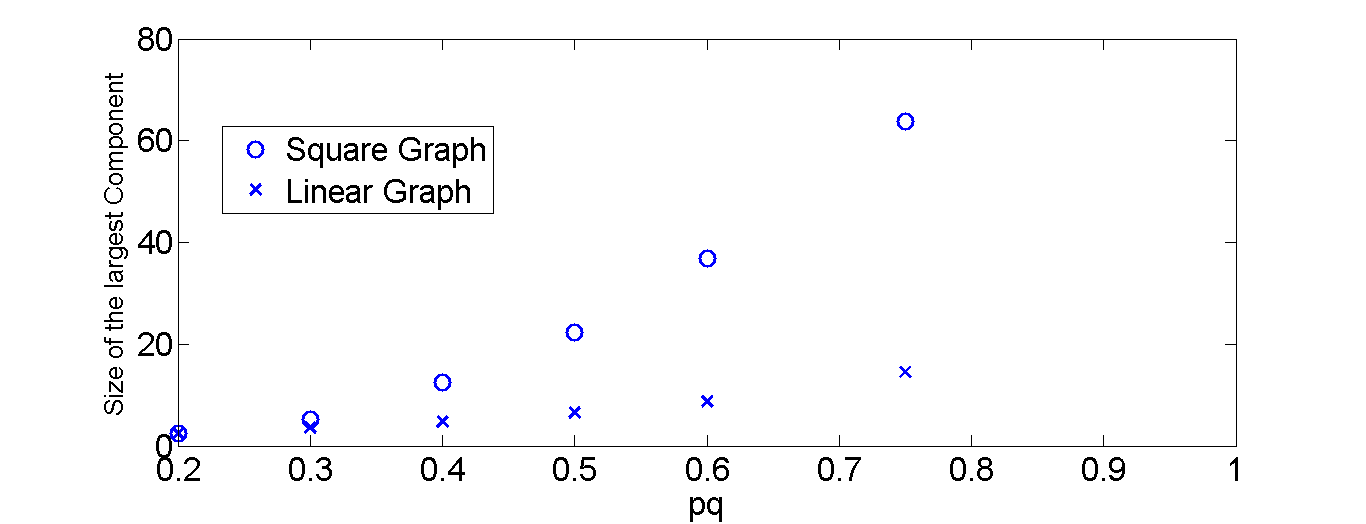}
\caption{\small Mean size of the  largest component in a Square Graph and Linear graph with $100$ nodes.}
\label{fig:prob}
\vspace{-0.3cm}
\end{center}
\end{minipage}
\end{figure*}

We now empirically investigate the variation of the mean size of the largest component with the number of nodes and the parameter $pq$. For each value of $pq$, we generate a certain number of random graphs. We compute the average of the largest component over $25$ samples. The number $25$ has been chosen since the average converges in $25$ samples.   Figure ~\ref{fig:component_nodes} shows  the variation of the mean size of the largest component  as the number of nodes increases.  Figure~\ref{fig:component_nodes} reveals that the growth of the average size of the largest component in a square graph is linear (not exponential) with the number of nodes whereas the growth of the mean size of the largest component in a linear graph is very slow with the number of nodes.  Additionally, when $pq=0.5$, the upper bound in (\ref{eq:exp_component}) is infinite both for square and linear graph, however, Fig.~\ref{fig:component_nodes} shows that the expected size of the largest component is moderate in the square graph as well as in the linear graph even when the number of nodes are large. Fig.~\ref{fig:prob} reveals that when $pq$ is exceeds a threshold the mean size of a largest component increases substantially in both linear conflict graph and square conflict graph.   However, Fig.~\ref{fig:prob} reveals that the upper bound computed in (\ref{eq:exp_component}) is loose. For example, when $0.25\leq pq\leq 0.6$, the upper bound in (\ref{eq:exp_component}) is infinite for square graph, however, Fig.~\ref{fig:prob} shows that the average size of the largest component is moderate. In the linear graph, the mean size of the largest component is small even when $0.5\leq pq\leq  0.75$ whereas the upper bound computed in (\ref{eq:exp_component}) is infinite when $0.5\leq pq\leq 0.75$.  

\section{Random Demand}\label{sec:random_demand}
Till now we have assumed that the number of secondaries ($m$) is constant at each node. But our analysis will readily generalize to the scenario where the number of secondaries at a given location is $Z\geq 1$ where $Z$ is a random variable independent of the number of secondaries at other locations with an additional assumption the p.m.f. $\Pr(Z=m)=\kappa_m$ must satisfy the condition $\sum_{m=1}^{l-1}\kappa_m>0$ (i.e. the total number of primaries exceeds the total number of secondaries with positive probability but not w.p. $1$). A primary does not know $Z$ apriori, however, it knows the p.m.f of $Z$. The analysis will go through with the following modifications in (\ref{d4})
\begin{align}
w(x)& =\sum_{m=1}^{\min\{\max(Z),l-1\}}\kappa_m\sum_{i=m}^{l-1}\dbinom{l-1}{i}x^i(1-x)^{l-i-1}
\end{align}
\section{Numerical Evaluations}\label{sec:numerical}
We numerically study the impact of competition on the payoffs of the primaries in the scenarios which we consider.  Towards that end, we compare the payoff under the symmetric NE strategy, $R_{M,NE}$, with the maximum possible value of social welfare, $R_{OPT}$ which is obtained when all the primaries collude.

\begin{align}
R_{M,NE}=\text{Number of Primaries }\times \text{Expected payoff of each primary }\nonumber
\end{align}
\begin{defn}\label{efficiency}
The efficiency of  NE is the ratio of the total expected payoff of primaries and the optimal value of social welfare ($R_{OPT}$). 
\end{defn}
In other words, efficiency ($\eta_M$)$=\dfrac{l\cdot R_{M,NE}}{R_{M,OPT}}$.

Fig. ~\ref{fig:effi} shows the variation of efficiency with the number of secondaries ($m$) in the scenario where the channel state remains the same throughout the network. Fig.~\ref{fig:effi_diff} shows the variation of efficiency with $m$ in the scenario where the channel state can be different at different locations. Both the figures reveal  that $\eta_M$ increases with increase in $m$. This is because when $m$ is low, competition becomes intense and  primaries  select lower penalties. Primaries also select independent sets of lower cardinalities when the channel state is the same at each location in Fig.~\ref{fig:effi}. But if they collude with each other, they still can offer highest penalty and only select the independent sets of the largest cardinalities in both of the settings, which lead to high payoff.  
%
%
\begin{figure*}
\begin{minipage}{0.49\linewidth}
\begin{center}
\includegraphics[width=90mm, height=30mm]{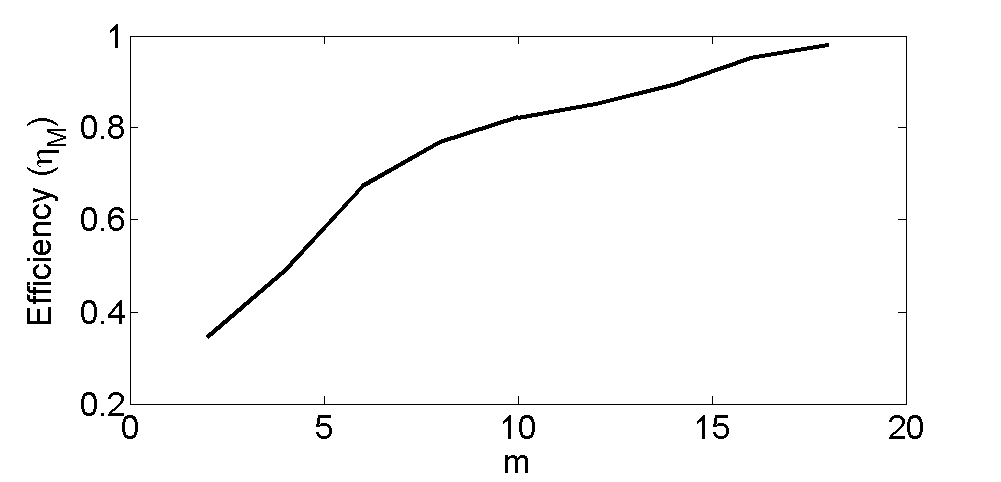}
\caption{\small This figure shows the variation of efficiency with $m$. We consider a $5\times 5$ grid graph (see Fig. ~\ref{fig:grid}). This is a mean valid graph with $d=4$ and $|I_1|=9, |I_2|=|I_3|=6, |I_4|=4$. We use the following parameter values, $l=20, n=3, v=100, c=1$, $g_i(x)=x^2-i^3$, $q_1=q_2=q_3=0.2$. }
\label{fig:effi}
\vspace{-0.5cm}
\end{center}
\end{minipage}\hfill
\begin{minipage}{0.49\linewidth}
\begin{center}
\includegraphics[width=90mm, height=30mm]{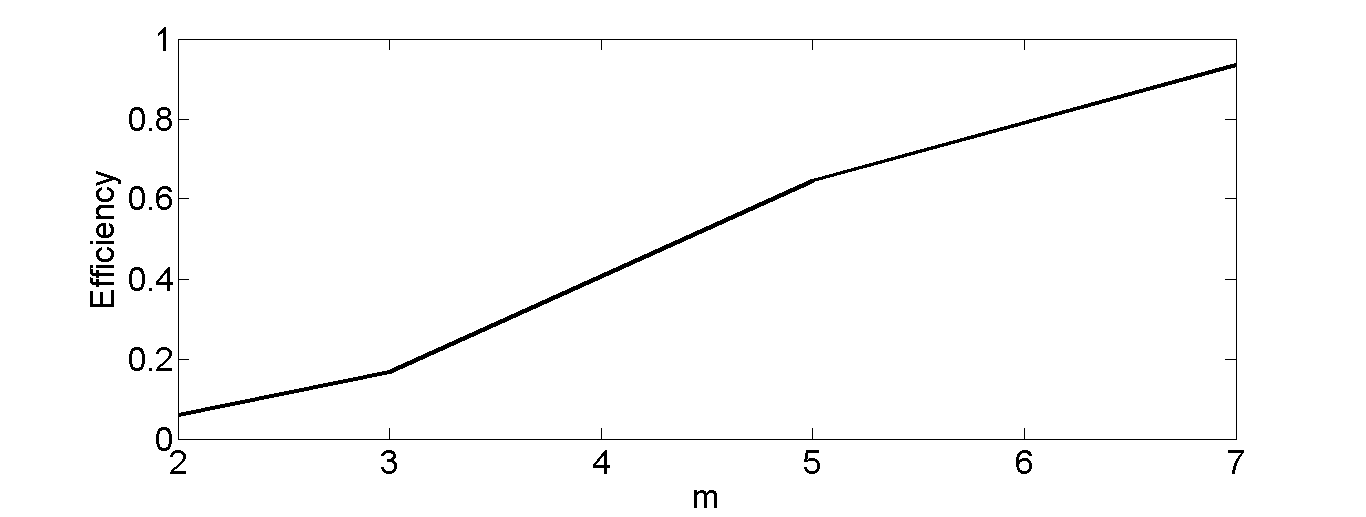}
\caption{\small This figure shows the variation of efficiency with $m$ when the channel states are independent and identically distributed at each location with $q_i=0.5, i=0,1$. We consider a cyclic graph of  100 We use the following parameter values: $l=10, v=11,c=1$.}
\label{fig:effi_diff}
\vspace{-0.5cm}
\end{center}
\end{minipage}
\end{figure*}

\section{Conclusions and Future Work}
We have studied a price competition model with the spatial reuse property where each primary selects a price and a set of non-interfering locations depending on the quality of its channel. We have considered two settings. In the first setting, we consider that the channel state is the same across the networks. We have shown that there exists a symmetric NE strategy profile in the class of mean valid graphs and we have computed a storage and computational efficient NE.  The NE strategy profile can be readily implemented as primaries need not communicate with each other even when the NE strategy is not unique. We show that the symmetric NE strategy profile is unique in a linear conflict graph. 

In the second setting, we allow that the channel state can be different at different locations. The above consideration significantly complicates the analysis as we have discussed in Section~\ref{sec:specific_assumptions}. We, therefore, consider that  the channel is either available or unavailable at each node. We have shown that there exists a symmetric NE strategy profile in the class of node symmetric graphs. In order to obtain the symmetric NE strategy, a primary only needs to enumerate the maximum independent sets. The NE strategy is also easy to implement. We have shown that symmetric NE strategy may not be  {\em unique}  in a linear conflict graph in contrast to the first setting.

The characterization of an NE in the second setting where the available channel may belong to multiple states remains open. The analytical results and tools that we have provided in this paper may provide the basis for developing a framework for this problem.
\appendix
\vspace{-0.1cm}  
We prove  Theorem~\ref{dist:existence}  (Section~\ref{sec:existencemultiplenodes}) in Appendix~\ref{sec:proof_dist_existence}. Subsequently, we show Theorem~\ref{nemulti} (Section~\ref{sec:existencemultiplenodes}) in Appendix~\ref{sec:proof_nemulti}. We show  Lemmas~\ref{thm:equalcardinality} and \ref{thm:disjointunequal} used in Section~\ref{sec:symmetricNEunique} to prove Theorem~\ref{thm:policyuniqueness} in Appendix~\ref{sec:proof_policyuniqueness}. We show Theorem~\ref{uniquelinear} (Section~\ref{sec:symmetricNEunique}) in Appendix~\ref{sec:appendix_uniquelinear}.   We prove Lemmas~\ref{thm:notsame} and \ref{thm:notane} (Section~\ref{sec:spsym}) in Section~\ref{sec:appendix_lineargraph}. Finally, in Appendix~\ref{sec:mrf} we provide  a formal definition of Markov Random Field and show that Markov random field modeling of correlated channel states satisfy Assumption~\ref{assum:prob} (Section~\ref{sec:channel_statistic}) under some additional assumptions which naturally arise in practice.

\subsection{Proof of Theorem~\ref{dist:existence}(Section~\ref{sec:existencemultiplenodes})}\label{sec:proof_dist_existence}
We proceed in two parts.  First, we will prove that there exists a distribution $t_{j}=(t_{1,j},\ldots,t_{d,j})$ which satisfies (\ref{dist}) and (\ref{ne2}). Subsequently, we will prove that such a distribution is the unique one. 

\textbf{Existence}:
First, we will show that the statement is true for $j=n$. Now, let $x\in [M_1W(q_n),M_1]$ and $s\in\{1,\ldots,d\}$. We will show that if $x\leq M_s$ then
\begin{align}
M_sW(rq_n)=x
\end{align} has a unique solution in $r$, which we will denote as $t_{s,n}(x)$. Let, $h(t_{s,n})=M_sW(t_{s,n}q_n)$, then
\begin{align}\label{s30}
h(1)& =M_sW(q_n)\nonumber\\
& \leq M_1W(q_n)
\leq x
\end{align}
and
\begin{align}\label{s31}
h(0)& =M_s
\geq x.
\end{align}
As $W(\cdot)$ is strictly decreasing and continuous, so is $h(\cdot)$, thus from (\ref{s30}) and (\ref{s31}), there exists a unique solution $t_{s,n}(x)$ between 0 and 1, such that $h(t_{s,n})=x$. Note that
\begin{align}\label{cont1}
t_{s,n}(x)=0 \quad (\text{if}\quad x=M_s).
\end{align}

$h(t_{s,n})$ is strictly decreasing in $0\leq t_{s,n}\leq 1$. Hence, inverse exists and $h^{-1}$ is also continuous as $h$ is continuous. But, $x=h(t_{s,n}(x))$. Hence, $t_{s,n}(x)=h^{-1}(x)$. Thus, $t_{s,n}(x)$ is continuous for $x\leq M_s$. For $x>M_s$, define $t_{s,n}(x)=0$. With the above definition and (\ref{cont1}) we obtain $t_{s,n}(x)$ is continuous function on $[M_1W(q_n),M_1]$ and thus
\begin{align}\label{cont1a}
t_{s,n}(x)=0 \quad(\text{if } x\geq M_s).
\end{align}
Now, let
\begin{align}
T_n(x)=\sum_{s=1}^{d}t_{s,n}(x).
\end{align}
As $h(t_{s,n})$ is strictly decreasing on $0\leq t_{s,n}\leq 1$ for $s\in \{1,\ldots,d\}$, $t_{s,n}(x)$ is strictly decreasing for $x\leq M_s$. Hence, $T_n(x)$ is strictly decreasing for $x\in [M_1W(q_n),M_1]$. Also, note that, $t_{s,n}(x)=0$ for $M_s<x\leq M_1$. Thus, for $x=M_1$, $t_{s,n}(x)=0$ $\forall s$, as $M_s\leq M_1$, hence for $x=M_1$, 
\begin{align}\label{lint}
T_n(x)=0
\end{align}
 Now, for $x=M_1W(q_n)$, $t_{1,n}(x)=1$, $t_{s,n}(x)\geq 0$ $s\in \{2,\ldots,d\}$, thus
\begin{align}\label{rint}
T_n(x)\geq 1.
\end{align}
As $t_{s,n}(x)$ are continuous, so is $T_n(x)$. Thus, from(\ref{lint}), (\ref{rint}) and intermediate value property, there exists a $x^{*}\in [M_1W(q_n),M_1]$ such that $T_n(x^{*})=1$ and this is unique as $T_n(\cdot)$ is strictly decreasing. Let, $d'_n=\max\{s:M_s> x^{*}\}$. By definition of $t_{s,n}$, for $s=1,\ldots,d'_n$,  $M_sW(t_{s,n}(x^{*})q_n)=x^{*}$ $t_{s,n}(x^{*})>0$ and for $s> d'_n$, $t_{s,n}(x^{*})=0$ (by (\ref{cont1a})). Since $\gamma_{s,n}=t_{s,n}q_n$ (from (\ref{recurg})) and $W(0)=1$, thus $M_sW(\gamma_{s,n})=x^{*}$ for $s\leq d\rq{}_n$ and for $s>d\rq{}_n$ $M_sW(\gamma_{s,n})\leq x^{*}$. Hence, $\{t_{1,n}(x^{*}),\ldots,t_{d,n}(x^{*})\}$ constitute a probability distribution and satisfy the equations (\ref{dist}) and (\ref{ne2}), with $d_n=d'_n$ and $\gamma_{s,n+1}=0$.  
Thus, the result is true for $n$.

Let, the statement be true for $k+1$, we have to show that the statement is indeed true for $k$. As the statement is true for $k+1$, thus, there exists unique distribution $t_{k+1}=(t_{1,k+1},\ldots,t_{d,k+1})$ such that (\ref{dist}) and (\ref{ne2}) holds for $j=k+1$.
 The argument will be similar to the case when $i=n$ with finding unique solution to the equation $M_sW(t_{s,k}(x)q_k+\gamma_{s,k+1})=x$ for $x\in [ M_1W(q_k+\gamma_{1,k+1}),M_1W(\gamma_{1,k+1})]$ for $x>M_sW(\gamma_{s,k+1})$ and making $t_{s,k}=0$ for $x\geq M_sW(\gamma_{s,k+1})$. Hence we omit the proof. Thus, the result is true by the principle of mathematical induction.\qed

\textbf{Uniqueness}:
We will prove the uniqueness by Induction hypothesis. First, consider the state $n$.\\
To reach a contradiction, assume that there exists $e,f\in \{1,\ldots,d\}$ such that $t'_{n}=(t'_{1,n},\ldots,t'_{d,n})$, $\bar{t}_n=(\bar{t}_{1,n},\ldots,\bar{t}_{d,n})$, $t'_{s,n}=0$ (respectively, $\bar{t}_{s,n}=0$) for $s>e$ (respectively $s>f$) and for some $y$ and $z$:
\begin{align}
y=& M_1W(t'_{1,n}q_n)=\ldots=M_eW(t'_{e,n}q_n) \geq M_{e+1}W(t^{\prime}_{e+1,n}q_n)\label{e4}\\
z=&M_1W(\bar{t}_{1,n}q_n)=\ldots=M_fW(\bar{t}_{f,n}q_n)\geq   M_{f+1}W(\bar{t}_{f+1q_n})\label{e5}.
\end{align}
First, suppose $e=f$, if $y=z$, then $M_sW(t'_{s,n}q_n)=M_sW(\bar{t}_{s,n}q_n)$ for $s\in\{1,\ldots,e\}$. But, $W(\cdot)$ is a strictly decreasing and one-to-one mapping, thus $t'_{s,n}=\bar{t}_{s,n}$ for $s\in\{1,\ldots,e\}$ and $t'_{s,n}=0=\bar{t}_{s,n}$ for $s>e$, which leads to a contradiction.

If $e=f$, but $y>z$, then $M_sW(t'_{s,n}q_n)>M_sW(\bar{t}_{s,n}q_n)$ for $s\in\{1,\ldots,e\}$. As $W(\cdot)$ is strictly decreasing function, hence we must have $t'_{s,n}<\bar{t}_{s,n}$. Now, $t^{\prime}_{s,n}=0$ for $s>e$. Thus, 
\begin{align}
\sum_{s=1}^{d}t'_{s,n}=\sum_{s=1}^{e}t^{\prime}_{s,n}<\sum_{s=1}^{d}\bar{t}_{s,n}=1.\nonumber
\end{align}The above inequality lea leads to a contradiction. Thus $y>z$ is not possible, by symmetry, $z>y$ is not possible.

Now, suppose $e>f$, thus $t\rq{}_{f+1,n}>0$. Since $W(\cdot)$ is strictly decreasing function, thus 
\begin{align}\label{u5}
M_{f+1}W(t\rq{}_{f+1,n}q_n)<M_{f+1}.
\end{align}
Since $\bar{t}_{f+1,n}=0$, thus
\begin{align}\label{eq:equality}
M_{f+1}W(\gamma_{f+1,n})=M_{f+1}.
\end{align}
Thus  from (\ref{e4}), (\ref{e5}), (\ref{eq:equality}) and (\ref{u5}), $y=M_{f+1}W(t'_{f+1,n}q_n)< M_{f+1}\leq z$. So, for $s\in\{1,\ldots,f\}$:
\begin{align}
M_sW(t'_{s,n}q_n)<M_sW(\bar{t}_{s,n}q_n).\nonumber
\end{align}
Hence, $t'_{s,n}>\bar{t}_{s,n}$, thus, $\sum_{s=1}^{f}t'_{s,n}>\sum_{s=1}^{f}\bar{t}_{s,n}=1$, which leads to a contradiction. Hence, $e>f$ is not possible, by symmetry, $e<f$ is not possible.\\ 
Thus, the result is true for $n$.

Now, assume that the statement is true for states $k+1,\ldots,n$. Since, the statement is true for states $k+1,\ldots,n$,  thus $t\rq{}_{s,j}=\bar{t}_{s,j}$ $\forall s$, $\forall j\geq k+1$. Hence,$\gamma\rq{}_{s,j}=\bar{\gamma}_{s,j}$ $\forall s$,$\forall j\geq k+1$. From (\ref{recurg}), $\gamma_{s,k}=\gamma_{s,k+1}+t_{s,k}q_k$. As $\gamma\rq{}_{s,k+1}=\bar{\gamma}_{s,k+1}$, the proof will be similar to the case when state is $n$.  \\
The result follows from the induction hypothesis.\qed

\subsection{Proof of Theorem~\ref{nemulti}(Section~\ref{sec:existencemultiplenodes})}\label{sec:proof_nemulti}
 In order to prove Theorem~\ref{nemulti} we have to show that any independent set selection strategy of the form (\ref{dist}) and (\ref{ne2}) is an NE. We can not use results derived in Section~\ref{sec:structuremultiplenodes} which we derive assuming that the strategy profile is an NE. However, we use Lemma~\ref{prop}, and Corollary~\ref{obs1} that any independent set selection strategy profile of the form (\ref{dist}) and (\ref{ne2}) satisfies regardless of whether it is an NE or not.  We also use the following result which can be easily seen from from Lemma~\ref{prop}.
Since $d_j\geq d_{j+1}$ (by Lemma~\ref{prop}), hence from (\ref{dist}), $t_{s,k}=0$ $\forall s>d_j, k\geq j$. Thus, from (\ref{recurg}) we obtain
\begin{align}\label{eq:gammazero}
\gamma_{s,j}=0 \quad\text{for } s>d_j.
\end{align}
Hence, we can write (\ref{ne2}) as 
\begin{align}\label{ne}
M_1W(\gamma_{1,j})=\ldots=M_{d_j}W(\gamma_{d_j,j})\geq M_{d_{j}+1}
\geq M_{d_{j}+2}\geq \ldots\geq M_d.
\end{align}
Now we state and prove Lemmas~\ref{obs3a} and \ref{obs3} which are satisfied by any strategy profile of the form (\ref{dist}) and (\ref{ne2}). We use these results to prove Theorem~\ref{nemulti}.
\begin{lem}\label{obs3a}
$U_{s,j}=U_{k,j}$ if $t_{s,j}, t_{k,j}>0$
\end{lem}
\begin{proof}
We prove the statement using induction argument.\\
The statement is trivially true for $j=1$ because $U_{s,1}= v$ $\forall s$ by (\ref{n152up}). 

Now, suppose the statement is true for $j=i$. Then, for any $s,k\in \{1,\ldots,d\}$, $t_{s,i}>0, t_{k,i}>0$, we have
\begin{align}\label{n60}
U_{s,i}=U_{k,i}
\end{align}
Now, let $t_{r,i+1}>0, t_{r_1,i+1}>0$ for $r,r_1\in \{1,\ldots,d\}$.  Note that  $L_{r,i}=U_{r,i+1}$ and $L_{r_1,i}=U_{r_1,i+1}$ from (\ref{n152up}). Thus, from (\ref{n152inv}),
\begin{eqnarray}
U_{r,i+1}=L_{r,i}=g_i(\dfrac{p_{r,i}-c}{W(\gamma_{r,i+1})}+c)\label{s99}\\
U_{r_1,i+1}=L_{r_1,i}=g_i(\dfrac{p_{r_1,i}-c}{W(\gamma_{r_1,i+1})}+c)\label{s100}
\end{eqnarray}
From corollary~\ref{obs1} $t_{r,i}>0, t_{r_1,i}>0$ since $t_{r,i+1}, t_{r_1,i+1}>0$. Using (\ref{n60}) for $r,r_1$  we have $U_{r,i}=U_{r_1,i}$, hence from (\ref{n152up})
\begin{align}\label{n61}
 \dfrac{p_{r,i}-c}{W(\gamma_{r,i})}& =\dfrac{p_{r_1,i}-c}{W(\gamma_{r_1,i})}
\end{align}
Now, since $t_{r,i+1}>0, t_{r_1,i+1}>0, t_{r,i}>0, t_{r_1,i}>0$ thus  $r,r_1\leq d_{i+1}\leq d_i$ (the last inequality follows from lemma~\ref{prop}). Hence,using (\ref{ne2}) for $r,r_1$ we obtain
\begin{align}
& M_rW(\gamma_{r,i})=M_{r_1}W(\gamma_{r,i})\label{n62}\\
& M_rW(\gamma_{r,i+1})=M_{r_1}W(\gamma_{r_1,i+1})\label{n63}
\end{align}
Thus, from (\ref{n61}), (\ref{n62}) and (\ref{n63}), we obtain
\begin{align}
& \dfrac{p_{r,i}-c}{W(\gamma_{r,i+1})}=\dfrac{p_{r_1,i}-c}{W(\gamma_{r_1,i+1})}\nonumber\\
& U_{r,i+1}=U_{r_1,i+1}\quad(\text{from (\ref{s99}) and (\ref{s100})})\nonumber
\end{align}
Hence, $U_{r,i+1}=U_{r_1,i+1}$. The result follows from the induction hypothesis.
\end{proof}
\begin{rmk}
 \textbf{Henceforth we denote  $U_{s,j}$ as $U_j$  when $t_{s,j}>0$ i.e. $s\leq d_j$}. Note that we have obtained similar result (Lemma~\ref{upend}) for any NE strategy profile of the form (\ref{dist1}). But Lemma~\ref{obs3a} is valid for any strategy profile of the form (\ref{dist}) and (\ref{ne2}) satisfies regardless of whether it is an NE or not.
\end{rmk}
\begin{lem}\label{obs3}
If $j\geq 2$,then for $d_k\geq s>d_j$, $L_{s,k}\geq U_j$, where $k<j$. 
\end{lem}
\begin{proof}
 Let $i=\max\{y\in\{1,\ldots,j-1\}:t_{s,y}>0\}$ . Thus, $t_{s,i}>0$, $t_{s,i+1}=0$, and $s>d_{i+1}$. 
 Since $t_{1,j}>0$ $\forall j$ by corollary~\ref{obs1}, thus $U_{k}>U_{i}$, (or $U_{1,k}>U_{1,i}$) if $k<i$. So, it is enough to show that $L_{s,i}\geq U_{i+1}$ because $i+1\leq j$ and thus $U_{i+1}\geq U_j$.  
 
Since $s>d_{i+1}$, thus $s>d_{k}$ $\forall k>i$ by Lemma~\ref{prop}. Thus, $t_{s,k}=0$ $\forall k>i$, thus $\gamma_{s,i+1}=0$ (from observation~\ref{identity1}) and thus $W(\gamma_{s,i+1})=1$. Thus, from (\ref{n152inv})
\begin{align}\label{p14}
L_{s,i}& =g_i(p_{s,i}-c+c)\nonumber\\
& =g_i((f_i(U_i)-c)W(\gamma_{s,i})+c)\quad (\text{from} (\ref{n152a}))
\end{align}
Since $L_{1,i}=U_{i+1}$ and $p_{1,i}-c=(f_i(U_i)-c)W(\gamma_{1,i})$ from (\ref{n152a}); hence, from (\ref{n152inv})
\begin{align}\label{p15}
L_{1,i}=U_{i+1}=g_i(\dfrac{(f_i(U_i)-c)W(\gamma_{1,i})}{W(\gamma_{1,i+1})}+c)
\end{align}
Since $s>d_{i+1}$, hence using (\ref{ne}) for state $i+1$, we obtain
\begin{align}\label{p16}
M_1W(\gamma_{1,i+1})\geq M_s
\end{align}
Again using (\ref{ne}) for state $i$ and noting that $1,s\leq d_i$, we obtain
\begin{align}\label{p17}
M_1W(\gamma_{1,i})& =M_sW(\gamma_{s,i})\nonumber\\
\dfrac{W(\gamma_{1,i})}{W(\gamma_{s,i})}& =\dfrac{M_s}{M_1}
\leq W(\gamma_{1,i+1})\quad(\text{from} (\ref{p16}))
\end{align}
Since $g_i(\cdot)$ is strictly increasing, in order to show that $U_i\leq L_{s,i}$, from (\ref{p14}) and (\ref{p15}) it is sufficient to show the following
\begin{align}\label{p18}
\dfrac{(f_i(U_i)-c)W(\gamma_{1,i})}{W(\gamma_{1,i+1})}\leq (f_i(U_i)-c)W(\gamma_{s,i})
\end{align}
Since $f_i(U_i)>c$, (\ref{p18}) readily follows from (\ref{p17}).  
\end{proof}

Now we are ready to show Theorem~\ref{nemulti}.

\textit{Proof of Theorem~\ref{nemulti}}:
We will show that for channel state $j\in\{1,\ldots,n\}$, probability distribution $t_{s,j}$ as described in (\ref{dist}) and (\ref{ne2})  for $s\in\{1,\ldots,d\}$ is a best response.
\begin{enumerate}
\item First, we will show that under the strategy profile for $s\leq d_j$, at any independent set $I_s$,maximum  expected payoff is given by $P^{*}_j$ (equation (\ref{c61}))  and the maximum value is obtained at each  penalty value in the interval $[L_{s,j},U_j]$ at every node of $I_s$. (Case i)
\item Next, we will show that for any choice of penalty a primary can only attain at most an expected payoff of $P^{*}_j$ at any $I_s$, $s>d_j$ (Case ii and Case iii).
\item Finally, we will show that if a primary selects any other independent set i.e. apart from $I_1,\ldots, I_d$ then its expected payoff is upper bounded by $P^{*}_j$ for any choice of penalty (Case iv).
 \end{enumerate} 
\textit{Case i}:At independent set $I_s,s\leq d_j$.\\
In this case $t_{s,j}>0$. From Theorem~\ref{singlelocation}, Lemma~\ref{separation-s1} and (\ref{n152a}) at a node $D\in I_s$, $s\leq d_j$,  a primary gets a maximum payoff of $p_{s,j}-c$ when the channel state is $j$. Since $s\leq d_j$, hence $U_{s,j}=U_j$. Thus,
 \begin{align}
 p_{s,j}-c=(f_j(U_j)-c)W(\gamma_{s,j})\quad(\text{from } (\ref{n152a})).\nonumber
 \end{align}
Thus, the expected payoff that a primary obtains when channel state is $j$, at independent set $I_s$, 
\begin{align}
M_s(f_j(U_j)-c)W(\gamma_{s,j})=P^*_{j} \quad (\text{from} (\ref{c61})).\nonumber
\end{align}
Hence, payoff at each node of independent set $I_s, s\leq d_j$ is 
\begin{align}\label{npay}
(f_j(U_j)-c)W(\gamma_{s,j})=p_{s,j}-c=\dfrac{P^{*}_j}{M_s}.
\end{align}
From Theorem \ref{singlelocation}, the best response penalty set is $[L_{s,j},U_j]$ under $t_{s,k}$ $k= 1,\ldots,n$ at node $D$. Thus, for any $x\notin[U_{j},L_{s,j}]$, payoff is atmost equal to (\ref{npay}). This completes case i.

\textit{Case ii}: At independent set $I_s$, where$d_1\geq s>d_j$.

Note from Lemma~\ref{prop} that $d_{i}\geq d_{i+1}$. Thus, in this case we must have $k=\max \{i\in\{1,\ldots,j-1\}:s\leq d_i\}$. As $s\leq d_1$, hence this case arises only when $j\geq 2$. So, we have $L_{s,k}\geq U_j, j>k$ from Lemma \ref{obs3}. 

Now, $\gamma_{s,k+1}=0$ as $s>d_{k+1}$ from (\ref{eq:gammazero}) and thus, $W(\gamma_{s,k+1})=1$. 
Thus, the expected payoff to a primary at $L_{s,k}$, when the channel state is $j$, is $(f_j(L_{s,k})-c)$. So, any penalty less than $L_{s,k}$ will fetch a strictly lower payoff compared to penalty $L_{s,k}$ at any node at $I_s$.   Hence, it is enough to show that if a primary chooses penalty in the interval $[L_{s,k},v]$ at a node of $I_s$, then its payoff will be strictly less than $P^{*}_j/M_s$. 

If $x\in [L_{s,k},v]$ then $x$ must belong to interval $[L_{s,r},L_{s,r-1}]$ for some $r\leq k$, with $L_{s,0}=v$. Without loss of generality, we can assume that $x\in [L_{s,i},L_{s,i-1}]$ where $i\leq k$ . From Corollary~\ref{obs1} $t_{s,i}>0$, since $t_{s,k}>0$; thus $x$ is a best penalty response for channel state $i$ by Theorem~\ref{singlelocation}. The expected payoff to a primary, when it selects penalty $x$ at channel state $i$ at a node $D\in I_s$, is given by
\begin{align}
(f_i(x)-c)P(A)=p_{s,i}-c=(f_i(U_i)-c) W(\gamma_{s,i})
\end{align}
where, $P(A)$ is the probability of winning, when a primary selects penalty $x$ at channel state $i$, at any node $D\in I_s$.

Since $s\leq d_i$, thus from (\ref{npay}) we obtain for state $i$
\begin{align}\label{g7}
P^{*}_i=M_s(f_i(U_i)-c)W(\gamma_{s,i})=M_s(f_i(x)-c)P(A).
\end{align}
Since $x\geq L_{s,i}$ and $f_i(L_{s,i})>c$ from (\ref{pay0}), thus $f_i(x)>c$. Since probability of winning only depends on the penalty selected by a primary, thus, when a primary selects penalty $x$ at node $D\in I_s$, at channel state $j$, its expected payoff is
\begin{align}\label{s4}
(f_j(x)-c)P(A)=\dfrac{P^{*}_i}{M_s}\dfrac{f_j(x)-c}{f_i(x)-c}\quad (\text{from} (\ref{g7})).
\end{align}
Since $1\leq d_i$, thus expected payoff at any node of $I_1$ at channel state $i$ is given by (\ref{n152inv}) and (\ref{npay})
\begin{align}\label{s6}
p_{1,i}-c=(f_i(L_{1,i})-c)W(\gamma_{1,i+1})=\dfrac{P^{*}_i}{M_1}.
\end{align}
Again since $1\leq d_j$, thus, at any node of independent set $I_1$, maximum expected payoff obtained by a primary at channel state $j$ is $P^{*}_j/M_1$ as given in (\ref{npay}). 
Expected payoff at $L_{1,i}$ at channel state $j$ is
\begin{align}\label{s7}
(f_{j}(L_{1,i})-c)W(\gamma_{1,i+1}) 
\leq \dfrac{P^{*}_{j}}{M_1}.
\end{align}
If $s\leq d_{i+1}$, then $L_{s,i}=U_{i+1}$; on the other hand if $s>d_{i+1}$, then by Lemma~\ref{obs3}, $L_{s,i}\geq U_{i+1}$. Hence, $x\geq L_{s,i}\geq U_{i+1}$. Also, note that $i<j$ by definition of $i$. Since $L_{1,i}=U_{i+1}$ (as $1\leq d_{i+1}$); hence, at penalty $x$, at channel state $j$ and at a node $D\in I_s$,  expected payoff to a primary is
\begin{align}\label{c99}
& \leq \dfrac{P^{*}_j}{M_s}\dfrac{(f_j(x)-c)(f_i(U_{i+1})-c)}{(f_j(U_{i+1})-c)(f_i(x)-c)}\quad (\text{from } (\ref{s6}) \& (\ref{s7}))\nonumber\\
& \leq \dfrac{P^{*}_j}{M_s} \quad (\text{by (\ref{con1}) as } x\geq U_{i+1}, i<j, f_i(U_{i+1})>c ).
\end{align}
\textit{Case iii}: $s>d_1$.

We have from (\ref{ne})
\begin{align}\label{s52}
M_1W(\gamma_{1,1})=\ldots=M_{d_1}W(\gamma_{d_1,1})\geq M_{s}\quad s>d_1.
\end{align}
Since $1\leq d_j$ thus the maximum expected payoff at $v$ is upper bounded by $P^{*}_j/M_1$ by (\ref{npay}). But, expected payoff to a primary at channel state $j$ at $v$ at any node of $I_1$ is 
\begin{align}\label{s59}
(f_j(v)-c)W(\gamma_{1,1})\leq \dfrac{P^{*}_j}{M_1}.
\end{align}
Since the expected payoff a primary can attain at a node is at most $f_j(v)-c$ at channel state $j$. Thus, a primary\rq{}s expected payoff at  any node $D\in I_s$ is always upper bounded by
\begin{align}\label{c100} 
f_j(v)-c& \leq \dfrac{M_1}{M_s}(f_j(v)-c)W(\gamma_{1,1})\quad(\text{from} (\ref{s52}))\nonumber\\
& \leq \dfrac{P^{*}_j}{M_s}\quad (\text{from } (\ref{s59})).
\end{align}

\textit{Case iv}: At any independent set other than $I_1,\ldots,I_d$:\\
From (\ref{npay}), (\ref{c99}) and (\ref{c100}), at any node at independent set $I_s, s>d_j$, we obtain that  maximum expected payoff a primary can obtain for state $j$-
\begin{align}\label{s8}
\leq \dfrac{P^{*}_j}{M_s}.
\end{align}
 The graph we have considered, is a d-partite graph (Section~\ref{sec:meanvalidgraph}) . Now, consider an independent set $I$  which contains $m_s(I)$ number of nodes from $I_s, s=1,\ldots, d$.  Then at channel state $j$, expected payoff at independent set $I$ is sum of all payoffs at all the nodes contained in $I$. Hence, from (\ref{s8})
\begin{align}
\text{Expected Payoff at $I$}& \leq \sum_{s=1}^{d}\dfrac{P^{*}_j}{M_s}m_s(I)\nonumber\\
& =P^{*}_j\sum_{s=1}^{d}\dfrac{m_s(I)}{M_s}\nonumber\\
& \leq P_j^{*}\quad (\text{from} (\ref{mvg})).\nonumber
\end{align}
Thus, at any independent set $I$, expected payoff to a primary at channel state $j$ is at most $P^*_j$ for any selection of penalty. From case (i) a primary attains $P^*_j$ at $I_s, s\leq d_j$ following the strategy profile. Hence, the result follows. 
\qed

\subsection{Proof of  Lemmas~\ref{thm:equalcardinality} and \ref{thm:disjointunequal}}\label{sec:proof_policyuniqueness}
Throughout this section we use $|I_i|$ and $M_i, i\in \{1,\ldots,d\}$ ($|\bar{I}_i|$ and $\bar{M}_i, i=1,\ldots,\bar{d}$ respv.) interchangeably. 
\subsubsection{Proof of Lemma~\ref{thm:equalcardinality}}
First we show that $M_1=\bar{M}_1$. 

Let $M_1\neq \bar{M}_1$. Without loss of generality assume that $M_1>\bar{M}_1$. Let $\bar{I}_1$ consists of $m_s(\bar{I}_1)$ number of nodes from $I_s$. Then
\begin{align}
\sum_{s=1}^{d}\dfrac{m_s(\bar{I}_1)}{M_s}\geq \sum_{s=1}^{d}\dfrac{m_s(\bar{I}_1)}{M_1}=\dfrac{\bar{M}_1}{M_1}>1.
\end{align}
which contradicts (\ref{mvg}).

Suppose that $M_j\neq \bar{M}_j$ for some smallest index $j\in \{2,\ldots,d\}$. Without loss of generality, we assume that $M_j<\bar{M}_j$. By the definition of $j$, $M_k=\bar{M}_k$ for $k< j$, thus $\sum_{k=1}^{j-1}M_k=\sum_{k=1}^{j-1}\bar{M}_k$. Note that 
\begin{align}
M_{j-1}=\bar{M}_{j-1}\geq \bar{M}_j>M_j.
\end{align}
We consider two possible scenarios:

\textit{Case i}: $\bar{I}_k, k\in \{1,\ldots,j-1\}$ does not contain node from $I_s$ $s\geq j$. 

Since $\sum_{k=1}^{j-1}|\bar{I}_k|=\sum_{k=1}^{j=1}\bar{M}_k=\sum_{k=1}^{j-1}M_k$, thus, $\bar{I}_j$ must consist of nodes of only $I_s,s\geq j$. Let $\bar{I}_j$ consist of $m_s(\bar{I}_j)$ nodes of $I_s$. Then,
\begin{align}
\sum_{k=j}^{d}\dfrac{m_k(\bar{I}_j)}{M_k}\geq \sum_{k=j}^{d}\dfrac{m_k(\bar{I}_j)}{M_j}
= \dfrac{|\bar{I}_j|}{|I_j|}>1
\end{align}
which is not possible by (\ref{mvg}).

\text{Case ii}:$\bar{I}_k$ contains at least one node from $I_s$ $s\geq j$ for some $k\in \{1,\ldots,j-1\}$. 
 
 Let $\bar{I}_k$ consist of $m_i(\bar{I}_k)$ number of nodes from $I_i$. 
 Since $M_{j-1}>M_j$, thus, $\dfrac{m_s(\bar{I}_k)}{M_i}<\dfrac{m_s(\bar{I}_k)}{M_s}$ for any $s\leq j$ and $i>j$. By (\ref{mvg}) for each $k\in \{1,\ldots,j-1\}$ we have
 \begin{align}\label{eq:sumind}
 \sum_{i=1}^{d}\dfrac{m_i(\bar{I}_k)}{M_i}\leq 1\nonumber\\
 \sum_{k=1}^{j-1}\sum_{i=1}^{d}\dfrac{m_i(\bar{I}_k)}{M_i}\leq j-1.
 \end{align}
 Since $\bar{I}_k$ s are disjoint, thus, $|\bar{I}_1\cup\ldots\bar{I}_{j-1}|=\sum_{k=1}^{j-1}\bar{M}_k=\sum_{k=1}^{j-1}M_k$. 
 Thus, $\sum_{i=1}^{j-1}M_i=\sum_{k=1}^{j-1}\sum_{i=1}^dm_i(\bar{I}_k)$. Hence,
 \begin{align}\label{eq:indidentity}
 \sum_{i=j}^{d}\sum_{k=1}^{j-1}m_i(\bar{I}_k)=\sum_{i=1}^{j-1}(M_i-\sum_{k=1}^{j-1}m_i(\bar{I}_k)).
 \end{align}
 Since $\bar{I}_k$ contains at least one node from $I_s, s\geq j$, thus, the above expression is strictly positive. Note that
\begin{align}
\sum_{k=1}^{j-1}\sum_{i=1}^{d}\dfrac{m_i(\bar{I}_k)}{M_i}& \nonumber\\
& \geq\sum_{k=1}^{j-1}(\sum_{i=1}^{j-1}\dfrac{m_i(\bar{I}_k)}{M_i}+\sum_{i=j}^{d}\dfrac{m_i(\bar{I}_k)}{M_j})\nonumber\\
& =\sum_{i=1}^{j-1}\sum_{k=1}^{j-1}\dfrac{m_i(\bar{I_k})}{M_i}+\sum_{i=j}^{d}\sum_{k=1}^{j-1}\dfrac{m_i(\bar{I}_k)}{M_j}\nonumber\\
& >\sum_{i=1}^{j-1}\sum_{k=1}^{j-1}\dfrac{m_i(\bar{I}_k)}{M_i}+\sum_{i=1}^{j-1}\dfrac{M_i-\sum_{k=1}^{j-1}m_i(\bar{I}_k)}{M_i}\nonumber\\
& (\text{from } (\ref{eq:indidentity}) \text{and } M_i>M_j, i<j)\nonumber\\
& =\sum_{i=1}^{j-1}\dfrac{M_i}{M_i}=j-1.
\end{align}
which contradicts (\ref{eq:sumind}), hence this case can not arise. 
Hence, the result follows. \qed

\subsubsection{Proof of Lemma~\ref{thm:disjointunequal}}

Let the lowest index be $j$ such that $I_j\cap \bar{I}_k\neq \Phi$ but $|I_j|\neq \bar{I}_k$. Thus, $I_j$ contains at least one node from $I_k$. Without loss of generality we can assume that $|I_j|<|\bar{I}_k|$.

Since $|I_k|=|\bar{I}_k|$ for all $k$ by Lemma~\ref{thm:equalcardinality}, thus, $M_k>M_j$. Let $k_1=\max\{i\in \{1,\ldots,j-1\}:M_i>M_j\}$. Let $I_i$ consists of $m_s(I_i)$ number of nodes from $\bar{I}_s$. Thus, 
\begin{align}\label{eq:totalnode}
\sum_{i=1}^{k_1}|I_i|=\sum_{s=1}^{d}\sum_{i=1}^{k_1}m_s(I_i)\nonumber\\
\sum_{i=1}^{k_1}\sum_{s=k_1+1}^{d}m_s(I_i)=\sum_{i=1}^{k_1}M_i-\sum_{i=1}^{k_1}\sum_{s=1}^{k_1}m_s(I_i)
\end{align}
 Note that LHS of (\ref{eq:totalnode}) is always non-negative. Now, we will show that (\ref{eq:totalnode}) is strictly positive. If it is not strictly positive then we must have
\begin{align}\label{eq:set1}
\sum_{i=1}^{k_1}M_i=\sum_{i=1}^{k_1}\sum_{s=1}^{k_1}m_s(I_i).
\end{align}
But RHS of (\ref{eq:set1}) is equal to 
\begin{align}
|(\bar{I}_1\cup\bar{I}_2\ldots\cup\bar{I}_{k_1})\cap (I_1\cup I_2\ldots\cup I_{k_1})|.
\end{align}
and LHS of (\ref{eq:set1}) is equal to 
\begin{align}
|\bar{I}_1\cup\ldots\cup \bar{I}_{k_1}|=|I_1\cup\ldots\cup I_{k_1}|.
\end{align}
Thus,
\begin{align}
I_1\cup\ldots\cup I_{k_1}=\bar{I}_1\cup\ldots\cup \bar{I}_{k_1}.
\end{align}
But $I_j$ contains at least one node from $\bar{I}_l$ and $k\leq k_1<j$. Thus, $I_j$ contains at least one node in common with $I_1\cup\ldots\cup I_{k_1}$ which is not possible since $I_j$s are disjoint. Thus (\ref{eq:totalnode}) is strictly positive. Thus, there must exist a $i\in \{1,\ldots,k_1\}$ such that $I_i$ contains at least one node from $\bar{I}_s$ $s>k_1$. Since $i\leq k_1$ and $s>k_1$, thus, $|\bar{I}_s|=|I_s|<|I_{i}|$. Hence, we have found a $i<j$ such that $I_i$ contains at least one node from $\bar{I}_s$ such that $|\bar{I}_s|<|I_i|$ which contradicts the definition of $j$. Hence, the result follows.\qed  
 \subsection{Proof of Theorem~\ref{uniquelinear} (Section~\ref{sec:symmetricNEunique})}\label{sec:appendix_uniquelinear}
In order to prove Theorem~\ref{uniquelinear} we must consider all symmetric NE strategy profiles which need not be of the form (\ref{dist1}); this precludes the use of the results in section~\ref{sec:structuremultiplenodes}. First, we characterize some properties that any symmetric NE strategy must follow (Lemmas~\ref{high},~\ref{high2}) in Appendix~\ref{sec:symmetricstructure}. Then we deduce some important properties (Lemma~\ref{lm:eqpayoff} and ~\ref{equalnodedist}) that any NE strategy profile must satisfy in a linear graph in Appendix~\ref{sec:propertieslinear}. We then use those properties to prove Theorem ~\ref{uniquelinear}. 
\subsubsection{Properties of any symmetric NE strategy profile (Lemmas~\ref{high} and ~\ref{high2})}\label{sec:symmetricstructure}

In order to prove Lemma~\ref{high} we state and prove Observations~\ref{identity3}, and \ref{identity4} and Lemma\ref{lem:best}. We subsequently state and prove Lemma~\ref{obs:recurrence} to prove Lemma~\ref{high2}. 

  We start with some notations, which we use throughout.
\begin{defn}
$u_{s,i,max}$ denotes the maximum expected payoff under an NE strategy for state $i$ at node $s$ \footnote{Even if node $a$ is selected with probability $0$ when the channel state is $i$, we can still defined $u_{a,i,max}$ as  the maximum expected payoff that a primary would have obtained if it would select node $a$} .
\end{defn}
Recall from (\ref{charac}) that the channel is offered at node $a$ when the state is $j$ with probability $\alpha_{a,j}$. With slight abuse of notation, we define $\gamma_{a,i}$ for node $a$ in the following manner:
\begin{align}\label{ne+1}
\gamma_{a,i}=\sum_{j=i}^{n}\alpha_{a,j}
\end{align}
Thus, $\gamma_{a,i}$ denotes the probability that the channel is offered at node $a$ when the state is higher or equal to $i$.

Since we have to consider all NE strategy profiles which may not be of the form (\ref{dist1}), thus, the payoff, upper and lower endpoint need not be the identical at each node of $I_s$ at a given channel state. By Lemma~\ref{separation-s1} if $\alpha_{a,j}$ is known then the above parameters can be obtained using Lemma~\ref{lm:computation} with $\alpha_{a,j}$ in place of $q_j$. With slight abuse of notation we denote $p_{a,i}, L_{a,i}$ and $U_{a,i}$ for node $a$ i.e. 
for $i=1,\ldots,n$
\begin{eqnarray}
p_{a,i}&=
&c+(f_i(U_{a,i})-c)W(\gamma_{a,i})(\text{from }(\ref{ne+1})\& (\ref{defnW}))\label{nnodepay}\\
L_{a,i}&=&g_i(\dfrac{p_{a,i}-c}{W(\gamma_{a,i+1})}+c), U_{a,i}=L_{a,i-1}, L_{a,0}=v\label{nnode151}
\end{eqnarray}
By Theorem~\ref{singlelocation} $p_{a,i}-c$ is the expected payoff at node $a$ when the channel state is $i$ if node $a$ is selected with positive probability. By Theorem~\ref{singlelocation} a primary selected penalty from the interval $[L_{s,j}, U_{a,j}]$ when the channel state is $j$ using the distribution ~(\ref{d4}) with $\alpha_{a,j}$ in place of $q_j$. 

Now we state some observations which we use throughout.
\begin{obs}\label{identity3}
At node $a$, $\gamma_{a,k}=\gamma_{a,k_1}+\sum_{i=k}^{k_1-1}\alpha_{a,i}$ where $n\geq k_1>k$.
\end{obs}
Observation~\ref{identity3} readily follows from (\ref{ne+1}). Since from (\ref{ne+1})
\begin{align}
\gamma_{a,k} =\sum_{i=k}^{k_1-1}\alpha_{a,i}+\sum_{i=k_1}^{n}\alpha_{a,i}=\sum_{i=k}^{k_1-1}\alpha_{a,i}+\gamma_{a,k_1}\quad (\text{from } (\ref{ne+1}))\nonumber
\end{align}
Similar to observation~\ref{identity2}, using observation~\ref{identity3}, (\ref{nnode151}) and (\ref{nnodepay}) we obtain
\begin{obs}\label{identity4}
At node $a$, $U_{a,j}=L_{a,j}$ for $j\in\{1,\ldots,n\}$ iff $t_{a,j}=0$. $U_{a,j}=L_{a,k}$ iff $t_{a,i}=0$ $\forall k<i<j$.  Hence, $U_{a,j}=v$ iff $t_{a,k}=0$ $\forall k<j$.
\end{obs}
\begin{lem}\label{lem:best}
 Maximum expected payoff under the NE strategy profile at a node $s$ is obtained at $L_{s,i}$ when channel states are  $i$ and $i+1$. When the channel state is $1$, primary attains its maximum expected payoff at $v$ at any node.
\end{lem}
Note that if $\alpha_{s,i}>0$ then by Theorem~\ref{singlelocation} $L_{s,i}$ is a best penalty response at channel state $i$. Here we show that even if $\alpha_{s,i}=0$, then the maximum expected payoff is obtained at $L_{s,i}$ at node $s$ under any NE strategy profile. 
\begin{proof}
First, we will prove the statement for channel state $i$. The proof for channel state $i+1$ will readily follow.\\
Suppose the statement is false for channel state $i$. Hence, there exists $x$ at which expected payoff is higher compared to the expected payoff at $L_{s,i}$ when the channel state is $i$. First we rule out $x>L_{s,i}$ (case i) and then $x<L_{s,i}$ (case ii).

{\em case i}: $x>L_{s,i}$:\\
Note that this case can not arise when $L_{s,i}=v$. Hence, $L_{s,i}<v$; thus by observation~\ref{identity4} there must exist $j=\max\{1,\ldots,i\}$ such that $\alpha_{s,j}>0$. If $j=i$, then $L_{s,j}=L_{s,i}$; on the other hand if $\alpha_{s,i}=0$ then by observation~\ref{identity4} $L_{s,i}=U_{s,j+1}=L_{s,j}$.  Expected payoff to a primary at state $i$ at $x$ is
\begin{align}\label{ln1}
(f_i(x)-c)P(A)
\end{align}
where $P(A)$ is the probability of winning at penalty $x$ at node $s$. By theorem~\ref{singlelocation}, $L_{s,j}$ is a best penalty response at node $s$ when the channel state is $j$. Now, expected payoff at $L_{s,j}$ when channel state is $j$, is
\begin{align}
p_{s,j}-c=(f_{j}(L_{s,j})-c)W(\gamma_{s,j+1})\quad (\text{from } (\ref{nnode151}))\nonumber
\end{align}
Note that players with channel state higher than $j$ select a penalty lower than or equal to $L_{s,j}$ with probability $1$ and players with channel state lower than or equal to $i$ select a penalty lower than or equal to $L_{s,j}$ with probability $0$ at node $s$. Thus, the expected payoff to a primary when it selects penalty $L_{s,j}$ at channel state $i$ at node $s$ is
\begin{align}\label{ln3}
(f_i(L_{s,j})-c)W(\gamma_{s,j+1})
\end{align}
Note that $f_{j}(L_{s,j})>c$ by (\ref{nnode151}).\\
Since expected payoff at $x$ is strictly higher compared to the expected payoff at $L_{s,i}$ at  node $s$ at channel state $i$ and $L_{s,i}=L_{s,j}$, thus, we have from (\ref{ln1}) and (\ref{ln3})
\begin{align}\label{ln4}
(f_i(x)-c)P(A)>(f_i(L_{s,j})-c)W(\gamma_{s,j+1})
\end{align}
On the other hand, the expected payoff that a primary will obtain when it selects penalty $x$ at node $s$ at channel state $j$-
\begin{align}
& (f_{j}(x)-c)P(A) \nonumber\\
& > (f_{j}(x)-c)\dfrac{f_i(L_{s,j})-c}{f_i(x)-c}W(\gamma_{s,j+1}) \quad (\text{from} (\ref{ln4}))\nonumber\\
& > (f_{j}(L_{s,j})-c)P(A_1)\quad \nonumber\\
& (\text{from} (\ref{con1}) \text{as }i\geq j, f_{j}(L_{s,j})>c, x>L_{s,j})
\end{align}
which contradicts the fact that $L_{s,j}$ is a best penalty response at channel state $j$. 

\textit{Case ii} $x<L_{s,i}$:\\
Note that, if $\alpha_{s,j}=0$ forall $j> i$, then it is trivial that this case can not arise\footnote{ In this case, $L_{s,j}=L_{s,i}$ (by observation~\ref{identity4}) for all $j>i$. Thus, expected payoff at any penalty strictly less than at $L_{s,i}$ will yield strictly lower payoff compared to payoff at $L_{s,i}$}. We only consider the scenario when $\alpha_{s,j}>0$ for some $j\in \{i+1,\ldots,n\}$. Note that $f_i(x)>c$. Now let, $k=\min\{j>i: \alpha_{s,j}>0\}$.  By definition of $k$ and observation~\ref{identity4} $L_{s,i}=U_{s,k}$. Since $\alpha_{s,k}>0$, thus expected payoff at $U_{s,k}$ is the maximum expected payoff at node $s$ when the channel state is $k$ (theorem~\ref{singlelocation}). Expected payoff to a primary channel state $k$ at $L_{s,i}$ is
\begin{equation*}\label{ln5}
(f_k(L_{s,i})-c)P(A_2)\nonumber
\end{equation*}
where $P(A_2)$ denotes the probability of winning when a primary offers penalty $L_{s,i}$ at node $s$. Since probability of winning does not depend on the channel state, hence, expected payoff to a primary at channel state $i$ and at penalty $L_{s,i}$ is
\begin{align}\label{ln6}
(f_i(L_{s,i})-c)P(A_2)
\end{align}
Let, probability of winning at penalty $x$ at node $s$ be $P(A_3)$. Since, probability of winning does not depend on the channel state , thus expected payoff to a primary when it offers penalty $x$ at channel state $k$ and at node $s$ is
\begin{equation*}
(f_k(x)-c)P(A_3)
\end{equation*}
Similarly expected payoff at node $s$, at channel state $i$ and at penalty $x$ is-
\begin{equation*}
(f_i(x)-c)P(A_3)
\end{equation*}
Since $L_{s,i}$ is a best penalty response to channel state $k$ at node $s$, thus
\begin{align}\label{ln7}
(f_k(L_{s,i})-c)P(A_2)\geq (f_k(x)-c)P(A_3)
\end{align}
From (\ref{ln6}), expected payoff at $L_{s,Z_{s,i}}$ at node $s$ and at channel state $i$ is given by
\begin{align}\label{ln8}
& (f_i(L_{s,i})-c)P(A_2) \nonumber\\
& \geq (f_i(L_{s,i})-c)P(A_3)\dfrac{f_k(x)-c}{f_k(L_{s,i})-c}\quad (\text{from} (\ref{ln7}))\nonumber\\
& >(f_i(x)-c)P(A_3)\nonumber\\& 
\quad(\text{Using} (\ref{con1}) \text{as} i<k, L_{s,i}>x, f_i(x)>c)
\end{align}
which contradicts the fact that expected payoff at $x$ is higher compared to the expected payoff at $L_{s,i}$ when the channel state is $i$.

Now, we show the result for channel state $i+1$.

If $\alpha_{s,i+1}=0$, then by observation~\ref{identity4} $L_{s,i}=U_{s,i+1}=L_{s,i+1}$. Hence, the same analysis will follow for channel state $i+1$. On the other hand if $\alpha_{s,i+1}>0$ which along with $L_{s,i}=U_{s,i+1}$ (by (\ref{nnode151}) implies that $L_{s,i}$ is the upper endpoint of the penalty selection strategy profile for channel state $i+1$ at node $s$. Since the upper endpoint is also a best penalty response by theorem~\ref{singlelocation}, thus the result follows.
\end{proof}

Now, we provide expressions for $u_{s,i,max}, u_{s,i+1,max}$ for node $s$, $i\in \{1,\ldots,n-1\}$ in terms of $L_{s,i}$which we use to prove Lemmas ~\ref{high} and \ref{high2}.

Since $v$ is  a best response at channel state $1$ at  any node in the network by Lemma~\ref{lem:best}, thus, 
\begin{align}\label{vbest}
u_{s,1,max}=(f_1(v)-c)W(\gamma_{s,1})
\end{align}
By (\ref{nnode151}) expected payoff at $L_{s,i}$ is
\begin{eqnarray}
(f_i(L_{s,i})-c)W(\gamma_{s,i+1})=u_{s,i,max}\label{eq:payi}\\
u_{s,i+1,max}=(f_{i+1}(L_{s,i})-c)W(\gamma_{s,i+1})\label{eq:payi+1}
\end{eqnarray}

\begin{lem}\label{high}
i) For, $i\in \{1,\ldots,n-1\}$, if $u_{s,i,max}\geq u_{r,i,max}$  and $\gamma_{s,i}\leq \gamma_{r,i},\alpha_{r,i}<\alpha_{s,i}$, then $u_{s,i+1,max}>u_{r,i+1,max}$. \\
ii) If $u_{s,i,max}\geq u_{r,i,max}$ and $\gamma_{s,i}<\gamma_{r,i}$, $\alpha_{s,i}\geq \alpha_{r,i}$, then $u_{s,i+1,max}>u_{r,i+1,max}$.
\end{lem}
\begin{proof} 
First we show part (i). Proof of part (ii) follows by simple modification of the proof of part (i).

Suppose, the statement is false, i.e. $u_{s,i+1,max}\leq u_{r,i+1,max}$ for some $s$ and $r$. 
As $\gamma_{s,i}\leq\gamma_{r,i}$ thus,
\begin{align}\label{n85a}
& \gamma_{s,i+1}+\alpha_{s,i}\leq \gamma_{r,i+1}+\alpha_{r,i}\quad (\text{by observation~\ref{identity3}})\nonumber\\
& \gamma_{s,i+1}<\gamma_{r,i+1} \quad (\text{since } \alpha_{s,i}>\alpha_{r,i})
\end{align}
 Now, as $u_{s,i+1,max}\leq u_{r,i+1,max}$, hence from (\ref{eq:payi+1})
\begin{align}\label{n86a}
(f_{i+1}(L_{s,i})-c)W(\gamma_{s,i+1})& \leq (f_{i+1}(L_{r,i})-c)W(\gamma_{r,i+1})\nonumber\\
\dfrac{W(\gamma_{s,i+1})}{W(\gamma_{r,i+1})} &\leq \dfrac{f_{i+1}(L_{r,i})-c}{f_{i+1}(L_{s,i})-c}
\end{align}
Since $\gamma_{r,i+1}>\gamma_{s,i+1}$ (from (\ref{n85a})) $W(\cdot)$ is strictly decreasing, thus $W(\gamma_{r,i+1})<W(\gamma_{s,i+1})$. Since $f_{i+1}(\cdot)$ is strictly increasing, thus we obtain from (\ref{n86a}) $L_{s,i}<L_{r,i}$.
Now, from (\ref{n86a}) and the fact that $f_i(L_{s,i})>c$ , we obtain
\begin{align}
\dfrac{W(\gamma_{s,i+1})}{W(\gamma_{r,i+1})}< \dfrac{f_{i}(L_{r,i})-c}{f_{i}(L_{s,i})-c}\quad (\text{from (\ref{con1}) and } L_{s,i}<L_{r,i})\nonumber
\end{align}
\begin{align}
& (f_i(L_{s,i})-c)W(\gamma_{s,i+1})<(f_i(L_{r,i})-c)W(\gamma_{r,i+1})\nonumber\\
& u_{s,i,max}<u_{r,i,max}\quad(\text{from } (\ref{eq:payi}))\nonumber
\end{align}
which contradicts the fact that $u_{s,i}\geq u_{r,i}$.

Note that, if $\gamma_{s,i}<\gamma_{r,i}$ and $\alpha_{s,i}\geq \alpha_{r,i}$, then we also obtain (\ref{n85a}) by simple algebraic manipulation, hence the proof of part (ii) is exactly similar to the proof of part (i).
\end{proof}
We use the following result in proving lemma~\ref{high2}.
\begin{lem}\label{obs:recurrence}
Suppose $u_{s,k,max}>u_{r,k,max}, \gamma_{s,k}<\gamma_{r,k}$. Let, $i=\min\{j\in \{k,\ldots,n\}: \alpha_{s,j}<\alpha_{r,j}\}$, then $\forall j$ such that $k\leq j\leq i$, we must have $u_{s,j,max}>u_{r,j,max}$.
\end{lem}
\begin{proof}
Suppose the statement is false. So, there exists a $j$ such that $k<j\leq i$, $u_{s,j,max}\leq u_{r,j,max}$\footnote{Note that the statement is true at state $k$, since $u_{s,k,max}>u_{r,k,max}$}. Since $u_{s,k}>u_{r,k}$, thus, there must exist a $k_1\in \{k,\ldots,j-1\}$, such that $u_{s,k_1,max}>u_{r,k_1,max}$ but $u_{s,k_1+1,max}\leq u_{r,k_1+1,max}$. Because otherwise we have $u_{s,j,max}>u_{r,j,max}$. 

Since $\gamma_{s,k}<\gamma_{r,k}$, thus from observation~\ref{identity3}
\begin{align}\label{n201}
\gamma_{s,k_1}+\sum_{j=k}^{k_1-1}\alpha_{s,j}<\gamma_{r,k_1}+\sum_{j=k}^{k_1-1}\alpha_{r,j}
\end{align}
By definition of $i$, $\alpha_{s,k_2}\geq \alpha_{r,k_2}$ for $k\leq k_2<i$, since $k_1<j$ and $j\leq i$, thus $\alpha_{s,k_2}\geq \alpha_{r,k_2}$ $\forall k_2\in \{k,\ldots,k_1\}$. Hence, from (\ref{n201}), we have $\gamma_{s,k_1}<\gamma_{r,k_1}$.

But $\alpha_{s,k_1}\geq \alpha_{r,k_1}$ and $u_{s,k_1,max}>u_{r,k_1,max}$, hence by lemma~\ref{high} we have $u_{s,k_1+1,max}>u_{r,k_1+1,max}$ which leads to a contradiction.
\end{proof}
\begin{lem}\label{high2}
Suppose, $u_{s,j,max}>u_{r,j,max}$, then there must exist a state $i\in \{1,\ldots,n\}$ such that $u_{s,i,max}>u_{r,i,max}$ but $\alpha_{s,i}<\alpha_{r,i}$.
\end{lem}
\begin{proof}
First we show that the statement is true when $u_{s,1,max}>u_{r,1,max}$ (case i) and then we show when $u_{s,1,max}\leq u_{r,1,max}$ (case ii); which completes the proof.

\textit{Case 1}: Suppose $u_{s,1,max}> u_{r,1,max}$. Since, $W(\cdot)$ is strictly decreasing, thus from (\ref{vbest}) we obtain $\gamma_{s,1}<\gamma_{r,1}$. Thus, from (\ref{ne+1}), there must exist $k=\min\{i\in\{1,\ldots,n\}: \alpha_{s,i}<\alpha_{r,i}\}$. By lemma~\ref{obs:recurrence}, $u_{s,j,max}>u_{r,j,max}$ $\forall j$ such that $1\leq j\leq k$. Since at $k$, $\alpha_{s,k}<\alpha_{r,k}$, $u_{s,k,max}>u_{r,k,max}$ thus, the statement is true for $k$. 

\textit{Case 2} Now, assume that $u_{s,1,max}\leq u_{r,1,max}$. Hence, it is obvious that $j\neq 1$. So, we must have $k=\min\{i\in\{1,\ldots,j-1\}: u_{s,i,max}\leq u_{r,i,max}, u_{s,i+1,max}>u_{r,i+1,max}\}$.  Note that if $\gamma_{s,k+1}<\gamma_{r,k+1}$, then from (\ref{ne+1}) there must exist $i=\min\{j:\{k+1,\ldots,n\}: \alpha_{s,j}<\alpha_{r,j}\}$. Since $u_{s,k+1,max}>u_{r,k+1,max}$, thus by lemma~\ref{obs:recurrence} at $i$, $u_{s,i}>u_{r,i}$ but $\alpha_{s,i}<\alpha_{r,i}$. Thus, the result is true for $i$ if we show that $\gamma_{s,k+1}<\gamma_{r,k+1}$. Now we complete the proof by showing that $\gamma_{s,k+1}<\gamma_{r,k+1}$.

Suppose that $\gamma_{s,k+1}\geq \gamma_{r,k+1}$. By definition of $k$, $u_{s,k}\leq u_{r,k}$, hence we obtain from (\ref{eq:payi})
\begin{align}\label{ns20}
(f_k(L_{s,k})-c)W(\gamma_{s,k+1})\leq (f_k(L_{r,k})-c)W(\gamma_{r,k+1})
\end{align}
Since $u_{s,k+1,max}>u_{r,k+1,max}$, thus from (\ref{eq:payi+1})
\begin{align}\label{ns21}
(f_{k+1}(L_{s,k})-c)W(\gamma_{s,k+1})
>(f_{k+1}(L_{r,k})-c)W(\gamma_{r,k+1})
\end{align}
Since$\gamma_{s,k+1}\geq \gamma_{r,k+1}$ and $W(\cdot)$ is strictly increasing, hence, $L_{r,k}< L_{s,k}$ from (\ref{ns21}). Thus from (\ref{ns21})
\begin{eqnarray}\label{ns22}
\dfrac{W(\gamma_{s,k+1})}{W(\gamma_{r,k+1})}
> \dfrac{f_{k}(L_{r,k})-c}{f_{k}(L_{s,k})-c}\\ (\text{from} (\ref{con1}) \text{as }c< f_k(L_{r,k}), L_{s,k}>L_{r,k})\nonumber
\end{eqnarray}
But (\ref{ns22}) contradicts (\ref{ns20}). Hence, $\gamma_{s,k+1}<\gamma_{r,k+1}$.
\end{proof}
\subsubsection{Properties of any symmetric NE strategy profile in a linear graph (lemmas~\ref{lm:eqpayoff} and ~\ref{equalnodedist})}\label{sec:propertieslinear}
We consider a linear graph (fig. ~\ref{fig:linear}) consisting of $M$ number of nodes. We use the properties of linear graph and a NE strategy profile to prove the results.  First, we state and prove Lemma~\ref{lownp}. Subsequently, we show that under an NE strategy profile the maximum expected payoff to a primary at a channel state at each node of $I_k, k\in \{1,2\}$ must be equal (Lemma~\ref{lm:eqpayoff}). Then, we show that under an NE strategy profile nodes of $I_k, k=\{1,2\}$ are selected with equal probability (Lemma~\ref{equalnodedist}). Finally, we show theorem~\ref{uniquelinear} using lemmas~\ref{lm:eqpayoff} and \ref{equalnodedist}.\\
In order to prove Lemma~\ref{lownp} we state and prove Observations~\ref{maximal},\ref{fact1},and \ref{necond}. 
\begin{obs}\label{maximal}
An NE independent set selection strategy profile only selects a maximal independent set with positive probability.
\end{obs}
\begin{proof}
Suppose not; so an independent set $I$ has been chosen with positive probability under an NE strategy profile, but it is not maximal which in turn implies that there exists a node $z$, such that $\bar{I}=I\cup \{z\}$ is an independent set. Since $\sum_{j=1}^{n}q_j=q<1$ (from (\ref{prob})), hence at node $z$, primary $1$ will attain at least a payoff of $(f_j(v)-c)W(q)>0$  for state $j$ when the primary selects the highest possible penalty $v$. Hence, a primary can attain strictly higher payoff  by choosing independent set $\bar{I}$ compared to $I$. Hence, the result follows.
\end{proof}
Observation~\ref{maximal} enables us to focus only on the maximal independent sets for an NE strategy profile.
\begin{obs}\label{fact1}
For a maximal independent set $I$-\\
(i) If $s\in I$, but $s+2\notin I$, then $s+3\in I$ for some $s\in V$. \\
(ii) If $s+2\in I$, but $s\notin I$, then $s-1\in I$ for some $s\in V$.
\end{obs}
\begin{proof}
{\em part (i)}: If it is not then $I\cup \{s+2\}$ is maximal, since $s+1\notin I$ (as $s\in I$ and $I$ is an independent set); which contradicts that $I$ is maximal.

{\em part (ii)}: If it is not then $I\cup\{s\}$ is an independent set since $s-1\notin I, s+1\notin I$ which contradicts that $I$ is maximal.
 \end{proof}
\begin{obs}\label{necond}
Consider an independent set $I$, such that $s\in I$, but $s+2\notin I$, for some $s\in \{1,\ldots, M-2\}$; NE independent selection strategy profile selects $I$ with positive probability, the following condition must be satisfied for $s\leq M-3$
\begin{align}\label{condn}
u_{s,j,max}\geq u_{s+1,j,max}, u_{s+3,j,max}\geq u_{s+2,j,max}
\quad \text{for } j\in \{1,\ldots,n\}
\end{align}
\end{obs}
\begin{proof}
Note that if $s=M-2$, then $I$ does not contain node $M, M-1$, hence $I$ is not maximal. Thus, an NE strategy profile can not select $I$ by Observation~\ref{maximal}. Hence, we must have $s\leq M-3$. 

If $u_{s,j,max}<u_{s+1,j,max}$, then we can replace node $s$ with node $s+1$ and we obtain an independent set $\bar{I}$ as $s+2\notin I$. But, we can get strictly higher payoff at the independent set $\bar{I}$, as all the nodes are same except $s$ and $u_{s,j,max}<u_{s+1,j,max}$. This contradicts that NE strategy profile selects $I$ with positive probability. 

Similarly if $u_{s+3,j,max}<u_{s+2,j,max}$ then we obtain an independent set by replacing node $s+3$ with $s+2$ in $I$ and can get a strictly higher payoff at that independent set.
\end{proof}
\begin{lem}\label{lownp}
i) If $u_{s,k,max}>u_{s+2,k,max}$, then $u_{1,i,max}>u_{3,i,max}$ for some $i\in \{1,\ldots,n\}$. \\
ii)If $u_{s,k,max}<u_{s+2,k,max}$, then $u_{M,i,max}>u_{M-2,i,max}$ for some $i\in \{1,\ldots,n\}$.
\end{lem}
\begin{proof}
We prove (i). The proof of (ii) will be similar to the proof of part (i) by symmetry.\\
 Since $u_{s,k,max}>u_{s+2,k,max}$, hence, from Lemma ~\ref{high2}, there exists $i\in \{1,\ldots,n\}$ such that $u_{s,i,max}>u_{s+2,i,max}$, but $\alpha_{s,i}<\alpha_{s+2,i}$. Hence, there must exist a maximal independent set $I$ such that $s\notin I$, but $s+2\in I$, which is chosen with positive probability in an NE strategy profile when the channel state is $i$. But, as $I$ is maximal, thus, $s-1\in I$ from Observation~\ref{fact1}. Also from Observation~\ref{necond}, we must have
\begin{align}\label{n80}
u_{s-1,i,max}\geq u_{s,i,max}, u_{s+2,i,max}\geq u_{s+1,i,max}
\end{align}
Since $u_{s,i,max}>u_{s+2,i,max}$, thus, from (\ref{n80}), we obtain
\begin{align}
u_{s-1,i,max}>u_{s+1,i,max}\nonumber
\end{align}
Hence, we obtain $u_{s-1,i,max}>u_{s+1,i,max}$  for some $i\in \{1,\ldots,n\}$  only using the fact that $u_{s,k,max}>u_{s+2,k,max}$.Thus, by recurrence on the index $s$ we obtain the result.
\end{proof}
Next Lemma characterizes that under an NE strategy profile maximum expected payoff must be equal at every node of $I_k, k\in \{1,2\}$.
\begin{lem}\label{lm:eqpayoff}
Under NE strategy profile, we must have $\forall j\in \{1,\ldots,n\}, \forall s,r \in I_k, k\in\{1,2\}$
\begin{align}\label{nsc}
u_{s,j,max}=u_{r,j,max}
\end{align}
\end{lem}
\begin{proof}
First, we prove $\alpha_{1,i}\geq \alpha_{3,i}$, $\alpha_{M,i}\geq \alpha_{M-2,i}, \forall i$.\\
We show that $\alpha_{1,i}\geq\alpha_{3,i}$ $\forall i$; by symmetry we get $\alpha_{M,i}\geq \alpha_{M-2,i}$.  Suppose, $\alpha_{1,j}<\alpha_{3,j}$ for some $j\in \{1,\ldots,n\}$. Then, there must exist a maximal independent set $I$ such that node $1\notin I$, but node $3\in I$; which is not possible (figure~\ref{fig:linear}). \\
Now, we are ready to prove the lemma.
Suppose the statement is false. So, we must have $u_{s,j,max}>u_{r,j,max}$ for some $j\in \{1,\ldots,n\}$ and $s,r\in I_k,k\in \{1,2\}$. We rule out $s<r$, by symmetry it follows that $s>r$; which completes the proof.

Since $u_{s,j,max}>u_{r,j,max}$, we must have some $a\in \{s,\ldots,r-2\}$, such that $u_{a,j,max}>u_{a+2,j,max}$.  Otherwise, $u_{s,j,max}\leq u_{r,j,max}$ since $r-s=2z$ for some positive integer $z$.  But, this entails that $u_{1,i,max}>u_{3,i,max}$ by Lemma~\ref{lownp} for some $i\in \{1,\ldots,n\}$, which in turn entails that $\alpha_{1,b}<\alpha_{3,b}$ for some $b\in \{1,\ldots,n\}$ (Lemma~\ref{high2}). But, we have already proved that $\alpha_{1,b}\geq \alpha_{3,b} \forall b\in\{1,\ldots,n\}$. Hence, the result follows.
\end{proof}

Next, lemma  shows that under an NE strategy profile nodes in $I_k, k\in \{1,2\}$ are selected with equal probability.
\begin{lem}\label{equalnodedist}
For state $z=1,\ldots,n$, $\alpha_{z,i}=\alpha_{z,j}$ where $i,j\in I_s, s\in \{1,2\}$.
\end{lem}
\begin{proof}
Let, $k$ be the lowest channel state, for which the statement is false. Thus, there must exist node $a,b\in I_s, s\in \{1,2\}$ such that, $\alpha_{a,k}>\alpha_{b,k}$, but $u_{a,k,max}= u_{b,k,max}$ (by Lemma~\ref{lm:eqpayoff}). First we rule out that $k=n$ (case i) and then $k<n$ (case ii).

\textit{Case 1}
Suppose, $k=n$.\\ By definition of $k$, $\alpha_{a,j}=\alpha_{b,j}$ $\forall j<k$, thus from Observation~\ref{identity3}, we have $\gamma_{a,1}>\gamma_{b,1}$. Since $W(\cdot)$ is strictly decreasing function, thus from (\ref{vbest}) we obtain $u_{a,1,max}<u_{b,1,max}$; which contradicts (\ref{nsc}).

\textit{Case 2} Now, suppose $k<n$.\\
Since $u_{a,1,max}=u_{b,1,max}$ by Lemma~\ref{lm:eqpayoff}, thus from (\ref{vbest}) $\gamma_{a,1}=\gamma_{b,1}$. Thus from Observation~\ref{identity3}
\begin{align}\label{n203}
\gamma_{a,k}+\sum_{j=1}^{k-1}\alpha_{a,j}=\gamma_{b,k}+\sum_{j=1}^{k-1}\alpha_{b,j}
\end{align}
By definition of $k$, we have $\alpha_{a,j}=\alpha_{b,j}$ $\forall j\leq k-1$. Hence, from (\ref{n203}), $\gamma_{a,k}=\gamma_{b,k}$. 
Since $\alpha_{a,k}>\alpha_{b,k}$, $\gamma_{a,k}=\gamma_{b,k}$, and $u_{a,k,max}=u_{b,k,max}$, hence  by Lemma ~\ref{high}, we obtain $u_{a,k+1,max}>u_{b,k+1,max}$. This expression again contradicts (\ref{nsc}). Hence $k\neq n$.
\end{proof}
From Lemma~\ref{equalnodedist}, we have $\alpha_{s,j}=\alpha_{r,j}=\bar{\alpha}_{k,j}$(let) where $s,r\in I_k$ $k\in \{1,2\}$ $j=1,\ldots,n$. From Lemma~\ref{lm:eqpayoff}, we have $u_{s,j,max}=u_{r,j,max}=\bar{u}_{k,j}$(let). 

\textit{Proof of Theorem~\ref{uniquelinear}}: First, we will show that for any NE strategy profile $\bar{\alpha}_{k,j}$ we must have $\sum_{k=1}^{2}\bar{\alpha}_{k,j}\geq 1$ $\forall j$. Then, we will show that if a primary chooses a maximal independent set other than $I_1$ and $I_2$ with positive probability, then we must have $\sum_{k=1}^{2}\bar{\alpha}_{k,j}<1$, which completes the proof.

Suppose $\sum_{k=1}^{2}\bar{\alpha}_{k,j}<1$ but it is an  NE  for some $j$. Since $I_1$ and $I_2$ constitute a partition of $V$, thus, the expected payoff that any primary at channel state $j$ will get is the following
\begin{align}\label{totalexpay}
& \sum_{s\in I_1}\bar{\alpha}_{1,j}u_{s,j}+\sum_{r\in I_2}\bar{\alpha}_{2,j}u_{r,j}\nonumber\\
& =\sum_{k=1}^{2}M_k\bar{\alpha}_{k,j}\bar{u}_{k,j}\quad (\text{since} |I_k|=M_k, u_{s,j}=\bar{u}_{k,j}, s\in I_k)
\end{align}
Consider the following unilateral deviation for primary 1 at channel state $j$: Primary 1 chooses $I_1$ with probability $\bar{\alpha}_{1,j}$ and $I_2$ with probability $1-\bar{\alpha}_{1,j}$. Since $\bar{u}_{k,j}$  remains the same, is strictly positive, and $1-\bar{\alpha}_{1,j}>\bar{\alpha}_{2,j}$ , hence primary 1 gets a strictly higher payoff following the above mentioned strategy by (\ref{totalexpay}). This contradicts that $\bar{\alpha}_{k,j}$ is an NE distribution.

Next, consider an NE strategy profile which selects a maximal independent set $I$, which has at least one node both from $I_1$ and $I_2$, with positive probability.  Hence, there exists a node $a$ such that $a, a+1\notin I$. Since $a$ and $a+1$ are adjacent, hence both can not appear in any independent set $\bar{I}\in \mathcal{I}$ otherwise $\bar{I}$ can not be  an independent set. Hence, by valid distribution property, we must have
 \begin{align}\label{validdistnode}
 \alpha_{a,j}+\alpha_{a+1,j}\leq 1
 \end{align}
 On the other hand for independent set $I$, both $a, a+1\notin I$. Since $I$ is chosen with positive probability, hence from (\ref{validdistnode})
 \begin{align}\label{n207}
 \alpha_{a,j}+\alpha_{a+1,j}<1
 \end{align}
 Without loss of generality, we can assume that $a\in I_1$, hence $a+1\in I_2$. Thus, $\alpha_{a,j}=\bar{\alpha}_{1,j}$ and $\alpha_{a+1,j}=\bar{\alpha}_{2,j}$. We have already shown that for any NE strategy profile we must have $\sum_{k=1}^{2}\bar{\alpha}_{k,j}=1$ which contradicts (\ref{n207}). Hence, a primary can not choose an independent set which contains at least one node from $I_1$ and $I_2$ under an NE strategy; since $I_1$ and $I_2$ constitute a partition of $V$; thus, only subsets of either $I_1$ or $I_2$ can be selected with positive probability. Since proper subsets of either $I_1$ or $I_2$ are not maximal, they can not be chosen with positive probability under an NE strategy by Observation~\ref{maximal}. Hence, the result follows.\qed


\subsection{Proof of Lemmas~\ref{thm:notsame} and \ref{thm:notane} (Section~\ref{sec:spsym})}\label{sec:appendix_lineargraph}

\subsubsection{Proof of Lemma~\ref{thm:notsame}}:In order to prove Lemma~\ref{thm:notsame} first, we describe an infinite  set of strategy profile $SP_{l,r,r_1}$. Subsequently, we show that every strategy profile in $SP_{l,r,r_1}$ is an NE.

Note that at a channel state vector $J$, a linear graph consists of disjoint smaller linear graphs (Fig.~\ref{fig:different_meanvalid}). First, we introduce some notations. Let $M_i$ be the linear graph which starts from node $i$ i.e. the channel is not available at node $i-1$ if $i>1$ (fig.~\ref{fig:different_meanvalid}), but it is available at node $i$. 

In $M_i$ the two maximal independent sets which partition the set of nodes in $M_i$ are: $I_{1,i}$ which contains the nodes numbered $i,i+2,\ldots$ and $I_{2,i}$ which contains the nodes numbered $i+1,i+3,\ldots$.  In figure ~\ref{fig:different_meanvalid}, $M_i$ and $M_j$ constitute two disconnected linear graphs. The cardinality of $M_i$ can be an even or odd number depending on the number of consecutive nodes where the channel is available starting from node $i$. To illustrate the cardinalities of $|M_i|$, consider the linear graph with $4$ nodes. Here, $|M_1|$ can take any value in $\{1,2,3,4\}$. When $|M_1|=1$, then the channel is available at node $1$ but not at node $2$. When $|M_1|=2$, then the channel is available at node $1$ and $2$, but the channel is not available at node $3$. Here $I_{1,1}=\{1\}$ and $I_{2,1}=\{2\}$. When $|M_1|=3$, then the channel is available at nodes $1,2, 3$, but the channel is unavailable at node $4$. Here, $I_{1,1}=\{1,3\}$ and $I_{2,1}=\{2\}$. When $|M_1|=4$, the channel of the primary is available at all nodes. Thus, $I_{1,1}$ coincides with $I_1$ and $I_{2,1}$ coincides with $I_2$ where $I_1=\{1,3\}$ and $I_2=\{2,4\}$.
\begin{figure}
\includegraphics[width=120mm,height=40mm]{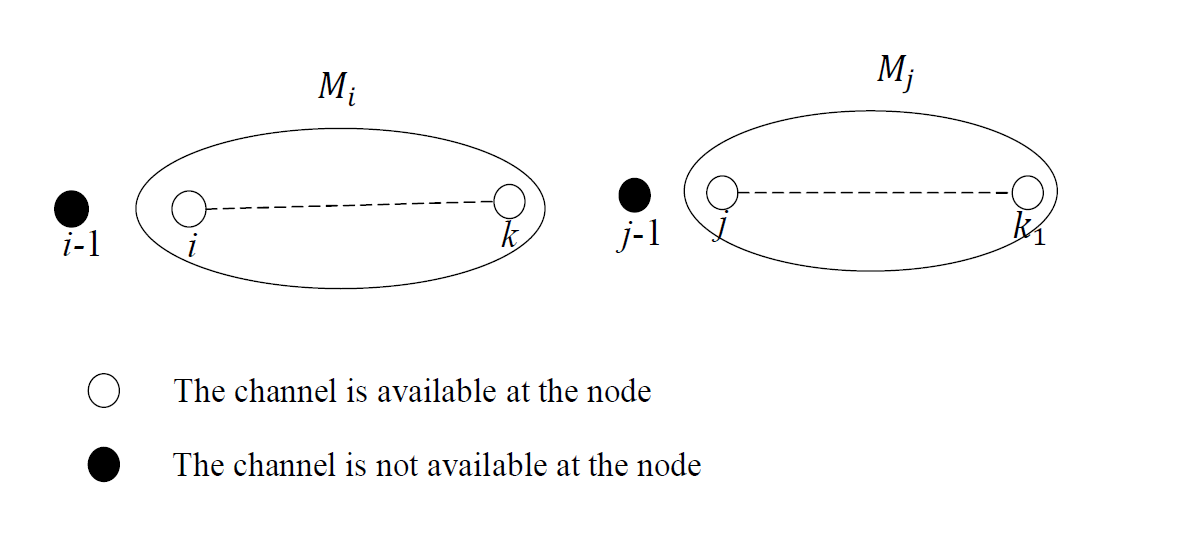}
\vspace{-0.5cm}
\caption{\small Figure shows the linear graph $M_i$ and $M_j$ with $|M_i|=k-i+1$ and $|M_j|=k_1-j+1$. $M_i$ and $M_j$ are disconnected. The maximal independent sets in $M_i$ is $I_{1,i}$ and $I_{2,i}$ where $I_{1,i}$ contains nodes numbered $i,i+2,\ldots$ and $I_{2,i}$ contains nodes numbered $i+1,i+3,\ldots$. }
\vspace{-0.3cm}
\label{fig:different_meanvalid}
\vspace{-0.3cm}
\end{figure}
Since for a given channel state vector $J$, the graph $G_J$  can be partitioned into linear graphs $M_i$   (fig.~\ref{fig:different_meanvalid}), thus, a primary only needs to select strategy for each such linear graph.  Thus,
\begin{lem}
Obtaining an NE strategy profile is equivalent to obtain an NE strategy at  each possible mean valid graph $M_i$, $i=1,\ldots,M$.
 \end{lem}

We use the following result to prove Theorem~\ref{thm:notsame}.
\begin{obs}\label{obs:maximum_ind}
When $|M_i|$ is odd, then the only maximum independent set is $I_{1,i}$, if $|M_i|$ is even, then both $I_{1,i}$ and $I_{2,i}$ are maximum independent sets of $M_i$.
\end{obs}
 Note that when $|M_i|$ is even, there can be other maximum independent sets apart from $I_{1,i}$ and $I_{2,i}$\footnote{For example, when $|M_1|=4$, then, the following are maximum independent sets, $\{1,3\}$, $\{2,4\}$, and $\{1,4\}$ where the first two independent sets belong to $I_{1,1}$ and $I_{2,1}$ respectively}.

%
%
%
%
%
%

Now, we consider a linear graph with $4$ nodes and the channel states are I.I.D. i.e. the channel is at state $1$ at a given node is w.p. $q_1=q=0.5$.
Let $t_{i,j}$ denote the probability of the event that  $|M_i|=j$.  
%
It is easy to show the following
\begin{align}\label{eq:1and2}
t_{1,3}=t_{2,3},\quad t_{4,1}=t_{1,1}
\end{align}
\begin{align}\label{eq:2and3}
t_{3,1}=t_{2,1},\quad  t_{3,2}=t_{1,2}.
\end{align}
\begin{align}\label{eq:12and14}
t_{2,2}=t_{1,4},\quad
t_{1,2}=2t_{1,4}, \quad t_{3,2}=2t_{1,4}
\end{align}
%
\begin{align}\label{eq:node_prob1}
t_{2,1}+t_{1,2}& =t_{1,1}, \quad t_{4,1}=t_{3,1}+t_{3,2}
\end{align}
%

Now, we describe an uncountable set of strategy profiles parameterized by parameters $r$ and $r_1$. 

Strategy profile $SP_{l,r,r_1}$: \textbf{If $|M_i|$ is odd, then $I_{1,i}$ will be selected w.p. $1$. If $|M_1|=2$, then $I_{1,1}$ will be selected w.p. $r$ and $I_{2,1}$ will be selected w.p. $1-r$. If $|M_1|=4$ i.e. when the channel is available at all nodes, then $I_1$ will be selected w.p. $r_1$ and $I_2$ will be selected w.p. $1-r_1$. If $|M_2|=2$, then $I_{1,2}$ will be selected w.p. $\dfrac{1}{2}$ and $I_{2,2}$ will be selected w.p. $\dfrac{1}{2}$. If $|M_3|=2$, then $I_{1,3}$ will be selected w.p. $0.75+r$ and $I_{2,3}$ will be selected w.p. $0.25-r$. }

where $r,r_1\geq 0$ are such that 
\begin{align}\label{eq:assum}
2r+r_1=0.75\quad \& r\leq 0.25
\end{align}

Since $0\leq r\leq 0.25$ and $0\leq r_1\leq 0.75$, thus, it is easy to discern that the strategy profile described in $SP_{l,r,r_1}$ constitutes a valid distribution. Note that there are uncountably infinite numbers of $r, r_1$ satisfying (\ref{eq:assum}). Thus, $SP_{l}$ gives rise an infinite number of strategies.

%

{\em Proof of Lemma~\ref{thm:notsame}}:  We first prove that for every $r, r_1$ which satisfy (\ref{eq:assum}) the strategy profile $SP_{l,r,r_1}$ is  an NE. 

  Towards this end we first show that  under the strategy profile $SP_{l,r,r_1}$ the channel is offered by a primary at every node with the same probability. 

Node selection probability of node $1$ i.e. $\alpha_{1}$ is
\begin{align}\label{eq:alpha1}
\alpha_1=t_{1,1}+t_{1,3}+t_{1,2}r+t_{1,4}r_1
\end{align}
and node selection probability of node $2$ is
\begin{align}\label{eq:alpha2}
\alpha_2=t_{1,2}(1-r)+t_{1,4}(1-r_1)+t_{2,1}+t_{2,3}+t_{2,2}/2
\end{align}
Node selection probability of node $3$ i.e. $\alpha_3$ is
\begin{align}\label{eq:alpha3}
t_{1,3}+t_{1,4}r_1+t_{2,2}/2+t_{3,2}(0.75+r)+t_{3,1}
\end{align}
Node selection probability of node $4$ is
\begin{align}\label{eq:alpha4}
\alpha_4=t_{1,4}*(1-r_1)+t_{3,2}*(0.25-r)+t_{4,1}+t_{2,3}
\end{align}
Note that 
\begin{align}\label{eq:mixed_equality}
2t_{1,4}r_1+2t_{1,2}r&=t_{1,4}(2r_1+4r) \quad (\text{from } (\ref{eq:12and14}))\nonumber\\
& =t_{1,4}*3/2\quad (\text{from } (\ref{eq:assum}))
\end{align}
Thus, from (\ref{eq:alpha1}) and (\ref{eq:alpha2}), we obtain that $\alpha_{1}-\alpha_{2}$ is equal to
\begin{align}\label{eq:diff_alpha}
t_{1,1}-t_{1,2}-t_{2,1}+t_{1,3}-t_{2,3}+2t_{1,2}r+ 2t_{1,4}r_1-t_{1,4}-t_{2,2}/2
\end{align}
Note that $t_{1,3}=t_{2,3}$ by (\ref{eq:1and2}) and $t_{1,1}=t_{2,1}+t_{1,2}$ by (\ref{eq:node_prob1}). 
Since $t_{2,2}=t_{1,4}$ from (\ref{eq:12and14}), thus it readily follows from (\ref{eq:diff_alpha}) that  $\alpha_{1}=\alpha_{2}$. 
From (\ref{eq:alpha3}) and (\ref{eq:alpha2}) we obtain $\alpha_2-\alpha_3$ is equal to
\begin{align}\label{eq:diff_alpha23}
t_{2,1}-t_{3,1}+t_{2,3}-t_{1,3}+t_{1,4}+ t_{1,2}-0.75t_{3,2}-rt_{3,2}-rt_{1,2}-2t_{1,4}r_1
\end{align}
Note that $t_{1,3}=t_{2,3}$ by (\ref{eq:1and2}), $t_{1,2}=t_{3,2}$ and $t_{2,1}=t_{3,1}$ by (\ref{eq:2and3}). 
Thus, from (\ref{eq:diff_alpha23}) 
\begin{align}\label{eq:diff_alpha23b}
\alpha_2-\alpha_3=t_{1,4}+t_{1,2}/4-2rt_{1,2}-2r_1t_{1,4}
\end{align}
Also note from (\ref{eq:12and14}) that $t_{1,4}=t_{1,2}/2$. Thus, $t_{1,4}+t_{1,2}/4=t_{1,4}*3/2$. Thus, from (\ref{eq:mixed_equality}) and (\ref{eq:diff_alpha23b})  it readily follows that $\alpha_{3}=\alpha_{2}$. 
From (\ref{eq:alpha3}) and (\ref{eq:alpha4}) we obtain that $\alpha_3-\alpha_4$ is equal to
\begin{align}\label{eq:diff_alpha34}
 t_{1,3}-t_{2,3}+t_{3,1}+t_{3,2}-t_{4,1}+t_{2,2}/2 -t_{1,4}-t_{3,2}/2+2t_{1,4}r_1+2rt_{3,2}
\end{align}
Note that $t_{1,3}=t_{2,3}$ by (\ref{eq:1and2}). Also note that $t_{4,1}=t_{3,1}+t_{3,2}$ by (\ref{eq:node_prob1}).  Thus, from 
\begin{align}\label{eq:diff_alpha34b}
\alpha_3-\alpha_4=t_{2,2}/2-t_{1,4}-t_{3,2}/2+2t_{1,4}r_1+2rt_{3,2}
\end{align}
Since $t_{3,2}/2=t_{1,4}$ by (\ref{eq:12and14}), thus, we obtain $t_{1,4}+t_{3,2}/2-t_{2,2}/2=t_{1,4}*3/2$.  Since $t_{1,2}=t_{3,2}$ by (\ref{eq:2and3}), thus, from (\ref{eq:mixed_equality}) and (\ref{eq:diff_alpha34b}) it readily follows that $\alpha_4=\alpha_3$. Hence, we obtain $\alpha_{1}=\alpha_2=\alpha_3=\alpha_4$. 

Since node selection probability is identical across the nodes, thus, when all the other primaries select a strategy profile in the set $SP_{l,r,r_1}$, then, the maximum payoff of primary $1$ at a node $i$ is  $(f_1(v)-c)(1-w(\alpha_{i}))$ by Theorem~\ref{singlelocation} and (\ref{eq:ex_paya1s}) and this is obtained for any penalty in the interval $[L_1,v]$ with $\alpha_i$ in place of $q_1$ by Lemma~\ref{separation-s1}. Hence, the maximum attainable expected payoff of primary $1$ at each location is the same since $\alpha_i$\rq{}s are identical. 

Now, we show that primary $1$ does not have any incentive to deviate from a strategy profile for fixed $r,r_1$ when other primaries also select that strategy profile.

  When $|M_{i}|$ is odd, then by Observation~\ref{obs:maximum_ind} $I_{1,i}$ is the only maximum independent set in $M_i$. Since the maximum attainable expected payoff for primary $1$ is the same at every node, thus,  the  expected payoff at $I_{1,i}$ is the highest for primary $1$ when all the other primaries select a strategy profile in $SP_{l,r,r_1}$. Hence,  primary $1$ does not have any incentive to deviate from $SP_{l,r,r_1}$ when $|M_i|$ is odd since under $SP_{l,r,r_1}$ primary $1$ selects $I_{1,i}$ w.p. $1$ when $|M_i|$ is odd.
  
  When $|M_i|$ is even, then $|I_{1,i}|=|I_{2,i}|$. By Observation~\ref{obs:maximum_ind}   both $I_{1,i}$, $I_{2,i}$ are the maximum independent sets. Since the maximum attainable expected payoff is the same at each node, thus, any strategy profile which randomizes between $I_{1,i}$ and $I_{2,i}$ gives the highest expected payoff to primary $1$. Thus, primary $1$ does not have any incentive to deviate from $SP_{l,r,r_1}$ when $|M_{i}|$ is even since under $SP_{l,r,r_1}$ primary $1$ only randomizes between $I_{1,i}$ and $I_{2,i}$. 
  
Though we only consider primary $1$ since the every strategy in $SP_{l,r,r_1}$ is symmetric, hence, no primary will have any incentive to deviate unilaterally from the strategy profile for a fixed $r,r_1$. Thus, we show that every $r,r_1$ which satisfy (\ref{eq:assum}), the strategy set in $SP_{l,r,r_1}$ is an NE.

Since there are uncountable number of $r,r_1$s which satisfy (\ref{eq:assum}), hence there are multiple NEs in this setting.\qed 

\subsubsection{Proof of Lemma~\ref{thm:notane}} We show that $SP_{sym}$ is not a NE strategy profile in the above linear graph with $4$ nodes where the channel is in state $1$ at a given location w.p $0.5$ irrespective of the channel states at other locations. In order to prove the result we use some of the results which we derived in the previous section to prove Lemma~\ref{thm:notsame}.

  First, we point out the how $SP_{sym}$ (described in Section~\ref{sec:spsym}) is different from  the class of strategy profile $SP_{l,r,r_1}$ (described in the previous section). Then, we show that $SP_{sym}$ is not an NE in this setting.

Since $I_{1,i}$ and $I_{2,i}$ are the only maximum independent sets of $M_i$ when $|M_i|=2$ by Observation~\ref{obs:maximum_ind}, thus,  according to $SP_{sym}$ (Section~\ref{sec:spsym}), when $|M_i|=2$, $I_{1,i}$ and $I_{2,i}$ must be selected w.p $\dfrac{1}{2}$.  Note that in $SP_{l,r,r_1}$ when $|M_1|=2$, $I_{1,1}$ is selected w.p. $r$, and $I_{2,1}$ is selected w.p. $1-r$ where $r\leq 0.25$.  Thus, $I_{1,1}$ and $I_{2,1}$ are not selected with equal probabilities even though they are of same sizes. Thus, $SP_{sym}$ does not belong to $SP_{l,r,r_1}$. Now we show that $SP_{sym}$ can not be an NE.
 
$SP_{sym}$ puts equal weight on every maximum independent sets. When $|M_1|=4$, then under $SP_{sym}$, each of the maximum independent sets $\{1,3\}, \{2,4\}$ and $\{1,4\}$ with equal probabilities. Hence, the channel is offered at node $1$ w.p. $t_{1,4}*2/3$ when $|M_1|=4$. Thus, under $SP_{sym}$, the node selection probability is
\begin{align}\label{eq:alpha-n1}
\alpha_{1}=t_{1,1}+t_{1,3}+t_{1,2}/2+2t_{1,4}/3
\end{align}
and node selection probability of node $2$ is
\begin{align}\label{eq:alpha-n2}
\alpha_2=t_{1,2}/2+t_{1,4}/3+t_{2,1}+t_{2,3}+t_{2,2}/2
\end{align}
Now, we show that $\alpha_{1}>\alpha_{2}$. Since $t_{1,3}=t_{2,3}$ (by (\ref{eq:1and2})) and $2t_{1,4}/3>t_{1,4}/3$, thus, we are left to show that $t_{1,1}>t_{2,1}+t_{2,2}/2$. By simple algebraic calculation for $q_1=q_0=q=0.5$, we have $t_{1,1}=0.25, t_{2,1}=1/8, t_{2,2}=1/16$. Hence $t_{1,1}>t_{2,1}+t_{2,2}/2$. Thus, $\alpha_1>\alpha_2$.

Thus, by the single location pricing strategy the maximum  expected payoff attained by a primary at node $1$ is $(f_1(v)-c)(1-w(\alpha_1))$ (from (\ref{eq:ex_paya1s}))  and the expected payoff attained by a primary at node $2$ is $(f_1(v)-c)(1-w(\alpha_2))$ ( by (\ref{eq:ex_paya1s})) when the other primaries select $SP_{sym}$. Since $\alpha_1>\alpha_2$ and $w(\cdot)$ is strictly increasing, thus, the expected payoff at node $2$ is strictly higher compared to the node $1$. Thus, when $|M_1|=2$, if a primary selects node $2$ w.p. $1$, then it would attain strictly higher payoff compared to the strategy   $SP_{sym}$ where a primary selects node $2$ w.p. $\dfrac{1}{2}$ and node $1$ w.p. $\dfrac{1}{2}$ when $|M_1|=2$. Hence, a primary has an incentive to deviate unilaterally from its strategy profile. Hence, $SP_{sym}$ {\em is not an NE}.\qed

\subsection{Markov Random Field}\label{sec:mrf}
\subsubsection{Background}
A Markov random field is a graphical model which represents the joint probability distributions of random variables having Markov property. It is represented by an undirected graph $H=(V,E)$ in which the nodes $V$ represents the random variables. The edges $E$ encodes the dependencies among the random variables in the following way: if $N(A)$ is the set of neighbors of $A$, then in a Markov random field\cite{markov},
\begin{align}
A \perp \text{other random variables}|N(A)\nonumber
\end{align}
Figure~\ref{fig:cycle_mrf} provides a cyclic Markov random field. Here, $A \perp C|B,D$. 
\begin{figure*}
\begin{minipage}{0.49\linewidth}
\begin{center}
\includegraphics[width=90mm,height=40mm]{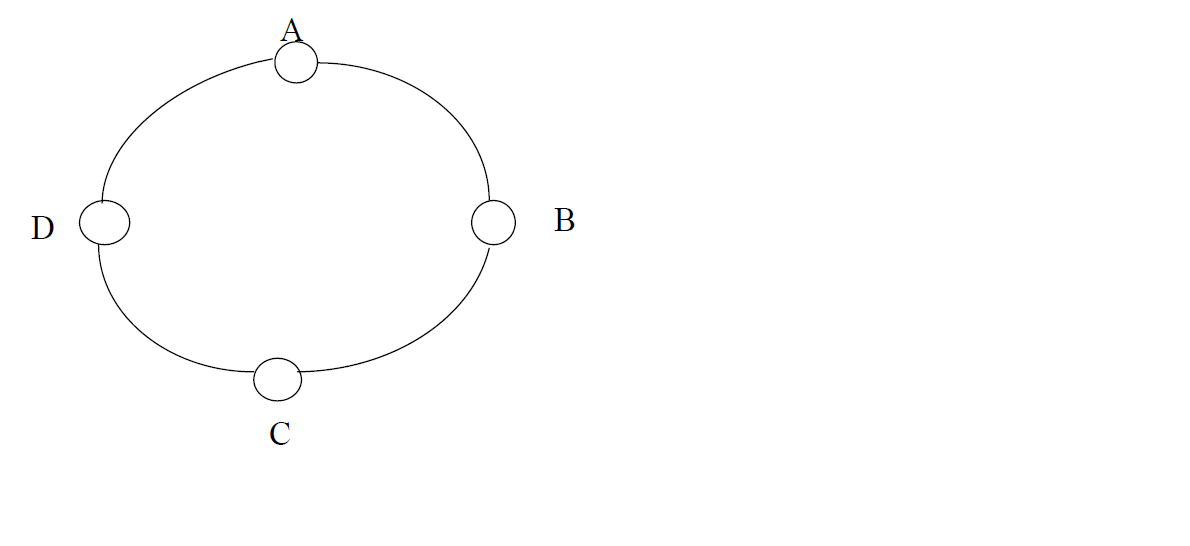}
\caption{\small A cyclic Markov Random field.}
\label{fig:cycle_mrf}
\vspace{-0.4cm}
\end{center}
\end{minipage}\hfill
\begin{minipage}{0.49\linewidth}
\begin{center}
\includegraphics[width=90mm,height=40mm]{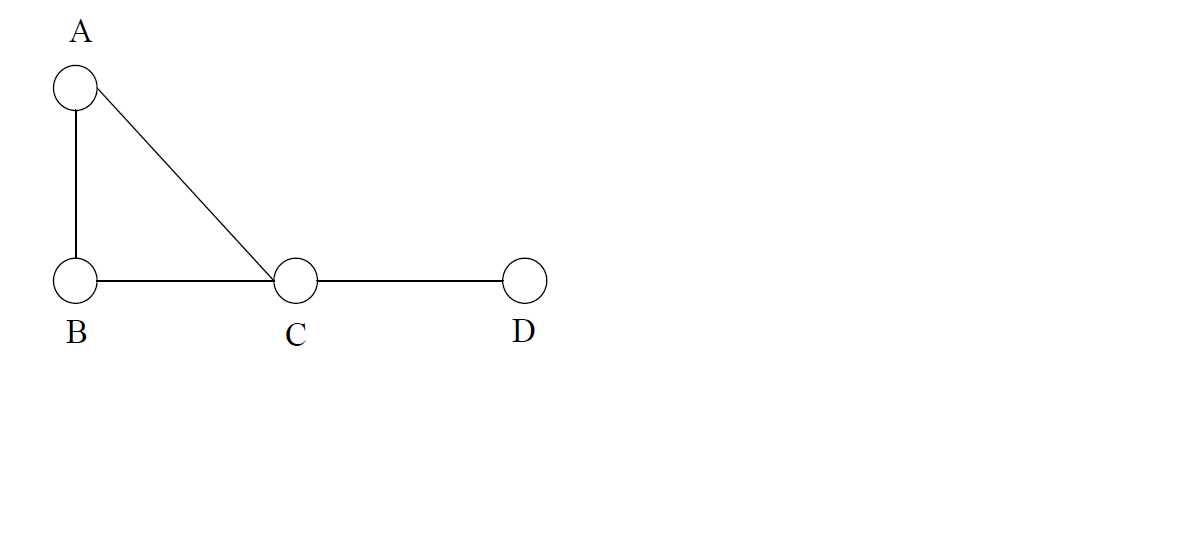}
\caption{\small The figure shows a Markov Random Field. Here the Maximal cliques are $(ABC, CD)$.}
\label{fig:asymmetric_mrf}
\vspace{-0.4cm}
\end{center}
\end{minipage}
\end{figure*}
The channel states in a conflict graph are random variables whose values are either $0$ or $1$. Since the channel states at adjacent locations are likely to be correlated, we model the correlation amongst the adjacent locations in the conflict graph using the Markov Random field where the nodes in the Markov random field represent the channel states of the corresponding nodes of conflict graph $G$.  Figure~\ref{fig:cycle_mrf} represents a Markov random field when the conflict graph is a cyclic graph with $4$ nodes and the values of the random variables $A, B, C, D \in \{0,1\}$ represent the channel states at nodes $A ,B, C, D$ of conflict graph $G$ respectively. 

We now discuss the joint probability distribution in the Markov random field. Markov random fields provide a compact representation of the joint probability distribution in terms of product of {\em potential functions}. Potential functions are defined on the set of maximal cliques, $\mathcal{C}$, in the graphical representation of the Markov random field $H$. A potential function $\zeta_C(\cdot)$ represent the values of the random variable of the maximal clique $C\in \mathcal{C}$. For example, in figure~\ref{fig:cycle_mrf} the set $\{AB\}$ is a maximal clique, thus, $\zeta_{AB}(a,b)$ denote the value of the potential function when the random variables $A=a$ and $B=b$,  $a,b\in \{0,1\}$. Note that $\zeta_{C}(\cdot)$ is defined on the vector $\mathbf{c}$ which represents the values of the random variables represented by nodes in the clique $C$. Formally, the probability of the channel state $J$ is given by:
\begin{align}
q_J=\dfrac{1}{Z}\prod_{C\in \mathcal{C}}\zeta_C(c_J)
\end{align}
where $Z$ is a normalization factor and $c_J$ denote the channel states in clique $C$ when the overall channel state vector is $J$.

For example, in figure~\ref{fig:cycle_mrf} the set of maximal cliques $\mathcal{C}$ is ${AB, BC, CD, DA}$.   The joint probability distribution is given by
\begin{equation}
P_{A,B,C,D}(a,b,c,d)=\dfrac{1}{Z}\zeta_{AB}(a,b)\zeta_{BC}(b,c)\zeta_{CD}(c,d)\zeta_{DA}(d,a)\nonumber
\end{equation}
Since $A,B, C, D$ only take values in $\{0,1\}$, we can represent $\zeta$ as a matrix where $\zeta_{AB}(a,b)$ denote the value of the $(a,b)$th position of the matrix. For example,  $\zeta$ can be the following:
\begin{equation}\label{eq:identicalpotential}
\zeta_{AB}=\zeta_{BC}=\zeta_{CD}=\zeta_{DA}=\begin{bmatrix}
0.8 & 0.2\\ 0.2 & 1
\end{bmatrix}
\end{equation}
In Figure~\ref{fig:asymmetric_mrf}, the maximal cliques are $(ABC, CD)$. Hence, the joint probability distributions are
\begin{align}
P_{A,B,C,D}(a,b,c,d)=\dfrac{1}{Z}\zeta_{ABC}(a,b,c)\zeta_{CD}(c,d)
\end{align}

\begin{defn}
The Markov random field representation of  random variables is symmetric if i) the maximal cliques are of identical sizes and ii) suppose $\mathbf{c}_1$ corresponds to the channel state vector of maximal clique $C_1$ and $\mathbf{c}_2$ corresponds to the channel state vector of maximal clique $C_2$, then
\begin{align}\label{eq:potential_max}
\zeta_{C_1}(\mathbf{c_1})=\zeta_{C_2}(\mathbf{c_2})
\end{align}
for every $\mathbf{c_1}$ and $\mathbf{c}_2$ such that $\mathbf{c}_1$ and $\mathbf{c}_2$ contain the same number of $1$s (and thus, the same number of $0$s since $C_1, C_2$ are of same sizes).
\end{defn}
 (\ref{eq:identicalpotential}) provides an example of potential functions which are symmetric and identical. But potential functions in Fig.~\ref{fig:asymmetric_mrf} can not be symmetric since the sizes of the maximal cliques are different.
%
 
 Now, we are ready to provide an example which satisfies Assumption~\ref{assum:prob}. 
\subsubsection{Result}
\begin{lem}\label{lm:mrf_sameprob}
 The probability distributions on the channel states satisfy Assumption~\ref{assum:prob} if\\
 i) The channel states constitute a Markov random field,\\
 ii) The graphical representation of the Markov random field $H$ is the same as the node symmetric graph $G$,\\
 iii) The Markov random field relation is symmetric\footnote{In a node symmetric graph, the maximal cliques are of the same size}, and\\
 iv) There are fixed integers $r_1, r_2\ldots$ such that  every clique containing $j\geq 1$ number of nodes is a subset of   identical ($r_j$) number of maximal cliques in $G$. 
 \end{lem}
First, it is easy to discern that the condition (iv) is satisfied by a large class of node symmetric conflict graphs including cyclic graph, infinite linear graph (Fig.~\ref{fig:linear_infinity}), infinite square graph(Fig.~\ref{fig:random_square}), infinite grid graph (Fig.~\ref{fig:random_grid}), infinite triangular graph (Fig.~\ref{fig:random_hexacell}). For example, in the infinite triangular graph (Fig.~\ref{fig:clique_hexacell}), a clique containing $3$ nodes is a maximal clique and hence, it is a part of only $1$ maximal clique; any clique containing $2$ nodes is a subset of $2$ maximal cliques; a single node is a part of $6$ maximal cliques. 
\begin{figure}
\includegraphics[width=120mm,height=40mm]{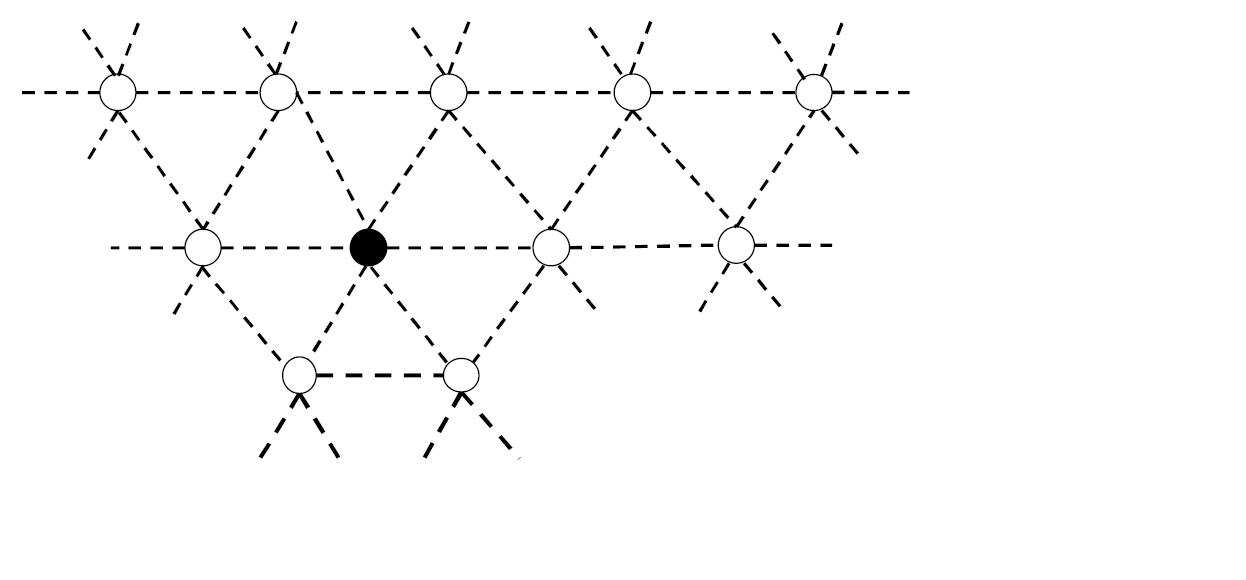}
\caption{\small Infinite triangular graphs: the black colored node is a part of $6$ maximal cliques. Any clique containing two nodes is a part of $2$ maximal clique. A clique containing $3$ nodes is a maximal clique in this graph.}
\label{fig:clique_hexacell}
\vspace{-0.5cm}
\end{figure}

In order to prove the above lemma, we first show the following for any node symmetric graph $G$ which satisfies condition (iv):
\begin{obs}\label{obs:equality_inclique}
Let $n_j$ be the number of maximal cliques in $G$ which contains exactly $j$ nodes of $G_{J}$, then there are exactly $n_j$ number of maximal cliques in $G$ which contains exactly $j$ nodes of $G_{K}$, when $G_{K}$ is isomorphic to $G_{J}$. 
\end{obs}
\begin{proof}
Let $G_{J}$ and $G_{K}$ be isomorphic (Definition~\ref{defn:isomorphism}) to each other,  where $G_{J}$ and $G_{K}$  are the conflict graphs corresponding to the channel state vectors $J$ and  $K$ respectively. Since $G_{J}$ and $G_{K}$ are isomorphic, there is an isomorphic function $F_1(\cdot)$ between the nodes of $G_{J}$ and $G_{K}$. 

Suppose that there is a maximal clique $C$ which contains $j$ nodes of $G_J$. Thus, this set of $j$ nodes is a subset of a maximal clique.  The isomorphic function $F_1(\cdot)$ maps those $j$ nodes into $j$ different nodes of $G_{K}$. Also note that since these $j$ nodes of $G_{J}$ belong to a clique in the original graph $G$, hence they are adjacent to each other, since $F_1(\cdot)$ is isomorphic, thus, the mapped $j$ nodes must also be adjacent to each other, hence that set of  mapped $j$ nodes is also a subset of a maximal clique  in the original graph $G$. 

Suppose the statement in the result is false. Thus, there must exist a set $V_j$ of $j$ nodes of $G_{J}$ which is a subset of $r_1$ number of maximal cliques in the original graph, however the mapped set of nodes $F_1(V_j)$ of $G_{K}$ is only a subset of  $r_2$ number of maximal cliques in the original graph where $r_2<r_1$. Thus, this violates the condition (iv). Hence, $r_2\geq r_1$. By symmetry, we can also show that the situation where  $r_2>r_1$ can not arise. Hence, the result follows.
\end{proof}
Now, we are ready to prove Lemma~\ref{lm:mrf_sameprob}.
\begin{proof}
Let $G_{J}$ and $G_{K}$ be isomorphic to each other,  where $G_{J}$ and $G_{K}$  are the conflict graphs corresponding to the channel state vectors $J$ and  $K$ respectively. We have to show that $q_J=q_K$.
 Let $c_{J}$ be the channel state vector at the nodes of $C$ when the channel state vector is $J$. 
 
Now, at channel state vector $J$, the potential function value at maximal clique is $\zeta_{C}(c_{J})$. Thus, the channel state vector $q_J$and $q_K$ are given by
 \begin{align}\label{eq:mrf_same}
 q_J=\prod_{C\in \mathcal{C}}\dfrac{1}{Z}\zeta_C(c_J)\nonumber\\
 q_K=\prod_{C\in \mathcal{C}}\dfrac{1}{Z}\zeta_{C}(c_{K})
 \end{align}  
By Observation~\ref{obs:equality_inclique}, the number of maximal cliques which contain $j$ number of nodes of $G_{J}$ and $G_{K}$ are identical. Note that at channel state vectors $J$ and $K$, the nodes where the channel state is $1$ are only the nodes of $G_{J}$ and $G_{K}$ respectively. Hence, the number of maximal cliques which contain exactly $j$ number of $1$s are the same (and thus, the number of $0$s since in the node symmetric graph, the size of maximal cliques are the same) in the channel state vectors $J$ and $K$.   Hence, $q_K=q_J$ from (\ref{eq:potential_max}) and (\ref{eq:mrf_same}).
\end{proof}

\begin{thebibliography}{10}
\bibitem{ciss}
A.~Ghosh and S.~Sarkar, ``Quality sensitive price competition in spectrum
  oligopoly over multiple locations,'' in \emph{Information Sciences and
  Systems (CISS), 2014 48th Annual Conference on}.\hskip 1em plus 0.5em minus
  0.4em\relax IEEE.

\bibitem{graph}
D.~West, \emph{{Introduction to Graph Theory}}.\hskip 1em plus 0.5em minus
  0.4em\relax 2nd Edition, Prentice Hall, 2000.

\bibitem{isit}
A.~Ghosh and S.~Sarkar, ``{Quality sensitive price competition in spectrum
  oligopoly},'' in \emph{Proceedings of IEEE International Symposium on
  Information Theory (ISIT)}, 2013, pp. 2770--2774.

\bibitem{arnob_ton}

A.~Ghosh, and S.~Sarkar ``{ Quality-Sensitive Price Competition in Secondary Market Spectrum
  Oligopoly-Single Location Game},'' \emph{to appear in Networking, IEEE/ACM
  Transactions on}, 2015. [Online]. Available: \url{10.1109/TNET.2015.2440422}

\bibitem{gauravjsac}
G.~Kasbekar and S.~Sarkar, ``{Spectrum Pricing Game with Bandwidth Uncertainty
  and Spatial Reuse in Cognitive Radio Network},'' \emph{IEEE Journal on
  Special Areas in Communication}, vol.~30, no.~1, pp. 153--164, 2012.

\bibitem{nodetransitive}
G.~Sabidussi, ``Vertex-transitive graphs,'' \emph{Monatshefte für Mathematik},
  vol.~68, no.~5, pp. 426--438, 1964.

\bibitem{lattice}
E.~Weisstein, ``Lattice graphs,'' 26th June, 2015. [Online]. Available:
  \url{http://mathworld.wolfram.com/LatticeGraph.html}


\bibitem{mwg}
A.~Mas~Colell, M.~Whinston, and J.~Green, \emph{{ Microeconomic Theory}}.\hskip
  1em plus 0.5em minus 0.4em\relax Oxford University Press, 1995.

\bibitem{Osborne}
M.~Osborne and C.~Pitchik, ``{Price Competition in a Capacity Constrained
  Duopoly},'' \emph{Journal On Economic Theory}, vol.~38, no.~2, pp. 238--260,
  1986.

\bibitem{Kreps}
D.~Kreps and J.~Scheinkman, ``{Quantity Precommitment and Bertrand Competition
  yield Cournot Outcomes},'' \emph{Bell Journal of Economics}, vol.~14, pp.
  326--337, Autumn, 1983.

\bibitem{Ileri}
O.~Ileri, D.~Samardzija, T.~Sizer, and N.~B. Mandayam, ``{Demand Responsive
  Pricing and Competitive Spectrum Allocation via a Spectrum Policy Server},''
  in \emph{IEEE Proceedings of DySpan}, 2005, pp. 194--202.

\bibitem{Mailespectrumsharing}
P.~Maille and B.~Tuffin, ``{Price War with Partial Spectrum Sharing for
  Competitive Wireless Service Provider},'' in \emph{Proceeding Of IEEE
  GLOBECOM}, 2009, pp. 1--6.

\bibitem{Mailepricecompslotted}
------, ``{Analysis Of Price Competition in a Slotted Resource Allocation
  Game},'' in \emph{Proceeding of 27th IEEE INFOCOM}, 2008, pp. 888--896.

\bibitem{Xing}
Y.~Xing, R.~Chandramouli, and C.~Cordeiro, ``{Price Dynamics in Competitive
  Agile Spectrum Access Markets},'' \emph{IEEE Journal on Special Areas in
  Communication}, vol.~25, no.~3, pp. 613--621, 2008.

\bibitem{Niyatospeccrn}
D.~Niyato and E.~Hossain, ``{Competitive Pricing for Spectrum Sharing in
  Cognitive Radio Network: Dynamic Games, Inefficiency of Nash Equilibrium, and
  Collusion},'' \emph{IEEE Journal on Special Areas in Communication}, vol.~26,
  no.~1, pp. 192--202, 2008.

\bibitem{Niyatomultipleseller}
D.~Niyato, E.~Hossain, and Z.~Han, ``{Dynamics of Multiple Seller and Multiple
  Buyer Spectrum Trading in Cognitive Radio Network: A Game theoretic Modeling
  approach},'' \emph{IEEE Transaction on Mobile Computing}, vol.~8, no.~8, pp.
  1009--1022, 2009.

\bibitem{Zhou}
X.~Zhou and H.~Zheng, ``{TRUST: A General Framework for Truthful Double
  Spectrum Auctions},'' in \emph{In the Proceedings of Infocom}, April 2009.

\bibitem{kavurmacioglu}
A.~M. Kavurmacioglu, E. and D.~Starobinski, ``Competition in secondary spectrum
  markets: Price war or market sharing?'' in \emph{IEEE Proceedings of DYSPAN}.

\bibitem{yitan}
Y.~Tan, S.~Sengupta, and K.~Subbalakshmi, ``Competitive spectrum trading in
  dynamic spectrum access markets: A price war,'' in \emph{Global
  Telecommunications Conference (GLOBECOM 2010), 2010 IEEE}, Dec 2010, pp.
  1--5.

\bibitem{duan}
L.~Duan, J.~Huang, and B.~Shou, ``Competition with dynamic spectrum leasing,''
  in \emph{New Frontiers in Dynamic Spectrum, 2010 IEEE Symposium on}, April
  2010, pp. 1--11.

\bibitem{zhang}
F.~Zhang and W.~Zhang, ``Competition between wireless service providers:
  Pricing, equilibrium and efficiency,'' in \emph{Modeling Optimization in
  Mobile, Ad Hoc Wireless Networks (WiOpt), 2013 11th International Symposium
  on}, May 2013, pp. 208--215.

\bibitem{jia}
J.~Jia and Q.~Zhang, ``Bandwidth and price competitions of wireless service
  providers in two-stage spectrum market,'' in \emph{Communications, 2008. ICC
  '08. IEEE International Conference on}, May 2008, pp. 4953--4957.

\bibitem{yang}
L.~Yang, H.~Kim, J.~Zhang, M.~Chiang, and C.~wei Tan, ``Pricing-based spectrum
  access control in cognitive radio networks with random access,'' in
  \emph{INFOCOM, 2011 Proceedings IEEE}, April 2011, pp. 2228--2236.

\bibitem{sengupta}
S.~Sengupta and M.~Chatterjee, ``An economic framework for dynamic spectrum
  access and service pricing,'' \emph{Networking, IEEE/ACM Transactions on},
  vol.~17, no.~4, pp. 1200--1213, Aug 2009.

\bibitem{kim}
H.~Kim, J.~Choi, and K.~Shin, ``Wi-fi 2.0: Price and quality competitions of
  duopoly cognitive radio wireless service providers with time-varying spectrum
  availability,'' in \emph{INFOCOM, 2011 Proceedings IEEE}, April 2011, pp.
  2453--2461.

\bibitem{cover}
T.~M. Cover and J.~A. Thomas, \emph{Elements of Information Theory},
  2nd~ed.\hskip 1em plus 0.5em minus 0.4em\relax Wiley, 2006.

\bibitem{proakis}
J.~G. Proakis and M.~Salehi, \emph{Digital Communications}.\hskip 1em plus
  0.5em minus 0.4em\relax McGraw Hill, 5th edition, 2008.

\bibitem{gupta}
P.~Gupta and P.~Kumar, ``The capacity of wireless networks,'' \emph{Information
  Theory, IEEE Transactions on}, vol.~46, no.~2, pp. 388--404, Mar 2000.

\bibitem{shroff}
C.~Joo, X.~Lin, and N.~Shroff, ``Understanding the capacity region of the
  greedy maximal scheduling algorithm in multi-hop wireless networks,'' in
  \emph{INFOCOM 2008. The 27th Conference on Computer Communications. IEEE},
  April 2008. (Best paper award).

\bibitem{srikant}
A.~Gupta, X.~Lin, and R.~Srikant, ``Low-complexity distributed scheduling
  algorithms for wireless networks,'' \emph{Networking, IEEE/ACM Transactions
  on}.

\bibitem{matula}
M.~I. D.~W. Matula, ``A 6-regular torus graph family with applications to
  cellular and interconnection networks,'' \emph{Graph Algorithms and
  Applications 3}, pp. 373--404, 2004.

\bibitem{galton-watson}
T.~E. Harris, ``The theory of branching processes,'' 1964.

\bibitem{archivedreport-part2}
A.~Ghosh and S.~Sarkar, ``Quality sensitive price competition in spectrum
  oligopoly: Part ii,'' \emph{CoRR}, vol. abs/1404.6766, 2014.

\bibitem{markov}
S.~Z. Li, \emph{Markov random field modeling in computer vision}.\hskip 1em
  plus 0.5em minus 0.4em\relax Springer Science \& Business Media, 2012.

\end{thebibliography}
\end{document}